\documentclass[
nobibnotes,
reprint,
superscriptaddress,
 amsmath,amssymb, 
 aps,
 prl,
floatfix
]{revtex4-2}

\makeatletter
\AtBeginDocument{\let\LS@rot\@undefined}
\makeatother

\usepackage[a4paper, margin=1in]{geometry}
\usepackage{xr}
\usepackage{graphicx} 
\usepackage{subcaption}
\usepackage{amsmath,amsthm,amssymb,amsfonts,mathtools,mathrsfs}
\usepackage{physics}
\usepackage{color}
\usepackage{appendix}
\usepackage{xcolor}
\usepackage{hyperref}
\hypersetup{breaklinks=true, colorlinks=true, linkcolor={magenta!100!black}, citecolor={blue!90!black}, urlcolor={blue!90!black}}
\usepackage{tabularx}
\usepackage{cleveref}
\usepackage{booktabs}
\usepackage{tikz}
\usepackage{comment}
\usepackage[normalem]{ulem}
\usetikzlibrary{positioning, arrows.meta, calc, backgrounds,               decorations.pathreplacing}
\usepackage{tabularx}
\usepackage{makecell}
\usepackage{pdfpages} 
\usepackage{pgffor}  
\usepackage{amsthm}
\usepackage{cleveref}
\usepackage[T1]{fontenc}

\usepackage{ragged2e}   
\DeclareCaptionFormat{capjust}{\justifying#1#2#3}
\usepackage{xurl}
\usepackage{soul}

\newtheorem{theorem}{Theorem}

\numberwithin{equation}{section}

\crefname{figure}{fig.}{figs.}
\Crefname{figure}{Fig.}{Figs.}
\crefname{subfigure}{fig.}{figs.}
\Crefname{subfigure}{Fig.}{Figs.}

\newcommand{\vars}[1]{\mathcal{#1}}

\newcommand{\hash}{f}

\newcommand{\brac}[1]{\left( #1 \right)}

\newcommand{\errth}{e_{\mathsf{th}}}
\newcommand{\timeth}{t_{\mathsf{th}}}

\newcommand{\etab}{\eta_{\mathsf{b}}}
\newcommand{\pmisb}{p_{\mathsf{mis,b}}}
\newcommand{\epscom}{\varepsilon_{\mathrm{com}}}
\newcommand{\epssou}{\varepsilon_{\mathrm{sou}}}
\newcommand{\advstrat}{\vars{S}}
\newcommand{\advset}{\vars{A}}
\newcommand{\measstrat}{\vars{M}}
\newcommand{\measstratprat}{\measstrat_{\mathrm{prac}}}

\newcommand{\pmisth}{p_{\mathsf{mis,th}}}
\newcommand{\pmis}{p_{\mathsf{mis}}}
\newcommand{\pdc}{p_{\mathsf{dc}}}
\newcommand{\respA}{z_0}
\newcommand{\respB}{z_1}

\newcommand{\Ndet}{N_{\mathrm{det}}}
\newcommand{\Nerr}{N_{\mathrm{err}}}
\newcommand{\Ndc}{N_{\mathrm{dc}}}
\newcommand{\Nmis}{N_{\mathrm{mis}}}
\newcommand{\verA}{V_0^Q}
\newcommand{\verB}{V_1^C}
\newcommand{\prover}{P}
\newcommand{\ploc}{X_P}

\usepackage{environ}
\usepackage{xcolor}
\usepackage{tcolorbox}
\tcbuselibrary{skins}

\newcounter{box}
\crefname{box}{Protocol}{Protocols}

\NewEnviron{protocol}[2][t]{
  \refstepcounter{box}
  \begin{table}[#1]
    \centering
    \begin{tcolorbox}[
      width=0.95\columnwidth,
      title={\textbf{Protocol \thebox:}\quad #2},   
      colbacktitle=cyan!60!black,              
      coltitle=white,
      colback=white!99!teal,   
      colframe=teal,
      boxrule=0.8pt,
      arc=2pt,
      left=6pt, right=6pt,
      top=6pt, bottom=6pt,
      fonttitle=\bfseries,
      halign=flush left,    
      sharp corners=south,
      enhanced,
      segmentation hidden
    ]
      \BODY
    \end{tcolorbox}
  \end{table}
}

\date{\today}

\makeatletter

\renewcommand{\emph}[1]{{#1}}

\begin{document}
\author{Wen Yu Kon}
\thanks{These authors contributed equally to this work.}
\affiliation{Global Technology Applied Research, JPMorganChase, New York, NY 10017, USA}

\author{Niccol\`{o} Bigagli}
\thanks{These authors contributed equally to this work.}
\affiliation{Qunnect Inc., 141 Flushing Ave, Ste 1110, Brooklyn, NY 11205-1005}

\author{Andrew Conrad}
\thanks{These authors contributed equally to this work.}
\affiliation{Global Technology Applied Research, JPMorganChase, New York, NY 10017, USA}

\author{Taylor Shields}
\affiliation{Global Technology Applied Research, JPMorganChase, New York, NY 10017, USA}

\author{Fatih Kaleoglu}
\affiliation{Global Technology Applied Research, JPMorganChase, New York, NY 10017, USA}

\author{Ignatius William Primaatmaja}
\affiliation{Global Technology Applied Research, JPMorganChase, New York, NY 10017, USA}

\author{Alexander Craddock}
\affiliation{Qunnect Inc., 141 Flushing Ave, Ste 1110, Brooklyn, NY 11205-1005}

\author{RJ Pisani}
\affiliation{Qunnect Inc., 141 Flushing Ave, Ste 1110, Brooklyn, NY 11205-1005}

\author{Mael Flament}
\affiliation{Qunnect Inc., 141 Flushing Ave, Ste 1110, Brooklyn, NY 11205-1005}

\author{Jude Seeber}
\affiliation{Global Technology Applied Research, JPMorganChase, New York, NY 10017, USA}

\author{Omar Amer}
\affiliation{Global Technology Applied Research, JPMorganChase, New York, NY 10017, USA}

\author{Charles Lim}
\affiliation{Global Technology Applied Research, JPMorganChase, New York, NY 10017, USA}

\author{Xinhua Ling}
\affiliation{Global Technology Applied Research, JPMorganChase, New York, NY 10017, USA}

\author{Rob Otter}
\affiliation{Global Technology Applied Research, JPMorganChase, New York, NY 10017, USA}

\author{Kaushik Chakraborty}
\email{kaushik.chakraborty@jpmchase.com}
\affiliation{Global Technology Applied Research, JPMorganChase, New York, NY 10017, USA}

\author{Mehdi Namazi}
\email{mehdi@quconn.com}
\affiliation{Qunnect Inc., 141 Flushing Ave, Ste 1110, Brooklyn, NY 11205-1005}

\title{Loss-Tolerant Quantum Position Verification for Metropolitan Area Networks}
\begin{abstract}
A \textit{spacetime seal}, a cryptographic guarantee that a digital event has occurred at an approved location and time, can augment a digital signature with location attestation for legal, financial, and regulatory use cases. In adversarial settings any purely classical realization of such a seal can be spoofed~\cite{Chandran2009,Buhrman2014}.
Quantum position verification (QPV) offers a physics-based solution, exploiting the no-cloning theorem and the no-signaling principle to certify a party's spacetime coordinates \cite{Buhrman2014}.    
While the feasibility of QPV has been recently shown via entanglement-~\cite{Kavuri2026} and coherent light-based protocols~\cite{fan2026relativistic}, achieving loss tolerance for these schemes~\cite{ABB23,Llorenc2023} substantially increases implementation complexity at metropolitan scales.
Here, we introduce and experimentally demonstrate a loss-tolerant QPV (LT-QPV) protocol whose security is independent of channel loss.
We prove finite-size security against quantum polynomial-time entangled adversaries in the quantum random oracle model instantiated with cryptographically secure hash functions. 
Implemented entirely with commercial off-the-shelf components, our system certifies position within 22 minutes of net data collection time against a restricted adversary, with a clear path to real-time certification ($<1$s) with upgraded hardware. 
Our architecture, requiring only a single quantum verifier node alongside classical infrastructure, is naturally compatible with metropolitan-area quantum networks, establishing the foundations for scalable, physics-backed spacetime certification as a deployable service.

\end{abstract}

\maketitle

\vspace{2cm}

\newpage

\section{Introduction}

Digital signatures are as ubiquitous in our daily lives as in regulatory settings. While they greatly simplify transactions and agreements, they can bear no verifiable attestation of the location where a signing took place. For a growing set of applications, this missing attribute is of paramount importance. 
In particular, many digital systems operate in environments where the jurisdiction in which data processing, access, or authorization occurs may be relevant.
Data-residency rules and related governance considerations arise in connection with regulatory frameworks such as the EU General Data Protection Regulation (GDPR) \cite{GDPR}, financial-market regulations such as MiFID II \cite{eu_mifid2_2014_65, eu_rts25_2017_574, eu_mifir_600_2014}, export-control frameworks such as ITAR \cite{us_itar_22cfr_120_130}, and healthcare statutes such as HIPAA \cite{us_hipaa_45cfr_160_164}, all bind digital activity to specific jurisdictions. Similarly, sovereign-cloud mandates, sanctions compliance, electronic notarization, and trusted confidential computing (TCC) \cite{TCC_Whitepaper} require a verifiable, tamper-evident attestation of where a digital event has occurred.
Developing trustworthy protocols for location verification is therefore an important technical problem for modern digital service providers.     

Within current communication frameworks, time and location metadata accompanying digital signatures can be subverted. For instance, an attacker can redirect traffic via DNS or API-gateway compromise, interfere with sensor paths, mask IP addresses through proxies and VPNs, spoof MAC addresses, or manipulate system clocks. More fundamentally, no purely classical position-verification protocol can be secure against colluding adversaries \cite{Chandran2009,Buhrman2014}. These limitations motivate the need to expand current techniques to include beyond-classical capabilities. 

A potential approach to provide secure location information is quantum position verification (QPV), a class of protocols that can be used to generate a \textit{spacetime seal}, a certified, publicly verifiable stamp of the time and location of a signing event whose security rests on physical principles rather than on the integrity of any particular network. A powerful technique, QPV can go beyond spacetime certification and support location-based access control, position-based key exchange \cite{Buhrman2014,Kon2025} and treaty and inspection verification.

QPV exploits the no-cloning theorem and the no-signaling principle to certify that a party occupies a claimed position in spacetime. In a QPV protocol, two or more cooperating verifiers challenge a prover to demonstrate that it occupies a claimed position by requiring it to respond correctly to quantum and classical inputs within a narrow timing window dictated by the speed of light. While unconditional security against adversaries possessing unbounded entanglement is known to be
impossible~\cite{Beigi2011, Buhrman2014}, meaningful security can be achieved under physically motivated resource constraints such as limited
entanglement~\cite{Bluhm2022,Llorenc2023,Kavuri2026,Llorenc2025}, computational hardness assumptions~\cite{Asadi2025}, or polynomial resources (e.g. query access) in the quantum random oracle model~\cite{Unruh14}. A central practical obstacle, however, has been channel loss: photons transmitted over realistic quantum channels between the prover and verifier are frequently lost.
To manage channel loss, existing protocols have to either exponentially increase the number of measurement bases \cite{Llorenc2023} or employ complex Bell-state measurements~\cite{ABB23}.

Experimental efforts to realize QPV have accelerated in recent years,
with three concurrent demonstrations highlighting both the progress and
the remaining challenges. Kanneworff~\textit{et
al.}~\cite{Kanneworff_2025} reported a partial implementation of a single-photon SWAP-based
demonstration, providing a proof of concept but lacking
security against entangled adversaries. Kavuri~\textit{et
al.}~\cite{Kavuri2026} proposed and demonstrated a device-independent protocol that
removes trust assumptions on the measurement devices at the cost
of stringent experimental requirements. 
Their protocol limits the loss tolerance to $1/3$, matching that of Bell tests.
Finally, Fan-Yuan \textit{et al.}~\cite{fan2026relativistic} demonstrated a prepare-and-measure scheme using coherent states, offering some practical simplicity but relying on hollow-core fibers. 
These specialized fiber optics are more expensive and less widely deployed than standard telecom fiber, presenting practical challenges for large-scale deployment. 
Furthermore, at the metropolitan regime, partial channel loss-tolerance is achieved with an exponentially growing measurement complexity while full channel loss-tolerance requires a complex Bell state measurement using the scheme in Ref.~\cite{ABB23}, which presents further experimental challenges.

Here, we propose and experimentally demonstrate a loss tolerant QPV protocol (LT-QPV) whose security is independent of channel loss by shifting the duty of preparing and distributing entangled qubits to the prover.
Assuming the verifier possesses a trusted measurement device, only local losses at the prover location limit security and any no-detection rounds at the verifier can be safely discarded, providing intrinsic tolerance to channel loss. 
While this does not allow for arbitrary scaling of prover-verifier separation in practical settings where standard optical fibers are used for quantum communication, it does provide better scaling compared to simple prepare-and-measure QPV that is sufficient for implementation at metropolitan distances.
Indeed, most practical applications of QPV do not require quantum channels spanning hundreds or thousands of kilometers. 
We prove the finite-size security of LT-QPV against quantum polynomial-time (QPT) adversaries in the quantum random oracle model (QROM) \cite{BDFLSZ11}, covering both independently and identically distributed (i.i.d.) attacks, where an adversary performs the same attack strategy independently for every round (also known as individual attacks)~\cite{Gisin2002,Scarani2009}, and general coherent attacks under sequential repetition.
We further validate the
protocol with a short-distance experimental demonstration, certifying the location of an untrusted prover at a distance of $\sim$60m with our LT-QPV protocol using only standard, commercially-available hardware, as shown in Fig. \ref{fig:1}\textbf{a}, against an adversary performing i.i.d. attacks with limited query access.

\begin{figure*}[!ht]
    \centering
    \includegraphics[width=1\textwidth]{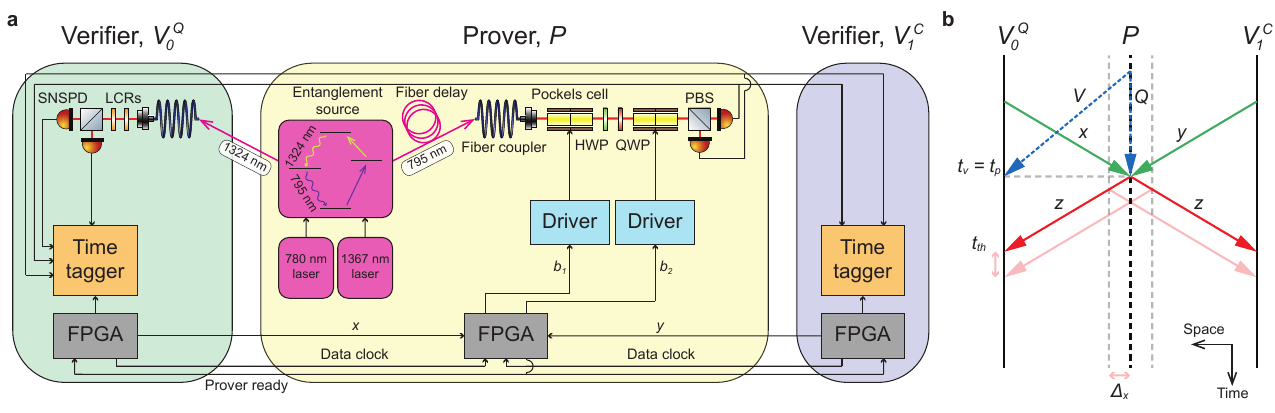}
    \caption{\textbf{a}, \textbf{Sketch of the experimental setup.} Our system consists of an untrusted prover $P$, a quantum verifier $V_0^Q$, and a classical verifier $V_1^C$. The prover generates an entangled pair and transmits 1324 nm (Telecom) photon to $V_0^Q$, then later measures the 795 nm (NIR) photon based on the challenge commands $x,y$ provided by $V_0^Q$ and $V_1^C$. The measurement result $z$ is transmitted back to both verifiers, which check if the result matches the measurement of the Telecom photon and determine the positional uncertainty based on the transmission times of $x,y$ and the receive times of $z$. \textbf{b}, \textbf{Spacetime diagram of the LT-QPV protocol}. 
    The quantum verifier $V_0^Q$, the quantum prover $P$, and the classical verifier $V_1^C$ are positioned on a line segment. The blue dotted arrow represents the quantum channel from $P$ to $V_0^Q$. The green and red lines represent classical channels. While the figure presented shows $t_V=t_P$, this is in general not necessary. More details on the choices of $t_V$ and $t_P$ can be found in Supplementary Section III B 1.}
    \label{fig:1}
\end{figure*}

The quantum communication requirements of the LT-QPV framework are naturally compatible with emerging quantum network architectures in which entanglement is distributed and measured at network nodes, and the classical requirements are plausibly realizable at metropolitan scales utilizing cellular towers or other standard wireless communication technologies. 
Together, these characteristics open a practical route towards realizing secure spacetime seals as an early application of metro-scale quantum networks and provide a basis for future developments.

\section{The Protocol}

\newcommand{\etath}{\eta_{\mathsf{th}}}
\newcommand{\regP}{Q}
\newcommand{\mismatch}{\theta}
\newcommand{\error}{\delta}
\newcommand{\transmission}{\eta}

We consider a setting where two verifiers, $\verA$ and $\verB$, cooperate to certify the location of a prover $\prover$ at a claimed position between them, as illustrated in Fig.~\ref{fig:1}. One verifier, $\verA$, is quantum-enabled, i.e. they are able to measure the state of a qubit, while the other, $\verB$, is fully classical.
The LT-QPV protocol is comprised of repeated QPV rounds during which classical and quantum information is exchanged between the three nodes.
In each round, the prover generates a Bell state
$|\Phi^{+}\rangle = \frac{1}{\sqrt{2}}(|00\rangle + |11\rangle)$, where 0 and 1 refer to any orthogonal qubit states (e.g. horizontal and vertical photon polarization, H and V). One qubit (subsystem $Q$) is retained by the prover, the other (subsystem $V$) is transmitted through a quantum channel
to verifier $\verA$. Concurrently, verifiers $\verA$ and $\verB$ issue
randomly chosen classical challenges $x$ and $y$, timed to arrive at
the prover's claimed position simultaneously. A shared public hash function $f$, such as SHA3~\cite{SHA3}, maps the pair of challenges $(x, y)$ to a measurement basis
$\alpha = f(x, y)$, selected from $m$ possible bases. 
Verifier $\verA$ holds a fully characterized measurement device that performs projective measurements in basis $\alpha$ with basis vectors
\begin{equation}
  |\psi^\alpha_0\rangle = \cos\!\left(\frac{\theta_\alpha}{2}\right)
  |0\rangle + \sin\!\left(\frac{\theta_\alpha}{2}\right)
  e^{i\phi_\alpha}|1\rangle,
  \label{eq:basis}
\end{equation}
and its orthogonal state $|\psi^\alpha_1\rangle$, where
$\theta_\alpha$ and $\phi_\alpha$ are the polar and azimuthal angles
on the Bloch sphere for basis $\alpha$. 
Valid rounds occur when the  measurement apparatus at both the prover and $\verA$ is set to the same basis, and $\verA$ detects a qubit. 
If the prover also detects a qubit, they forward the measurement result, $z$, to both verifiers. 
After $N$ valid rounds, the verifiers note the number of rounds when the prover had a detection, $\Ndet$, the number of error rounds in which the prover’s and verifier’s outcomes disagree, $\Nerr$, the number of double-click events at $\verA$, $\Ndc$, and the number of rounds the prover's responses to each verifier are mismatched (i.e. the $z$ values received by $\verA$ and $\verB$ are different), $\Nmis$.
The prover's location is certified by checking that: (1) the prover efficiency (the ratio of valid rounds that are detected by the prover) meets a minimum threshold $\eta=\Ndet/N \geq \eta_{\mathrm{th}}$, (2) the effective quantum bit error rate (QBER) is below a maximum threshold, $\text{QBER} = (N_{\mathrm{err}} + N_{\mathrm{dc}})/N_{\mathrm{det}} \leq e_{\mathrm{th}}$, and (3) $\Nmis$ is below a maximum threshold, $\Nmis/N\leq\pmisth$. The protocol is detailed in the methods and Supplementary Information, and a spacetime diagram of a single
round is shown in Fig.~\ref{fig:1}\textbf{b}.
We note that if the classical and quantum communication speeds, $v_{cl}$ and $v_V$, are not at the speed of light, $c$, and the prover's processing time is non-zero, the prover $P$ can only be certified to be within a region.
We also note that we do not simply abort the protocol when mismatch is observed since in practice, mismatched caused by timing fluctuations where one value may arrive after the cut-off time or other imperfections such as noise in the classical channel may occur.

LT-QPV can be viewed as a source-independent version of a $f$-measure QPV protocol~\cite{Bluhm2022}.
The classical communications in both protocols are identical, and the main difference is in the quantum state preparation and transfer.
$f$-measure QPV has the verifier prepare qubit $Q$ to be sent to the prover, while we allow the prover to prepare and entangled state and send $V$ to the verifier.
As such, the channel loss for $f$-measure protocols is experienced by the qubit received by the prover, while the channel loss in LT-QPV is experienced by the verifier's qubit $V$ instead.
Switching to source-independence also requires a trusted measurement device instead of a trusted source at the verifier, and an additional entanglement source at the prover.

The central theoretical result of our protocol is that LT-QPV security is independent of channel loss between the prover and verifier by having the prover prepare and broadcast the entangled quantum state, assuming that the verifier holds a trusted measurement device.
We note that this does not imply arbitrary separation between the prover and verifier $\verA$ since the distance may indirectly affect the prover-side loss.
The loss-tolerance follows from the causal structure of the protocol: under the QROM, the measurement basis $\alpha$ is first revealed at the prover location $P$ at time $t_P$.
Since the detection event at $\verA$ is space-like separated from the basis revelation (see Fig.~\ref{fig:1}\textbf{b}), the verifier's detection decision is independent of $\alpha$ and no-detection rounds can be safely discarded. 
While our protocol is intrinsically tolerant to channel loss, local losses on the prover side (including source heralding efficiency) can degrade performance, with the tolerated prover efficiency scaling as $\eta\sim\frac{1}{m}$ in the asymptotic regime ($N\rightarrow\infty$) for zero QBER and uniform basis choice. 

We prove finite-size security against entangled QPT adversaries (adversaries with memory, computational power, entanglement and query access to $f$ bounded by a polynomial in the security parameter $n$) in the QROM model, establishing that the protocol is $\varepsilon_{\mathrm{sou}}$-sound and $\varepsilon_{\mathrm{com}}$-complete for both i.i.d.\ and general attacks.
The theorem and proof structure is presented in the Methods and the complete proof is provided in the Supplementary Information.
This indicates that an honest prover fails verification with probability at most $\varepsilon_{\mathrm{com}}$, while any QPT adversary passes verification with probability at most $\varepsilon_{\mathrm{sou}}$.
Notably, the security proof requires no assumptions on the prover's source, rendering the protocol source-independent.
Multi-photon attacks on $\verA$'s measurement device are handled
through universal squashing~\cite{FCL11}, which treats double-click
events as errors, and the monogamy of entanglement~\cite{Tomamichel2013} ensures that any
adversary attempting to correlate with both verifiers simultaneously incurs QBER exceeding the verification threshold.

\section{The Experiment}

After sketching the contours and the strengths of LT-QPV, we experimentally demonstrate its feasibility. 
The setup is comprised of quantum and classical layers connecting the prover to the two verifiers. The quantum layer provides and measures the entangled particles used in the protocol;
the classical layer handles the selection of measurement bases, the calculation of the uncertainty areas 
around the verifiers, and the collection, compiling, and analysis of measurement results. Fig. \ref{fig:1}\textbf{a} shows a diagram of the experimental setup. 

The qubits used in our protocol are polarization-entangled photon pairs generated by a warm rubidium entanglement source \cite{craddock2024high} placed at the prover's location. Each pair is constituted by a photon at 795 nm (near infrared, or NIR) and one at 1324 nm (telecom), emitted in the 
$\ket{\Phi^+}$ Bell state. The NIR photon is kept at the prover's location, delayed by a 42 m optical fiber delay line, and then coupled into a measurement station capable of measuring in 16 unique bases, a number that ensures no security issues for $\eta \gtrsim 6.25 \% $ for zero QBER and uniform basis choice. The telecom photon is coupled into a fiber deployed between the prover and $\verA$, after which it encounters its own measurement station. 

Low-latency communication and control between the verifier and prover nodes are realized using field-programmable gate arrays (FPGAs). For each QPV round, the verifiers' FPGAs transmit random four-bit challenge signals $x,y\in \{0,1\}^4 $, which are timed to arrive simultaneously at the prover's claimed location. The prover then selects a measurement basis $\alpha=f(x,y)$ out of the 16 possible measurement basis using a $16 \times 16$ randomly-generated lookup table stored in the FPGA's memory. 
To enact fast basis selection for each measurement we employ two back-to-back high-voltage Pockels cells, each selecting one of 4 unique polarization rotations and capable of nanosecond switching. QPV rounds are repeated at a 200 kHz rate, limited by the cells' drivers and clocked by the FPGAs. The probabilistic nature of the entanglement source makes it impossible to predict whether a classical round will lead to zero, one, or more than one valid rounds. More details on how the quantum and the classical layers cooperate to enact the QPV protocol are presented in the Methods.

\begin{figure}[ht]
    \centering
    \includegraphics[width = \columnwidth]{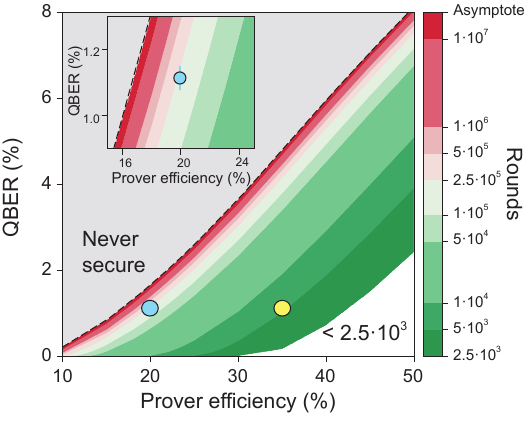}
    \caption{Performance of the QPV system required to achieve a security of $\varepsilon_{com}=10^{-2}$ and $\varepsilon_{sou}=10^{-2}$ against i.i.d. adversary assuming $\varepsilon_{ROM}=0$ with expected mismatch probability $\pmis=5.49\times10^{-6}$ and double click probability $\pdc=1.45\times10^{-4}$. The dotted line represents the performance requirements in the asymptotic regime and the solid lines represents the performance requirements for $N$ valid rounds. The blue dot denotes the performance of our QPV system implementing the LT-QPV protocol, which requires $\geq 2.55\times 10^5$ rounds to achieve security. The yellow dot demonstrates how the number of required rounds can be reduced by increasing the prover efficiency.}
    \label{fig:Operating_Curve}
\end{figure}

Prior to the QPV demonstration, we execute a full scale calibration round to characterize hardware performance and calculate the parameters under which the protocol can successfully certify the prover's location. 
Note that an independent calibration run is necessary to avoid violating the assumption that the performance thresholds and number of valid rounds $N$ are chosen prior to the protocol run, as explained in the Supplementary Information.
The calibration round is performed by running the full setup over 47 minutes of net data collection time, equal to 184 minutes of continuous operation. This overhead is due to our data acquisition software and can be shortened as explained in the Methods. 
During the calibration, we measure an overall QBER of $1.17\pm0.05\%$ and a prover's efficiency of $\eta = 20.10\pm0.07\%$.
As shown in Fig.~\ref{fig:Operating_Curve}, under an i.i.d. adversary model with limited query access and with security parameters $\varepsilon_{com}=10^{-2}$ and $\varepsilon_{sou}=10^{-2}$, a total of $N\geq 2.55\times 10^5$ rounds is required to certify the prover's position, and the observed performance must stay within the acceptance thresholds $\etath=19.86\%$, $\errth=1.31\%$, and $\pmisth=2.9\times 10^{-5}$. 

With these requirements established, we carry out the full LT-QPV demonstration. In under 22 minutes of net data collection time we complete $N=2.6\times 10^5$ rounds and measure an effective QBER of $1.23\pm0.05\%$, a prover's efficiency of $\eta= 19.96\pm0.08\%$, and only one mismatched response, $\pmis=3.8\times 10^{-6}$, clearing the required acceptance checks and certifying the prover's location. To highlight our setup's stability, we point out that our QPV data collection is performed three days after the calibration dataset with minimal intervention between runs. 

The positional uncertainty of the prover is determined by calculating the distance that light in free space could have traveled during the time elapsed between the transmission of the challenge commands and the reception of the result of the quantum measurement. In our QPV instance, the radii for verifiers $\verA$ and $\verB$ are found to be $r_0=58.1\text{ m}$, and $r_1=64.6\text{ m}$, respectively. We calculate the quantum-advantage ratio compared to classical ranging as in \hbox{\cite{Kavuri2026}} which is 2.18 (1-D), 2.43 (2-D), and 2.65 (3-D) for our LT-QPV demonstration. 
As shown in Fig. \ref{fig:Positional_Uncertainty}, our positional uncertainty region does not include the verifiers' locations. In contrast, the equivalent region of any classical position verification protocol always contains the line segment connecting the two verifiers even in the ideal case without any latency. In practice, network latency leads the positional uncertainty region to circumscribe an ellipsoid with foci at the verifiers \cite{Kavuri2026}.
This implies that our protocol securely verifies the position of the prover more precisely than any classical protocol, providing an indisputable quantum advantage \cite{Chandran2009}. 

\begin{figure*}[ht]
    \centering
    \includegraphics[width = 1 \textwidth]{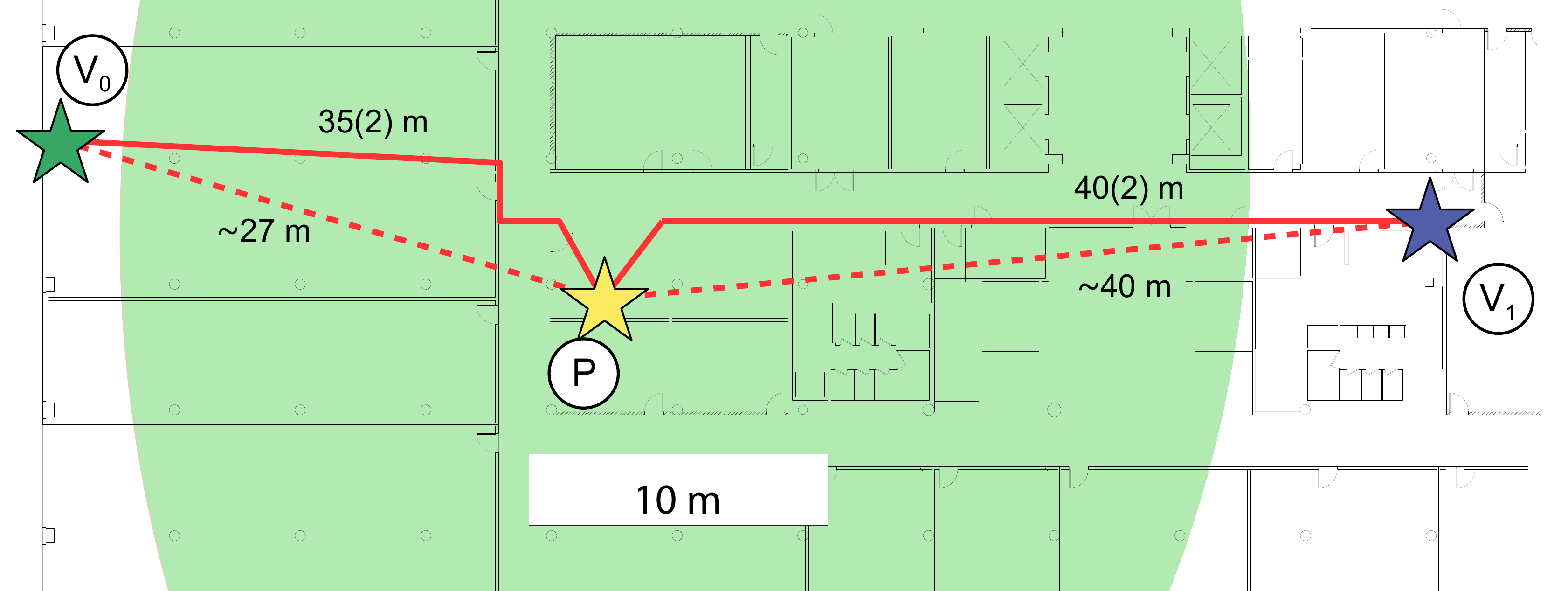}
    \caption{\textbf{Positional uncertainty diagram.} The area in which the prover could be located based on the roundtrip time of the QPV signals is shaded in green and overlaid over a diagram of our experiment, including the real locations of both verifiers and of the prover, the cable paths and the building's floorplan.}
    \label{fig:Positional_Uncertainty}
\end{figure*}

\section*{Discussion}

Thus far we have demonstrated that the LT-QPV can be implemented experimentally using commercially available hardware and standard telecom infrastructure against i.i.d. adversaries with limited query access. Together with finite-size security guarantees against both i.i.d. and general entangled adversaries under the QROM, and a source-independent security proof that places no trust assumptions on the prover's hardware, LT-QPV serves as a potential candidate for
positional verification at the metropolitan scale.

Our protocol occupies a distinct position within the current QPV experimental landscape. Prior demonstrations have each addressed only a subset of the main deployability challenges, namely security, channel loss tolerance at metropolitan scales, or hardware accessibility \cite{Kavuri2026, fan2026relativistic, Kanneworff_2025}. 
Device-independent approaches \cite{Kavuri2026} offer the strongest security guarantees with minimal assumptions but have low loss tolerance, fundamentally limiting their reach over metropolitan distances. 
The implementation by Ref.~\cite{fan2026relativistic} achieves practical simplicity but the underlying protocol~\cite{Llorenc2023} does not provide improved loss tolerance with increased number of measurement bases at practical $0.27\%$ QBER with current security analysis (see Supplementary ~Fig. 12).
Combining it with the loss-tolerance scheme proposed in Ref.~\cite{ABB23} introduces the use of a Bell state measurement and entanglement source that may not have sufficiently good performance for security when commercial devices are used, and would require additional modifications due to the possibility of post-selection on multi-photon events at arbitrary channel loss.
A detailed discussion of the challenges is presented in Supplementary Section VIII A.
Relying on hollow-core fiber remains expensive and sparsely deployed relative to standard telecom infrastructure. 
The single-photon SWAP-based demonstration \cite{Kanneworff_2025} provides valuable proof-of-concept results but lacks provable security against entangled adversaries.

LT-QPV addresses all these simultaneously by decoupling loss tolerance from measurement complexity through a causal argument rather than a hardware workaround. 
It provides sufficiently strong loss-tolerance guarantees that allows for security at metropolitan distances while not being prohibitively challenging to implement with commercial hardware. We note however that scaling to larger separations may still introduce other practical considerations such as higher polarization drift and longer storage times (discussed later) that may translate to additional noise, lower rates and prover-side loss, and the impact depends on the choice of implementation.
The trade-off accepted in this work is that security is proved under QROM and with a computational assumption.
This can be perceived as another form of resource bound on QPV which cannot be information-theoretic secure~\cite{Chandran2009,Buhrman2014}, where query access to $f$ is limited~\cite{Unruh14} instead of bounded entanglement~\cite{Bluhm2022,Llorenc2023,Kavuri2026}.
The computational assumption on the adversary only has to hold during the short QPV protocol run and do not face the same ``harvest-now-decrypt-later" issues that are present in secure key exchange.
Furthermore, QROM is well-precedented in classical \cite{fiatshamir87,bellare1996exact} and quantum cryptography \cite{BDFLSZ11,Unruh14,chiesa2021succinct} and is widely regarded as a practically sufficient assumption \cite{koblitz2015random}.
We also note that classical position verification is insecure even with computational assumptions~\cite{Chandran2009}, so our QPV demonstration still presents a genuine advantage over classical schemes.
As such, we believe that this is a reasonable tradeoff and assumption to make, but a careful choice of the instantiation (e.g. SHA3) of the hash function with sufficiently large input size is required to ensure security under QROM. This suggests that the barriers to QPV deployment at metropolitan distances are not fundamental and that our protocol design choices address the limitations that have constrained prior demonstrations.

A comparison between our demonstration and protocol with the prior demonstrations in shown in Fig.~\ref{fig:scaling_plot}. Kavuri et al.'s implementation~\cite{Kavuri2026} is secure above $2/3$ end-to-end efficiency and Fan Yuan et al.'s implementation~\cite{fan2026relativistic} is secure above $50\%$ efficiency for the link between the quantum verifier and the prover in the asymptotic regime with zero QBER. 
In this setting, their implementation is able to scale to a maximum of $6.0$km and $11.3$km respectively.
Our LT-QPV protocol is secure above $16\%$ efficiency for the prover's delay line and measurement station in the asymptotic regime at our current QBER. 
We compare the performance of our LT-QPV protocol for several configurations: a) system as-is, b) swapping NIR and telecom fibers, c) adding Ultra-Low Loss (ULL) fibers and matching the initial end-to-end efficiency of $81\%$ as Kavuri et al., and d) increasing the source heralding efficiency (HE) from $43\%\rightarrow70\%$ (commercially achievable), and adding future low-loss Hollow-Core Fibers (HCF), see Fig. \hbox{\ref{fig:scaling_plot}}. 
These results indicate that using commercially-available hardware ULL fiber, our LT-QPV protocol can scale to over 100 km, well within the range of metropolitan scale quantum networks.

\begin{figure*}[ht]
    \centering
    \includegraphics[width = 0.9 \textwidth]{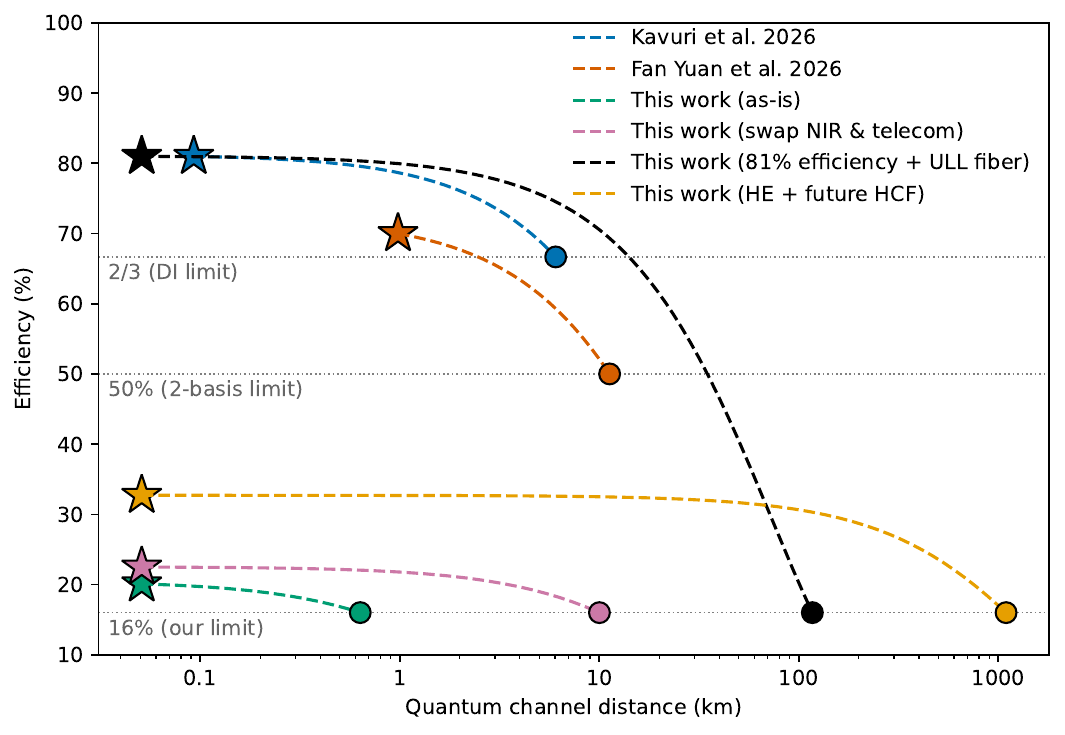}
    \caption{\textbf{Scaling performance of our protocol.} Prior QPV work is secure until the efficiency falls below $2/3$ for the Device Independent asymptotic limit (Kavuri et al.), or $50\%$ for the 2-basis limit (Fan Yuan et al.). Our LT-QPV scaling performance is shown for four configurations: a) system as-is, b) swapping Near-Infrared (NIR) and Telecom fibers, c) adding Ultra-Low Loss (ULL) fibers and matching the initial end-to-end efficiency of $81\%$ as the prior work (Kavuri et al.), and d) increasing the source heralding efficiency (HE) from $43\%\rightarrow70\%$, and adding future low-loss Hollow-Core Fibers (HCF). Stars denote initial operating point, while end points denote scaling limit.}
\label{fig:scaling_plot}
\end{figure*}

While LT-QPV addresses many of the outstanding challenges of QPV protocols, it still faces limitations and practical requirements.
A key protocol-level requirement is the necessity for the prover to store one qubit of an entangled pair (subsystem $Q$, see Fig.~\ref{fig:metro}\textbf{b}) until the classical challenges arrive and the measurement basis is determined. 
This storage time is constrained by the need for subsystem $V$ to arrive at the quantum verifier before challenge $y$ hypothetically would, with contributions from: 
(1) the latency of transmitting $y$ from $V_1^C$ to the prover, $\Delta t_{\mathrm{cl}}=d_{V_1^C\prover}\brac{\frac{1}{v_{\mathrm{cl}}} - \frac{1}{c}}$, and (2) the latency for transmitting $V$ from the prover to $V_0^Q$, $\Delta t_V=d_{V_0^Q\prover}\brac{\frac{1}{v_V} - \frac{1}{c}}$, where $d_{AB}$ represents the distances between $A$ and $B$ (see Supplementary Section III B 1 for details).
In favorable settings where both the classical communication and the quantum link operate close to $c$ (e.g., free-space or hollow-core-fiber segments), the required storage time can be made arbitrarily small without the need for a quantum memory. By contrast, in a practical setting where quantum transmission is performed over standard optical fiber with $2/3\cdot c$ transmission speed, the storage time scales linearly with distance.
In this regime, depending on the choice of quantum memory the required storage time can indirectly couple prover efficiency to channel loss, potentially affecting channel loss tolerance.
Nevertheless, even without long-lived, error-corrected quantum memories, passive optical delay (fiber spools or delay lines), as used in our experiment, or a commercially-available free-space photonic memory \cite{arnold2024all} can provide sufficient storage for LT-QPV to retain better loss-tolerance scaling than \cite{fan2026relativistic}.
For instance, using the same standard fiber in the delay loop with the same attenuation as the quantum channel can allow the loss in the fiber loop to be only $1/3$ of the fiber loss in the quantum channel (see Supplementary Section III B 1 for an example).

Our current demonstration has two further implementation-specific limitations that we expect can be addressed through hardware improvements achievable with existing or near-term technology. The first is the short length of the challenge signals (4 bits), which provides security only against adversaries with limited queries to $f$.
The choice of a small 4-bit challenge is due to the short distance of our demonstration, in which the additional latency required to compute $f$ for longer challenges would introduce a timing uncertainty that is too large for our setup.
At the metropolitan-scale separations the protocol is designed for, the correspondingly larger timing budget readily accommodates an extension to a 128- or 256-bit challenge and use of a cryptographic hash function such as SHA3~\cite{SHA3}, narrowing the gap with the ideal random oracle.

The second limitation is that the implementation is only secure against i.i.d. attacks but vulnerable to general ones. This is partly due to the slow switching speed of the verifier's liquid crystal retarders (LCRs), which leads to multiple consecutive valid rounds sharing the same verifier measurement basis that an adversary may exploit (more details in Supplementary Section V D). Replacing the LCRs with commercial Pockels cells would provide the fast per-round basis switching required for security against general attacks.
Even with such fast switching, security against general attacks imposes substantially stronger finite-size requirements.
With our current performance, a large number of valid rounds $N\approx5\times10^8$ would be needed to be secure against general attacks (see Supplementary Section V C for details).
We expect that a tighter security analysis against general attacks, which we leave as an open problem, could reduce these requirements, and that performance improvements could further reduce both the finite-size requirements and the time to accumulate the necessary rounds.
A separate improvement that can be made is to  tighten $\epssou$ to provide a stronger cryptographic guarantee by having a longer net collection time.
Selecting $\epssou=2^{-64}$ would require $N=1.5\times 10^6$ valid rounds and a net collection time of 130min while selecting $\epssou=10^{-6}$ would require $N=5.5\times 10^5$ valid rounds and a net collection time of 48min. Addressing both limitations, extending the challenge length and replacing the liquid crystals with fast Pockels cells, represents a clear and technically feasible path toward a QPV demonstration that is fully secure against general QPT adversaries.

Besides hardware improvements to address query size and security level, further safeguards are required to address possible side-channel attacks on the measurement device, which LT-QPV assumes to be trusted.
Depending on the implementation of the measurement devices, there are known side-channel attacks on quantum measurement devices~\cite{Jain2016,Xu2020,BSI_QKD_Attacks}, such as detector blinding~\cite{Makarov2009,Lars2010} and Trojan-horse attacks~\cite{Vakhitov2001,Gisin2006}, and various countermeasures have been proposed, such as introducing variable attenuation and photocurrent monitoring~\cite{Yuan2011,Koehler2018}, random detector efficiency~\cite{lim2015random}, or careful quantification of leakages~\cite{Pinheiro18}.
We leave a complete implementation of LT-QPV with suitable protection against side-channel attacks for future work. 

LT-QPV is a protocol intended for deployed networks, thus fast operation is an important performance metric limited by the number of required rounds and the hardware repetition rate. The prover efficiency currently limits the number of valid rounds for position verification. The warm-rubidium entanglement source used in our demonstration provides a native heralding efficiency of $\sim43\%$, which translates into a prover efficiency of $\eta \approx 20\%$ due to losses in the measurement apparatus.
The most direct route to improved performance is to increase the system detection efficiency by upgrading to higher-efficiency SNSPDs and reducing fiber-coupling losses. 
Reaching the experimentally demonstrated $\sim 81\%$ end-to-end detection efficiency~\cite{Kavuri2026} would boost the prover efficiency to $\eta \approx 35\%$.
This shifts the operating point to higher transmission at comparable QBER, reducing the number of valid rounds required by more than an order of magnitude (see Fig.~\ref{fig:Operating_Curve}). 
A further improvement in source heralding efficiency of $70\%$, which is commercially achievable, 
can further improve the prover-side efficiency to $>56\%$, requiring $N<2500$ rounds against i.i.d. adversaries and $N<5\times 10^4$ against general attacks.
Further improvement of the prover-side efficiency of LT-QPV to match the end-to-end efficiency of Ref.~\cite{Kavuri2026} (81\%) reduces the requirement to $N<5000$ against general attacks. 
Higher repetition rates are achievable by increasing the switching rate of the prover's Pockels cells. 
In principle, rates up to 100 MHz are possible, although 1-10 MHz is likely more realistic, giving a 5 to 50 times boost in rate. 
Switching the verifier's LCR to Pockels cells also allows removes the need for sifting, providing another 16 times boost in rate.
With these reasonable updates, a certification time $\lesssim1$ s is achievable against general attacks.  

With a credible path to sub-second certification times, we now present a conceptual design for deploying LT-QPV in metropolitan-area networks.
In the most pared-down version of our protocol, only one verifier requires quantum capabilities (a fiber-optic link, single-photon detectors, and photon polarization control). All other verifiers operate purely classically, exchanging challenges and receiving measurement outcomes via wireless or optical links. As illustrated
in Fig.~\ref{fig:metro}\textbf{a}, in a practical deployment a small number of quantum verifier nodes connected to the coverage area by standard telecom fiber could provide quantum measurements for a much larger number of provers within the network. Multiple classical verifier nodes, each
consisting of a time tagger and a synchronized clock, are
distributed across the metropolitan area at surveyed positions,
providing the geometric diversity required for three-dimensional
triangulation from the arrival times of the prover's classical
responses. Multiple provers can operate within
this shared infrastructure simultaneously, with each carrying
an entanglement source and a measurement module. The quantum
infrastructure thus reduces to a few nodes serving the entire
coverage area. The deployment scales through the addition of
inexpensive classical verifiers, components that are mature,
commercially available, and already deployed at scale for
applications such as 5G network synchronization and financial
exchange co-location. For a target metropolitan network distance of
$12$\,km (about half the length of Manhattan island), fiber losses at telecom
wavelengths are fully absorbed by the protocol's arbitrary loss
tolerance, requiring no quantum repeaters or entanglement
distillation. This architecture inverts the usual assumption in
quantum position verification that extending coverage requires
proportionally scaling the quantum hardware, making city-wide
deployment economically viable. 
A formal security analysis of LT-QPV beyond the one-dimensional (collinear) setting is required for a complete security treatment of such deployments, and we leave this for future work.
Moving to a network setting breaks the ideal collinear configuration of QPV, shown in Fig.~\ref{fig:metro}\textbf{b}, thereby increasing the necessary quantum storage time to $\Delta t_V\approx d_{V_0^QP}(\frac{1}{v_V}-\frac{\cos\theta}{c})$ as a necessary condition for security, where $\theta$ is the angle the classical verifier makes with the line connecting the quantum verifier to the prover (see \Cref{fig:metro}\textbf{c}).
However, since the angle $\theta$ scales inversely with the number of classical verifiers $k$ ($\theta \sim \pi/k$), the storage time can in principle be reduced to a level similar to that of the collinear case by a careful choice of verifiers' number and locations (more details in Supplementary section III B 2).

\begin{figure*}[ht]
    \centering
    \includegraphics[width = \textwidth]{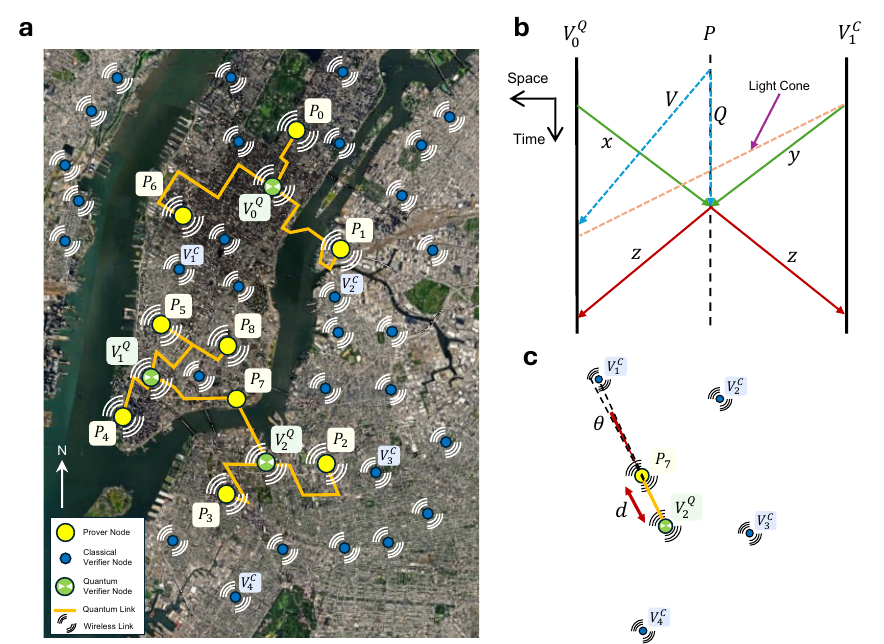}
    
    \caption{\textbf{a. Conceptual design of LT-QPV deployment architecture on a metropolitan scale. } 
    In a real-life realization of our LT-QPV protocol, multiple verifiers and prover nodes may be scattered around a metropolitan area, New York City as an example. 
    A few quantum verifier nodes (green) with measurement devices will be connected via fibers (black links) to several prover nodes (yellow) each.
    Classical verifiers (blue) can be scattered around the city. All classical communications can take place wirelessly. \textbf{b. Quantum storage time estimation in collinear setting.} Quantum verifier $\verA$ and classical verifier $\verB$ certifying the location $\ploc$ of a prover $P$. The timings have to be chosen such that challenge $y$ arrives at $\verA$ after the quantum system $V$, i.e. the blue arrow must always be above the orange dotted line (future light cone of the event where $\verB$ sends out challenge $y$). \textbf{c. Scaling of quantum storage time of the prover.} An example configuration of a prover, a quantum verifier, and a subset of classical verifiers (blue) to be used in the position verification protocol. The storage time $\Delta t_V + \Delta t_\mathsf{cl}$ depends on the distance $d$ between the prover and the quantum verifier, as well as the minimum angle $\theta$ between the prover, the quantum verifier, and a classical verifier. }
    \label{fig:metro}
\end{figure*}

The QPV protocol and deployment we propose suggests a concrete model for spacetime certification as a service. 
At the moment of signing a document, certifying a regulatory filing, or authenticating a high-value transaction, a client acting as the prover initiates the LT-QPV protocol: one photon from their entanglement source travels via fiber to a quantum verifier, while classical responses are timed against the surrounding classical verifiers to determine the prover's position. 
Upon successful verification, the verifiers jointly issue a digitally signed certificate attesting to the prover's position and the time of the event, a spacetime seal anchored to the laws of physics rather than to the honesty of the signee or to device integrity or self-reported metadata. 
This certificate can subsequently be verified by any relying party or can serve as the location attestation in TCC.
Realizing a spacetime seal service will require systematic protocol design and engineering developments, but this work sketches a concrete path: improving prover efficiency and valid-round rate through higher end-to-end detection efficiency and faster switching; strengthening robustness beyond i.i.d. operating regimes through rapid, per-round basis selection; extending the challenge length to better approximate the random-oracle ideal; and validating kilometer-scale operation over deployed urban fiber with multiple classical verifiers for full three-dimensional triangulation.
The LT-QPV architecture integrates naturally with emerging quantum networks in which entanglement is generated at network nodes and distributed for measurement, enabling LT-QPV to provide spacetime-seal certificates as one application running alongside other entanglement-enabled services, such as key distribution and distributed sensing, on the same infrastructure.

\section{Methods} \label{sec:methods}

\subsection{Protocol Description}

We provide an informal description of the LT-QPV protocol here, and the full protocol can be found in Supplementary Section III B. 
The protocol begins with repeating QPV rounds until there are $N$ valid rounds:
\begin{enumerate}
    \item \textbf{Quantum State Transfer}: The prover prepares a Bell state $\ket{\Phi^+}_{VQ}$. The prover sends quantum system $V$ to $\verA$ and keeps quantum system $Q$.
    \item \textbf{Verifier Message}: $\verA$ and $\verB$ randomly selects and sends $x\in\{0,1\}^n$ and $y\in\{0,1\}^n$ respectively such that they both arrive at $P$ at time $t_{P}$.
    \item \textbf{Verifier Measurement}: $\verA$ computes basis $\alpha=\hash(x,y)$ and measures $V$ in basis $\alpha$ with measurement outcome $s$ (which can include double-clicks, $s=dc$). If no photon is detected or if quantum system $V$ fails to arrive by $t_{V}$, the round is considered invalid.
    \item \textbf{Prover Measurement}: After receiving $x$ and $y$, the prover computes $\alpha=f(x,y)$ and measures $Q$ in basis $\alpha$, with measurement outcome $z$. If no photon is detected, the prover sets $z=\perp$. The prover sends $z$ to both verifiers ($\respA$ to $\verA$ and $\respB$ to $\verB$). 
    \item \textbf{Timing Check}: The verifiers check if the responses arrive within a threshold $\timeth$ of the expected arrival time. If any check fails, the protocol aborts immediately.
\end{enumerate}

Of the $N$ valid rounds, let $\Ndet$ be the number of rounds detected by the prover ($z\neq\perp$), $\Ndc$ be the number of double-click rounds ($s=dc$), $\Nerr$ be the number of error rounds ($s\neq z$), and $\Nmis$ be the number of rounds with mismatched responses ($\respA\neq\respB$). The verifiers conclude that the prover is at the claimed location if 
\begin{enumerate}
    \item Prover efficiency is above a certain threshold, $\Ndet/N\geq \etath$
    \item Effective QBER is below a certain threshold, $(\Ndc+\Nerr)/\Ndet\leq \errth$
    \item Probability of mismatched responses is below a threshold, $\Nmis/N\leq \pmisth$
\end{enumerate}

\subsection{Completeness} \label{sec:completeness}

Recall that the completeness error $\epscom$ is the probability that an honest prover fails the LT-QPV verification.
Consider an honest prover at location $X_P$ and verifiers with imperfect devices that results in transmission $\eta$, QBER of $e$, and double-click probability $p_{dc}$. 
Since classical communication is almost error-free, we assume that an honest prover can always provide matching responses and respect the timing constraints. 
The protocol proceeds in an i.i.d. manner, allowing us to use tight concentration bounds~\cite{Zubkov2013_SVBound} to bound the probability of failing the checks.
More formally,
\begin{theorem}
\label{thm:completeness}
    Consider an honest implementation with single-round transmission $\eta$, QBER $e$, double-click probability $\pdc$, and mismatch probability $\pmis$. 
    LT-QPV is $\epscom$-complete with completeness error
    \begin{equation*}
    \begin{split}
        \epscom=&\Phi\left(-\sqrt{2NH\left(\etath+\frac{1}{N},\eta\right)}\right)\\
        &+1-\Phi\left(\sqrt{2N\etath H\left(\errth-\frac{1}{N\etath},e'\right)}\right)\\
        &+1-\Phi\left(\sqrt{2NH\left(\pmisth-\frac{1}{N},\pmis\right)}\right),
    \end{split}
    \end{equation*}
    where $e'=p_{dc}+(1-p_{dc})e$, $H(x,p)=x\ln(x/p)+(1-x)\ln[(1-x)/(1-p)]$ is the Kullback-Leibler divergence for a binary distribution, $\Phi(x)$ is the cumulative distribution function of a standard normal distribution, $\etath$, $\errth$ and $\pmisth$ are the threshold values for transmission, QBER, and mismatch probability and are selected such that $\etath\leq\eta-\frac{1}{N}$, $\errth\geq e'+\frac{1}{N\etath}$ and $\pmisth\geq\pmis+\frac{1}{N}$.
\end{theorem}

\begin{proof}[Proof Sketch.]
The probability of failing any check is upper bounded by the sum of probability that each check failed.
Since the honest party operates in an i.i.d manner, we bound the probability of each check failing using tight concentration bound from Ref.~\cite{Zubkov2013_SVBound}.
\end{proof}

\subsection{Soundness}

Recall that soundness is the probability that an adversary without a presence at the claimed location passes the LT-QPV verification.
\begin{theorem} \label{thm:soundeness_iid}
    LT-QPV is $\epssou$-sound against i.i.d. query-bounded adversaries in the quantum random oracle model, i.e.
    \begin{equation*}
        \max_{\advstrat\in\advset_{QPT,iid},\measstratprat}\Pr[checks\, pass]\leq\epssou,
    \end{equation*}
    where the maximum is taken over adversary strategies $\advstrat$ that are in the set of QPT and i.i.d. strategies $\advset_{QPT,iid}$,
    \begin{equation*}
    \begin{split}
        \epssou=&\varepsilon_{ROM}+\\
        \max&\left\{1-\Phi\left(\sqrt{2NH\left(\etath-\frac{1}{N},\etab\right)}\right),\right.\\
        &\Phi\left(-\sqrt{2NH\left(\pmisth+\frac{1}{N},\pmisb\right)}\right),\\
        &\left.\Phi\left(-\sqrt{2NH\left(\errth+\frac{1}{N},e_{SDP}\right)}\right)\right\}
    \end{split}
    \end{equation*}
    for any choice of parameters $\etab\leq\etath-\frac{1}{N}$ and $\pmisb\geq\pmisth+\frac{1}{N}$, and where $e_{SDP}(\etab,\pmisb)$ is the solution to the SDP in Equation (IV.19), with the added condition that $e_{SDP}\geq\errth+\frac{1}{N\etath}$.
\end{theorem}
The proof of soundness is split into five distinct steps, with the first three steps simplifying the analysis and the final two steps proving directly that the adversary cannot pass the checks.
\begin{itemize}
    \item \textbf{STEP 1:} Apply the universal squashing model~\cite{FCL11} to reduce to a scenario where the first verifier's measurement device only accepts single-qubit inputs. 
    \item \textbf{STEP 2:} Invoke the QROM and computational assumption to ensure that basis information is only revealed at $(t_P,X_P)$ and to rule out attacks that exploit the structure of function $f$.
    \item \textbf{STEP 3:} Demonstrate that the soundness is independent of channel loss, allowing the no-detection rounds at $V_1$ to be safely discarded.
    \item \textbf{STEP 4:} Derive a relationship between the single-round expected QBER, probability of mismatch of responses and transmission for any adversary strategy. 
    \item \textbf{STEP 5:} Using statistical methods such as concentration bounds, demonstrate that given the relationship between the expectation values, no adversary can pass the checks except with low probability.
\end{itemize}
The detailed proof is provided in Supplementary Section IV C.

The security proof extends to general QPT adversaries by the same argument, with two key modifications: (i) the effective QBER check is replaced by $\Ndc+\Nerr\leq N\tilde{e}_{\mathsf{th}}$, and (ii) statistical analysis in Step 5 is carried out using conditional expectation values instead of unconditional ones.
The soundness against general adversaries is presented in Supplementary Theorem V.2, and the detailed proof is provided in Supplementary Section V B.

\subsection{Numerical Analysis}

To evaluate the parameter regime where QPV is secure, we numerically simulate the minimum performance that an honest prover has to achieve for a given transmission $\eta$. 
The results are shown in Fig.~\ref{fig:Operating_Curve}, with security parameters $\varepsilon_{com}=10^{-2}$ and $\varepsilon_{sou}=10^{-2}$ for various number of rounds.
The completeness budget $\varepsilon_{com}$ is distributed with $0.1\varepsilon_{com}$ attributed to the statistical fluctuations of each of $\eta$ and $\pmis$ and $0.8\varepsilon_{com}$ attributed to the statistical fluctuations of QBER (terms of the completeness in \Cref{thm:completeness}).
The simulation is carried out by solving the semidefinite program (SDP) in Equation (IV.19), where we use ncpol2sdpa~\cite{Ncpol2sdpa} and PICOS~\cite{picos} interface with QICS~\cite{qics} solver.

\subsection{Experimental Setup}

An overview of our experimental setup is presented in Fig. \ref{fig:1}\textbf{a}. The setup consists of two trusted verifier nodes and one untrusted prover node connected by a quantum and a classical layer. 

The core of the quantum layer is constituted by a warm rubidium entanglement source \cite{craddock2024high} placed at the prover's location. The source provides entangled photon pairs at 795 nm (near infrared, or NIR) and 1324 nm (telecom) in the $|\Phi^+\rangle = \frac{1}{\sqrt{2}}(|HH\rangle+|VV\rangle)$ Bell state. The specific requirement of our QPV protocol to maximise the NIR heralding efficiency led us to modify the source with respect to previous work. Enriched $^{85}$Rb is used as the medium for the four-wave mixing, and the pumping is enacted by lasers at 780 nm and 1367 nm. 
The source is operated such that, at the origin, it outputs about 4.2 million NIR photons/s, 5.3 million telecom photons/s, and 2.1 million coincidences/s leading to a $43 \%$ heralding efficiency (defined as the rate of coincidences divided by the rate of telecom photons). The effective signal-idler cross-correlation for the coincidence window we use in the experiment is $g_{si}\approx 190$, leading to a minimum achievable QBER of $\text{QBER}_{\text{min}}\approx 0.52\,\%$. 

After their generation, the photons from the entanglement source are coupled into optical fibers. The NIR photon is brought via a 42 m delay line to its measurement station comprised of two Pockels cells, a half-wave plate (HWP), a quarter-wave plate (QWP), and a polarizing beam splitter (PBS). The Pockels cells and wave plates are used to set the measurement in one of the 16 bases used in our protocol. After the PBS, the photons are coupled into optical fibers and routed to superconducting nanowires single photon detectors (SNSPDs).

The telecom photon serves as an agreement by the prover to perform a QPV round and is transmitted from the prover to verifier 0, where it leads to a second measurement station where two liquid crystal retarders (LCRs) at $45^\circ$ to each other are used to set arbitrary measurement bases and a PBS and another set of SNSPDs are used for projective measurements at the verifier's location. 

The classical layer is coordinated by three FPGAs connected by high-speed copper-shielded coaxial air-core plenum cables that can transmit signals at $0.88\cdot c$ \cite{AirCoreCoaxDatasheet_v1_6}, where $c$ is the speed of light. 
Verifiers $V_0$ and $V_1$ transmit classical messages $x$ and $y$ consisting of 4-bit binary numbers that are chosen uniformly at random from $\{0,1\}^4$ to the prover, respectively.
The verifiers coordinate the transmission of $x,y$ such that they arrive simultaneously at the prover's claimed location. Upon receiving classical values of $x,y$ the prover selects the measurement basis from a random lookup table with 16 possible measurement bases, and sends the appropriate control signals to fire the Pockels cells. After passing through the Pockels cells and measurement optics, the photon is detected by SNSPDs. The classical measurement result is relayed back to each verifier using high-speed air-core cables. All the signals necessary for the security determination ($x$ and $y$ values, Bell state measurement results, synchronization clocks) are logged onto two time taggers at the verifiers. 
Any rounds with mismatched prover and verifier measurement bases are discarded.

Our data collection system currently takes approximately 22.7 s to collect 5 s of data. This overhead arises from two sources: first, a conservative 15 s settling period allotted for the LCR voltages to stabilize prior to QBER measurement, and second, the remainder from the data transfer method. Neither represents a fundamental limitation. In practice, the LCRs settle within approximately 1 s, and the transfer overhead can be eliminated by utilizing the network features of the Swabian time tagger, such that 5 s of data would require only 6 s in total. 

\begin{figure}[ht]
    \centering
    \includegraphics[width = \columnwidth]{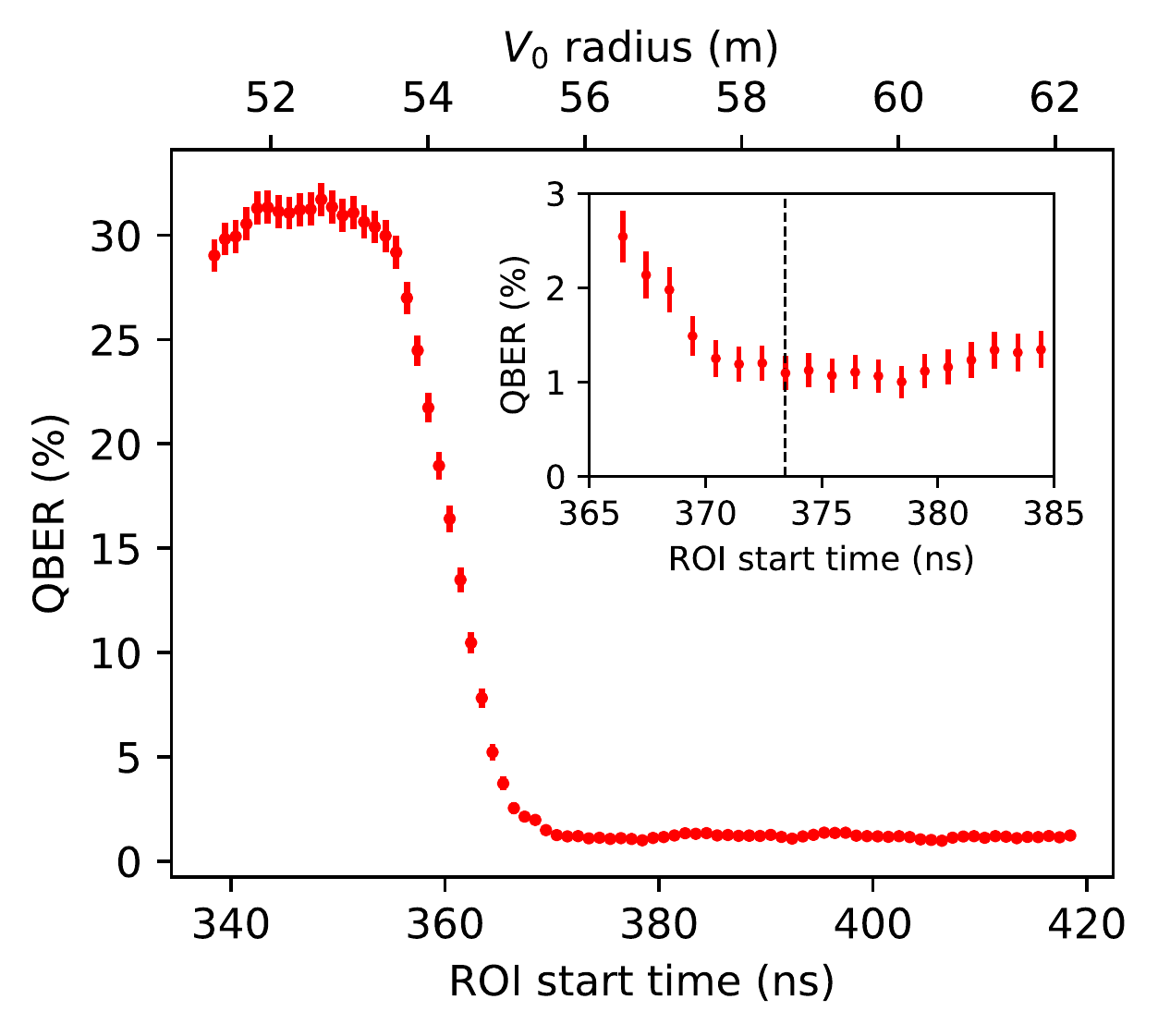}
    \caption{ROI start time tuning accounts for the delays in the prover's quantum measurement. Specifically, we place the ROI start time (the dashed vertical line in the inset plot) immediately after the Pockels cell has finished settling, which corresponds to low QBER, and reasonably low positional uncertainty. Moving the ROI start time earlier could achieve improved positional uncertainty, but at the expense of higher QBER.}
    \label{fig:ROI_Time}
\end{figure}

Central to the positional determination of QPV are the operational delays at the prover. Measurement delays including receiving and processing the \textit{x} and \textit{y} values, processing latencies in determining the requested measurement basis, firing the Pockels cells, and detecting the projected single photon result, all increase the uncertainty radius around each verifier (specific details are provided in the Supplementary Information). Instead of predicting the sum of these effects, the verifiers have the ability to set the region of interest (ROI) where to accept photon measurements to optimize experimental parameters, namely positional uncertainties and QBER. This allows us to select photons that reach the Pockels cells right after the measurement basis is settled, minimizing delays. To determine the ROI, we perform the measurement shown in Fig. \ref{fig:ROI_Time}, measuring the QBER with varying ROI start time after the prover is ready to receive commands (earlier start times lead to improved positional uncertainties, but at the expense of higher QBER). 

\subsection{Photon Losses} \label{sec:Source Characterization}

The quantum signals experience various losses before being measured and converted into electrical signals logged into the time taggers. The NIR losses directly affect the level of security of the protocol. The telecom losses only affect the time it takes to certify the positional verification. The cumulative efficiency of the NIR and telecom channels respectively were measured to be $46(5)\,\%$ and $41(5)\,\%$. A full breakdown is provided in the Supplementary information.

\bibliographystyle{naturemag}
\bibliography{refs}

\section*{Acknowledgments}

Special thanks to Giuseppe Di Cera (JPMorganChase) for detailed insights and discussions on quantum position verification and quantum networking technologies, notably their potential commercial applications. We thank Brett Sanford (JPMorganChase) for administrative support for this project. We thank Marco Pistoia for assisting with the internal patent filing approval process for patent IDF-2024-0733 during his tenure at JPMorganChase. C.L. and I.W.P. contributed to this work during their tenure at JPMorganChase. Additionally, a major thank you to Noel Goddard of Qunnect, without whom this collaboration would not have been possible.

\section*{Competing Interests}
W.Y.K., I.W.P., C.L., and K.C. are co-inventors on a patent application related to this work (U.S. Patent US-12675596-B2 ``Systems and methods for source-independent quantum position verification").
A.C., T.S., J.S., O.A. are co-inventors on a patent application related to this work (U.S. Patent Application 63/782,777 ``Low-Latency Wireless Communication").

\section*{Disclaimer}
This paper was prepared for informational purposes with contributions from the Global Technology Applied Research center of JPMorgan Chase \& Co. This paper is not a product of the Research Department of JPMorgan Chase \& Co. or its affiliates. Neither JPMorgan Chase \& Co. nor any of its affiliates makes any explicit or implied representation or warranty and none of them accept any liability in connection with this paper, including, without limitation, with respect to the completeness, accuracy, or reliability of the information contained herein and the potential legal, compliance, tax, or accounting effects thereof. This document is not intended as investment research or investment advice, or as a recommendation, offer, or solicitation for the purchase or sale of any security, financial instrument, financial product or service, or to be used in any way for evaluating the merits of participating in any transaction.

\newpage



\section*{Supplementary material}

# Supplementary Materials for "Loss-Tolerant Quantum Position Verification for Metropolitan Area Networks"

## CONTENTS

# I. NOTATION

We list our notation in Table I. Note that $N_{\max}$ and $N$ are fixed parameters decided before the protocol begins, whereas $N_{\rm det}, N_{\rm dc}, N_{\rm err}, N_{\rm mis}$ are values observed during the protocol.

| | |
|---|---|
| $V_0 = V_0^Q, V_1 = V_1^C$ | Left (quantum) and right (classical) verifiers |
| $P$ | Prover |
| $X_0, X_1$ | The position of $V_0, V_1$ |
| $X_P$ | Prover's claimed location |
| $x, y$ | Left and right classical challenges |
| $t_x$ | Time that classical challenge $x$ is transmitted from $V_0$ |
| $t_y$ | Time that classical challenge $y$ is transmitted from $V_1$ |
| $\alpha$ | Basis choice |
| $s$ | Measurement outcome of $V_0$ |
| $\|\psi_0^\alpha\rangle, \|\psi_1^\alpha\rangle$ | Measurement basis for basis choice $\alpha$ and outcome $s = 0, 1$ |
| $z$ | Prover's response |
| $z_0$ | Prover's response to $V_0$ |
| $z_1$ | Prover's response to $V_1$ |
| $t_{z_0}$ | Time that prover's response is received by $V_0$ |
| $t_{z_1}$ | Time that prover's response is received by $V_1$ |
| $\Delta_X^0$, $\Delta_X^1$ | Positional uncertainty (left and right) |
| $Q$ | Prover's quantum register |
| $V$ | $V_0$'s quantum register |
| $t_V$ | Time of $V$ arriving at $V_0$ |
| $t_P$ | Time of $x, y$ arriving simultaneously at $X_P$ |
| $f : \{0,1\}^n \times \{0,1\}^n \to [m]$ | Function that computes the measurement basis (Note $[m] = \{0, \cdots, m-1\}$) |
| $\eta$ | Prover efficiency |
| $p_{\sf mis}$ | Mismatch probability |
| $e$ | QBER |
| $t_{\sf th}$ | Amount of delay tolerated on time of arrival of $z_0$ |
| $\eta_{\sf th}$ | Prover efficiency threshold |
| $e_{\sf th}$ | QBER threshold conditioned on detection by $P$ (error in $N_{\rm det}$ detection rounds) |
| $p_{\sf mis,th}$ | Mismatch probability threshold |
| $N_{\max}$ | Maximum number of rounds |
| $N$ | Target number of valid rounds (rounds with detection by $V_0$, including double-clicks) |
| $N_{\rm det}$ | Number of rounds with detection by $V_0$ and $P$ |
| $N_{\rm dc}$ | Number of rounds with double-click detection by $V_0$ |
| $N_{\rm err}$ | Number of rounds with single detection by $V_0$ with incorrect prover answer ($s \neq z_0$) |
| $N_{\rm mis}$ | Number of rounds with mismatched responses ($z_0 \neq z_1$). |
| $n$ | Security parameter |
| $m$ | Number of measurement bases |
| $d_{AB}$ | Distance between two parties $A$ and $B$ |
| $c$ | Speed of light in vacuum |
| $v_{\sf cl}$ ($\beta_{\sf cl} = v_{\sf cl}/c$) | Speed of light in classical communication channel |
| $\tilde{e}_{\sf th}$ | Amount of error tolerated in all $N$ rounds (unconditioned error) |
| $\mathcal{S}$ | Adversary strategy |
| $\mathcal{A}_a$ | Set of adversary strategies, with restrictions labeled by $a$ |

TABLE I. List of notations used. Subscript $i$ denotes the value of a variable during the $i$-th round with successful detection by $V_0$ (e.g. $x_i$), while superscript $j$ denotes the $j$-th overall round regardless of successful detection (e.g. $x^j$).

# II. PRELIMINARIES

## A. Quantum Random Oracle Model

The Quantum Random Oracle Model [1] is an idealized cryptographic security model, in which a function $f$ is assumed to have no white-box leakage of information. In other words, $f$ is indistinguishable from a truly random function, and the contents of $f$ are only accessible via black-box queries, given by the unitary $U_f \|x\rangle \|y\rangle = \|x\rangle \|y \oplus f(x)\rangle$.

A *query-bounded* algorithm can invoke $U_f$ at most a polynomial number of times in the input length $|x|$. If $U_f$ is queried with quantum state $\rho$, then the *query weight* on input $x$ is given by $\mathrm{Tr}\left(\left(\langle x| \otimes I_{\mathcal{Y}}\right) \rho \left(|x\rangle \otimes I_{\mathcal{Y}}\right)\right)$, where $I_{\mathcal{Y}}$ is the identity operator on the auxiliary register.

We recall the following results from the literature:

**Lemma II.1** (Query Extraction from a Random Oracle [2])**.** *Let $\mathcal{A}$ be an oracle algorithm which makes at most $T$ oracle queries to a function $f : \{0,1\}^m \to \{0,1\}^n$. Define $|\phi_i\rangle$ as the global state after $\mathcal{A}$ makes $i$ queries, and $W_y(|\phi_i\rangle)$ as the sum of squared amplitudes in $|\phi_i\rangle$ of terms in which $\mathcal{A}$ queries $f$ on input $y$. Let $\epsilon > 0$ and let $F \subseteq \{0,1,\ldots,T-1\} \times \{0,1\}^m$ be a set of time-input pairs such that $\sum_{(i,y)\in F} W_y(|\phi_i\rangle) \leq \epsilon^2/T$.*

*For $i \in \{0,1,\ldots,T-1\}$, let $f'_i$ be an oracle obtained by reprogramming $f$ on inputs in $\{y \in \{0,1\}^m : (i,y) \in F\}$ to arbitrary outputs. Let $|\phi'_T\rangle$ be the global state after $\mathcal{A}$ is run with oracle $f'_i$ on the $i$-th query (instead of $f$). Then, $\mathsf{TD}\left(|\phi_T\rangle, |\phi'_T\rangle\right) \leq \epsilon/2$.*

**Lemma II.2** (Simulation of a Random Oracle [3])**.** *A random oracle $f : \mathcal{X} \to \mathcal{Y}$ is perfectly indistinguishable from a $2q$-wise independent hash function $f' : \mathcal{X} \to \mathcal{Y}$ against a quantum algorithm $\mathcal{A}$ which makes at most $q$ queries.*

## III. LOSS-TOLERANT QPV

### A. Protocol Assumptions

The goal of QPV is for two verifiers $V_0$ and $V_1$, located at $X_0$ and $X_1$, to certify the location of a prover $P$ claiming to be at $X_P$. The verifiers are assumed to have:

- **Secure Communication:** The verifiers share a secure classical communication channel.
- **Measurement Device:** $V_0$ possesses a trusted measurement device, which can be modeled by Fig. 3(a). The measurement device is assumed to have no side-channels through which the adversary could learn the measurement basis (e.g. via Trojan horse attack) or control the device externally (e.g. via blinding attack). The device is also assumed to have basis-independent and balanced detector loss, and no dark counts.

Only verifier $V_0$ needs to perform quantum operations, while $V_1$ is fully classical.

We make the following assumptions regarding the honest prover:

- **Position:** The honest prover's claimed location $X_P$ lies on the line segment connecting the two verifiers.
- **Communication:** The honest prover $P$ has a classical communication link to both verifiers $V_0$ and $V_1$ and a quantum channel to verifier $V_0$.
- **Quantum Source**: The honest prover $P$ has access to a quantum source to generate entangled states such as the Bell state.

We note that while we assume basis-independent and balanced detector loss and no dark counts for simplicity, these can in general be accounted for by artificially balancing the detector loss and by modifications in the security analysis respectively. An imbalance on the verifier device can generally be compensated by artificially introducing loss in the SNSPD arm where there is a higher transmission (changing click to no click in the data with some probability) or for specific basis where there is a higher transmission. To account for the dark count rate, they can be incorporated into the security analysis either by (1) modeling the detector POVMs with the same measurement basis, but with an additional map from no click outcome to click outcome for the detector with probability $p_d$ or (2) upper bounding the number of dark count events and treating them in the most favorable manner (detected, no error, no mismatch).

### B. Protocol Description

Before the LT-QPV protocol begins, the verifiers and prover agree on some public parameters. These include the security parameter $n$ and the number of rounds $N$. The number of bases $m$ and the bases $\{|\psi_0^\alpha\rangle, |\psi_1^\alpha\rangle\}_{\alpha\in[1,m]}$ can be chosen based on the loss in the prover's device, since increasing $m$ improves the loss tolerance of the scheme. Ideally, the bases should be chosen to be uniformly spread on the Bloch sphere to minimize any overlap between basis states of different bases. They also share a hash function $f$ to compute the basis choice from the challenges $x$ and $y$, which is assumed to have no white-box leakage of information, i.e. the structure of the function should not reveal any information beyond its input-output relation.

The verifiers also decide on suitable threshold values to meet the performance and security requirements of the protocol. The LT-QPV protocol is described in Protocol 1 below.

**Protocol 1: Loss-Tolerant QPV**

*Goal.* Verifiers $V_0$ and $V_1$ ($V_0$ located at $X_0$ and $V_1$ located at $X_1$) certify that the prover $P$ is at position $X_P \in (X_0, X_1)$.
*Public Parameters.* Security parameter $n$. Number of bases $m$. Target number of valid rounds $N$. Maximum number of rounds $N_{\max}$. Hash function $f : \{0,1\}^n \times \{0,1\}^n \to [m]$. Effective QBER threshold $e_{\mathsf{th}}$. Time threshold $t_{\mathsf{th}}$. Prover efficiency threshold $\eta_{\mathsf{th}}$. Mismatch threshold $p_{\mathsf{mis,th}}$.

1. Set $i = 1$ (current valid round) and $j = 1$ (current round). Repeat until $N$ valid rounds are completed or total round limit is reached, i.e. repeat while $i \leq N$ and $j \leq N_{\max}$:
   (a) **Verifier Description**
      i. **Setup**: $V_0$ randomly selects $x \in \{0,1\}^n$ and $y \in \{0,1\}^n$ and sends $y$ to $V_1$ using a secure channel. The verifiers communicate with the prover to agree on a time $t_V$ when the quantum system $V$ is expected to arrive at $V_0$. The verifiers also communicate and agree on a time $t_P$ for $x$ and $y$ signals to arrive at $X_P$.
      ii. **Verifier Message**: $V_0$ and $V_1$ send $x$ and $y$ respectively at $t_P - d_{V_0P}/v_{\mathsf{cl}}$ and $t_P - d_{V_1P}/v_{\mathsf{cl}}$ such that they both arrive at $P$ at time $t_P$.
      iii. **Verifier Quantum Measurement**: $V_0$ computes basis $\alpha = f(x, y)$ and measures the received quantum system $V$ in basis $\alpha$. $V_0$ records measurement outcome $s$. If a double-click is detected, the verifier sets $s = dc$. If no photon is detected or if quantum system $V$ fails to arrive by $t_V$, $V_0$ sets $s = \perp$.
   (b) **Prover Description**
      i. **Quantum State Transfer**: The prover prepares a Bell state $|\Phi^+\rangle_{VQ}$. The prover sends quantum system $V$ to the verifier and keeps quantum system $Q$ in a quantum memory.
      ii. **Prover Quantum Measurement**: After receiving $x$ and $y$, the prover computes $\alpha = f(x, y)$ and measures quantum system $Q$ in the computed basis, and obtains outcome $z \in \{0, 1, \perp\}$, where $z = \perp$ corresponds to the case of no detector click. If a double-click is detected, the prover assigns $z \in \{0,1\}$ randomly. The prover sends the measurement outcome $z$ to both verifiers at time $t_P$.
   (c) **Detection Check**: After receiving outcome $z_{0,i}$ at $V_0$ and $z_{1,i}$ at $V_1$, if $s = \perp$ (i.e. no detection at verifier), the round is discarded and the round restarts from (a), incrementing $j = j + 1$. If $s \neq \perp$, the verifiers record the round information $(x_i, y_i, z_{0,i}, z_{1,i}, s_i) = (x, y, z_0, z_1, s)$ and proceed to the final step.
   (d) **Timing Check**: $V_0$ and $V_1$ record the arrival time of the respective responses $z_0$ and $z_1$ as $t_{z_0}$ and $t_{z_1}$. The verifiers check if the timings are valid, (1) $t_{z_0} - t_P \leq d_{V_0P}/v_{\mathsf{cl}} + t_{\mathsf{th}}$, (2) $t_{z_1} - t_P \leq d_{V_1P}/v_{\mathsf{cl}} + t_{\mathsf{th}}$ for some threshold $t_{\mathsf{th}}$. If any check fails, the protocol aborts immediately. The verifiers also note that a successful detection round has been recorded by incrementing $i = i + 1, j = j + 1$.
2. **Parameter Estimation**: If $i < N$, abort the protocol. The verifiers record the number of detected rounds $N_{\text{det}}$ where $z_{0,i} \neq \perp$, the number of double-click rounds $N_{\text{dc}}$ where $z_{0,i} \neq \perp$ and $s_i = dc$, the number of error rounds $N_{\text{err}}$ where $z_{0,i} \neq \perp$ and $z_{0,i} \neq s_i$, and the number of rounds with mismatched responses $N_{\text{mis}}$ where $z_{0,i} \neq z_{1,i}$. The verifiers check if (1) the prover efficiency is above a certain threshold, $N_{\text{det}} \geq N\eta_{\mathsf{th}}$, (2) the effective QBER is below a certain threshold, $N_{\text{dc}} + N_{\text{err}} \leq N_{\text{det}} e_{\mathsf{th}}$, and (3) the rate of mismatched responses is below a threshold, $N_{\text{mis}} \leq N p_{\mathsf{mis,th}}$. If all checks pass, the verifiers conclude that the prover is in the $(X_P \pm \Delta_X)$-region. Otherwise, the protocol aborts.

### 1. *Choice of Timing Parameters*

There are two key timing parameters to be selected. The first is $t_P$, the time at which the challenges $x$ and $y$ are expected to arrive at $X_P$. The second is $t_V$, the time at which quantum system $V$ is expected to arrive at $V_0$. The choices of timing parameters $t_P$ and $t_V$ directly impact the storage time of quantum system $Q$ and the region around $X_P$ where we can certify the location of the prover. When the classical or quantum communication between the prover and verifiers are not at the speed of light in vacuum, the required storage times and verification regime increase.

For a classical communication speed of $v_{\mathsf{cl}} = \beta_{\mathsf{cl}} c$, verifiers $V_0$ and $V_1$ send challenges $x$ and $y$ at times $t_x = t_P - \frac{d_{V_0P}}{c\beta_{\mathsf{cl}}}$ and $t_y = t_P - \frac{d_{V_1P}}{c\beta_{\mathsf{cl}}}$ respectively. Importantly, the timings need to be chosen such that the spacetime coordinate $(t_y, X_1)$ and $(t_V, X_0)$ are spacelike separated to prevent the basis information $\alpha$ from influencing the preparation of quantum system $V$ since $\alpha = f(x, y)$ requires knowledge of $y$. This is illustrated in Fig. 1, where the future light cone of $(t_y, X_1)$ does not contain $(t_V, X_0)$. In practice, this means that the storage time of quantum system $Q$ scales linearly

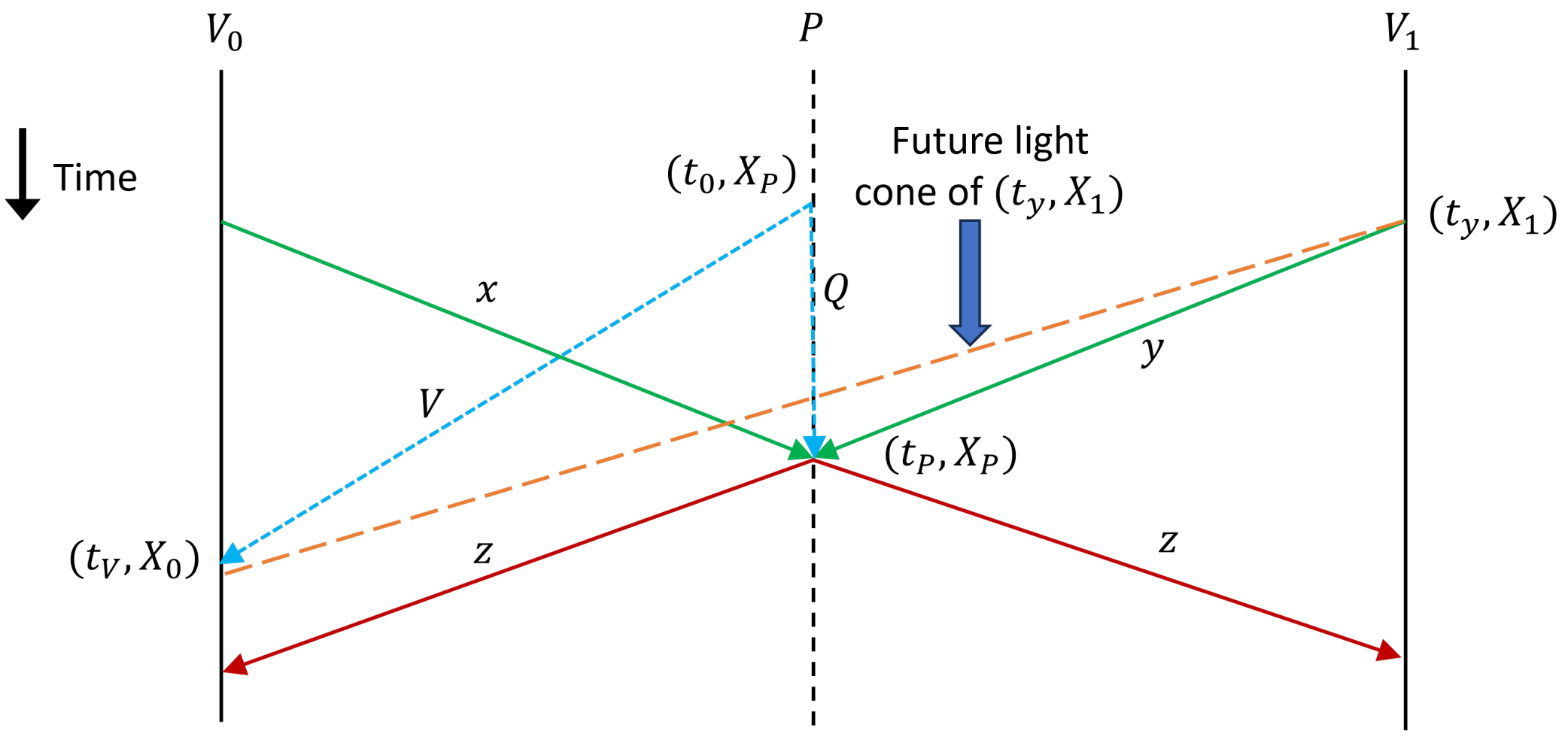


FIG. 1. Communication process between provers and verifiers in the LT-QPV protocol. Verifiers $V_0$ and $V_1$ are located at $X_0$ and $X_1$, and prover $P$ is located at $X_P$. $t_0$ is the time quantum systems $VQ$ is generated, $t_y$ is the time challenge $y$ is sent, $t_V$ is the expected time of arrival of quantum system $V$ and $t_P$ is the expected time of arrival of challenges $x$ and $y$ at $X_P$. The timings have to be chosen such that $y$ (and in turn the basis $\alpha$) cannot influence the quantum system $V$, i.e. at $V_0$, the blue arrow must always be above the orange dotted line (future light cone of the event where $V_1$ sends out challenge $y$).

with the amount of delay in the communication of system $V$ and challenges $x$ and $y$ relative to the speed of light. More precisely, if the joint quantum system $VQ$ is generated at time $t_0$ and $V$ travels at speed $\beta_V c$, then the storage time needs to satisfy $t_P - t_0 > \frac{d_{V_0P}}{c}\left(\frac{1}{\beta_V} - 1\right) + \frac{d_{V_1P}}{c}\left(\frac{1}{\beta_{\mathsf{cl}}} - 1\right)$. Therefore, for high-speed classical communication and free-space or hollow-core fiber quantum communication ($\beta_{\mathsf{cl}}, \beta_V \approx 1$), the storage time can be arbitrarily small. For quantum communication over standard optical fiber ($\beta_{\mathsf{cl}} \approx 1, \beta_V \approx 2/3$), the storage time scales linearly with distance as $\frac{d_{V_0P}}{2c}$. The impact of storage time on the prover-side loss depends on the choice of quantum storage medium. Consider a quantum storage with $\eta_i$ internal storage loss (loss at zero delay) and a delay-dependent loss $\eta(\tau) = 10^{-L\tau}$, where $L$ is the decay rate. The overall prover side loss would thus be $\eta = \eta_i \eta(\tau)\eta'$, where $\eta'$ is the transmission of the remaining components of the prover-side device excluding the quantum storage. If the quantum storage is implemented using a delay loop with standard optical fiber ($L = 0.2c/10^4$), a loop of distance $\frac{1}{3}d_{V_0P}$ is required. This allows us to extend the distance between the prover and verifier $V_0$ by a factor of 3 compared to having the entanglement source at the verifier.

In the ideal scenario where messages are sent at the speed of light and the prover's response to the challenges is instantaneous, the responses are expected to arrive at time $t_{z_0} = t_P + d_{z_0P}/c$ and $t_{z_1} = t_P + d_{z_1P}/c$, and we set $t_{\mathsf{th}} = 0$. This allows the verifiers to certify the location of the prover to be at $X_P$ exactly. However, in practice, the classical communication occurs at a lower speed $\beta_{\mathsf{cl}} < 1$ and the honest prover requires some processing time, which is used to set $t_{\mathsf{th}}$. This extends the time between when the messages are sent, $t_x$ and $t_y$, and the time the messages arrive, $t_{z_0}$ and $t_{z_1}$. The time extension can be mapped to positional uncertainty via four inequalities: the prover must be able to receive a signal from $V_b$ and respond to $V_{b'}$ in time for $b, b' \in \{0,1\}$. Because Protocol 1 imposes symmetric timing constraints on both verifiers, it is enough to consider the case $b = b'$, which restricts the location $X$ of the prover to the intersection of the following two circles:

$$d_{V_0X} \leq \frac{d_{V_0P}}{\beta_{\mathsf{cl}}} + \frac{ct_{\mathsf{th}}}{2}, \quad d_{V_1X} \leq \frac{d_{V_1P}}{\beta_{\mathsf{cl}}} + \frac{ct_{\mathsf{th}}}{2}.$$

The inequalities obtained by the two cases when $b \neq b'$ are redundant due to symmetry. Note that in the general case, they will impose additional constraints as elliptical regions, as discussed in [4].

Consequently, in the case of one dimension (i.e. $V_0, P, V_1$ are on a line), any party within the region between $X_P - \Delta_X^0$ and $X_P + \Delta_X^1$ can theoretically pass the QPV with fast enough processing time and the ability to receive the challenges at the speed of light, where $\Delta_X^0 = d_{V_1P}\left(\frac{1}{\beta_{\mathsf{cl}}} - 1\right) + \frac{ct_{\mathsf{th}}}{2}$ and $\Delta_X^1 = d_{V_0P}\left(\frac{1}{\beta_{\mathsf{cl}}} - 1\right) + \frac{ct_{\mathsf{th}}}{2}$. As such, the verifiers can only certify that the prover is located within this region, which we term $(X_P \pm \Delta_X)$-region.

## 2. Timing Parameters in Networks

In a real metropolitan setting with classical and quantum networks, the prover and the verifiers are better modeled in two dimensions. In general, when the prover does not lie on a line connecting two verifiers, at least three verifiers will be needed and the measurement basis will need to be split into three messages instead of two.

While we do not formally discuss the security analysis in two dimensions, we discuss how the timing parameters may be set to satisfy the spacelike separation requirement (see Fig. 1) as a necessary condition for security as well as its effect on the positional uncertainty region.

We consider a set of $k$ classical verifiers $V_1, V_2, \ldots, V_k$ and a separate quantum verifier $V^Q$. We will argue below that the quantum storage time scales linearly with the distance $d_{V^QP}$ between the quantum verifier and the prover. Accordingly, it is generally wise to pick a configuration such that the quantum verifier is as close as possible to the prover. For instance, one may simply pick the nearest quantum verifier node to be used in a protocol with a given prover location $X_P$.

As before, the quantum state needs to arrive at $V^Q$ at time $t_V$, and the classical signals $x_1, x_2, \ldots x_k$ are scheduled to arrive at the honest prover's location $X_P$ at time $t_P$. Suppose the adversary has a party at each verifier's location. For security to hold, the quantum state sent to $V^Q$ cannot depend on all the classical signals, which yields the condition:

$$t_V < \max_{i\in[k]} t_P - \frac{d_{V_iP}}{\beta_{\mathsf{cl}} c} + \frac{d_{V_iV^Q}}{c}.$$

This translates to a lower bound on storage time of $P$ as before. Let $t_0$ be the time $P$ starts sending the quantum state, then we have

$$t_P - t_0 = (t_P - t_V) + (t_V - t_0) > \min_{i\in[k]} \left( \frac{d_{V_iP}}{\beta_{\mathsf{cl}} c} - \frac{d_{V_iV^Q}}{c} \right) + \frac{d_{V^QP}}{\beta_V c}. \tag{III.1}$$

To see the scaling of the storage time in Eqn. (III.1), suppose the optimal vertex $V_i$ makes an angle $\theta_i$ with the line connecting $V^Q$ and $P$, hence we have

$$\begin{aligned} t_P - t_0 &\approx \frac{d_{V_iP}}{\beta_{\mathsf{cl}} c} - \frac{\sqrt{d_{V_iP}^2 + 2 d_{V_iP} d_{V^QP} \cos(\theta_i) + d_{V^QP}^2}}{c} + \frac{d_{V^QP}}{\beta_V c} \\ &\approx \frac{d_{V_iP}}{c}\left(\frac{1}{\beta_{\mathsf{cl}}} - 1\right) + \frac{d_{V^QP}}{c}\left(\frac{1}{\beta_V} - \cos(\theta_i)\right) \end{aligned} \tag{III.2}$$

up to an $O(d_{V^QP}^2/d_{V_iP}c)$ term which is small when[1] $d_{V^QP} \ll d_{V_iP}$. Note that as the angle $\theta_i$ tends to 0 this converges to the expression for the one dimensional case. Eqn. (III.2) favors classical verifiers that are closer to $P$, and those that make smaller angles with the line $V^QP$.

We depict a possible metropolitan network configuration in Fig. 2. In this configuration, quantum verifier nodes are spread across the area covering possible locations $X_P$, whereas the classical verifiers are positioned in a way that enables picking $k$ classical verifiers that form a convex polygon containing $P$ and $V^Q$ in the interior. In the $\beta_{\mathsf{cl}} \approx 1$ approximation, the first term in Eqn. (III.2) vanishes and the storage time scales linearly with $d_{V^QP}$, which is proportional to $1/\sqrt{N_Q}$ in a quantum network with $N_Q$ quantum nodes over a fixed metropolitan area, assuming the quantum nodes are placed roughly equidistantly.

There are two scenarios that realize the $\beta_{\mathsf{cl}} \approx 1$ approximation. The first is when classical communication is performed close to the speed of light, e.g. radio or free-space communication. Another way the approximation is valid is by assuming the adversary is limited to the same classical communication constraints as the honest parties in the network. Since $\beta_{\mathsf{cl}}$ is the ratio between honest and malicious communication speeds, effectively this means $\beta_{\mathsf{cl}} \approx 1$. One can imagine an adversary tapping into the classical fiber-optic cables and controlling the transmission over the network, rather than using its own separate communication channel, which may induce overhead and make the attack less effective. The optimal $V_i$ for the storage time will be chosen by the minimal angle $\theta_i$, which depends on the configuration. The angle only affects the constant factor since $1/\beta_V - \cos(\theta_i) \in [1/\beta_V - 1, 1/\beta_V + 1]$. In a random configuration with the constraints above, the angle will satisfy $\theta_i \sim \pi/k$.

As for the positional uncertainty region, it is given by the intersection of circles $d_{V_0X} \leq \frac{d_{V_iP}}{\beta_{\mathsf{cl}}} + \frac{t_{\mathsf{th}}}{2}$ for $i \in [k]$ as in the one dimensional case. The region always contains this intersection, and is equal to it when the time thresholds are symmetrical at each verifier as in Protocol 1.

[1] When this is not the case, we have $d_{V_iP} = O(d_{V^QP})$ and the storage time scales linearly with $d_{V^QP}$ as desired.

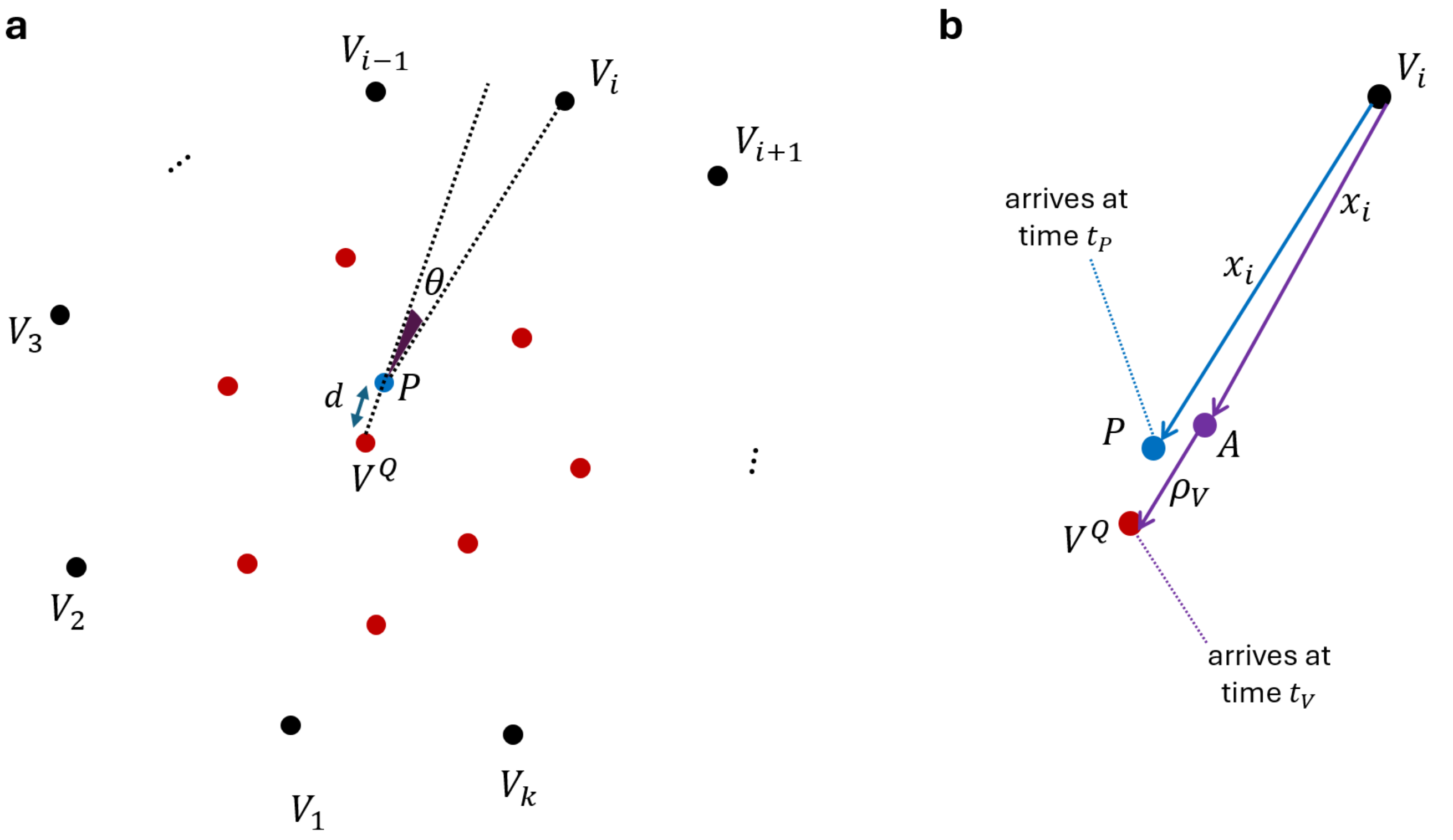


FIG. 2. **Depiction of two-dimensional QPV network. a** The red nodes represent quantum verifiers, and the black nodes represent classical verifiers. The prover $P$ is represented by the blue node. To verify the position of $P$, the nearest quantum verifier $V^Q$ at distance $d$ is chosen for the protocol. $k$ classical verifiers are used such that they form a $k$-gon containing $P$ and $V^Q$. A particular classical verifier $V_i$ makes the minimal angle $\theta_i$ with the line connecting $V^Q$ and $P$. Thanks to the signal sent by $V_i$ the quantum storage time of $P$ scales linearly with $d$, hence proportional to $1/\sqrt{N_Q}$, where $N_Q$ is the number of quantum verifiers placed in a fixed region. For an average position of $P$, the optimal angle satisfies $\theta_i \approx \pi/k$. Using more classical verifiers (i.e. a larger value of $k$) decreases the angle $\theta$ and reduces the storage time by a constant factor. **b** Visualization of the timing constraint induced by $V_i$. A malicious prover $A$ (purple node) must receive classical signal $x_i$ from $V_i$ and send a quantum state $\rho_V$ to $V^Q$ before time $t_V$, whereas $x_i$ must arrive at $P$ at time $t_P$.

## IV. SECURITY OF LT-QPV

### A. Security Definition

Having established the protocol, we now analyze its security. The goal of LT-QPV is to certify the location of an honest prover, and we consider two separate aspects to define security. Firstly, we expect that any honest prover should pass the checks – a secure QPV that fails to certify any honest prover would be of no practical use. This security condition is termed completeness, for which we can define a completeness error $\varepsilon_{\text{com}}$ as the probability that an honest prover fails the checks in QPV. In general, we expect that any implementation should respect the timing constraints, and we focus on the parameter checks.

**Definition IV.1.** *Consider an honest prover $P$ at location $X_P$. Protocol 1 is $\varepsilon_{\text{com}}$-complete if it fails the checks with probability at most $\varepsilon_{\text{com}}$, i.e.* $\Pr[\textit{checks fail}] \leq \varepsilon_{\text{com}}$.

Secondly, for an adversary not at the correct location, we expect that they should fail the QPV checks. Before defining this condition, let us consider how to model the adversary. Since the verifiers certify the location of the prover within the $(X_P \pm \Delta_X)$-region, the adversary can be considered to be any party (or parties) located outside the region. Without loss of generality, it is sufficient to consider an adversary consisting of two parties, which we term Alice and Bob, with Alice located between $V_0$ and $X_P - \Delta_X^0$, and Bob located between $V_1$ and $X_P + \Delta_X^1$. In general, the adversary has access to ideal quantum devices, including quantum computers and entangled states between them. We model such an adversary as holding state $\rho_{E_0}$ at the start of the protocol, then for each round $j \geq 1$:

1. The adversary generates a shared state $\rho_{A^j B^j} = \mathcal{E}^j(\rho_{E^{j-1}})$ between Alice and Bob.

2. When Alice receives $x^j$ and Bob receives $y^j$, they perform an operation on their component of the state respectively. At the same time, Alice also generates quantum system $V^j$ to be sent to the verifier $V_0$. Alice and Bob can also split their quantum system into two components, with one component to be kept (labeled with $K$) and one component to be sent to the other party (labeled with $S$). The resulting state is $\rho_{V^j A^j_K A^j_S B^j_K B^j_S} = \mathcal{E}^{x^j}_{A^j \to V^j A^j_K A^j_S} \otimes \mathcal{E}^{y^j}_{B^j \to B^j_K B^j_S}[\rho_{A^j B^j}]$. $A^j_S$ and $B^j_S$ are sent to Bob and Alice respectively.

3. After time $t^j_P$, when Alice receives $y^j$ and Bob receives $x^j$, Alice performs a measurement on her subsystems $A^j_K B^j_S$ with POVM $A^{x^j,y^j}_{z^j_0}$ while Bob performs a measurement on his subsystems $B^j_K A^j_S$ with POVM $B^{x^j,y^j}_{z^j_1}$. Alice and Bob send $z^j_0$ and $z^j_1$ to verifier $V_0$ and $V_1$ respectively. The verifier performs the measurement with projective operator $\Pi^{f(x^j,y^j)}_{s^j}$ on subsystem $V^j$. The resulting statistics are

$$\Pr\left[s^j, z^j_0, z^j_1\right] = \mathrm{Tr}\left[(\Pi^{f(x^j,y^j)}_{s^j} \otimes A^{x^j,y^j}_{z^j_0} \otimes B^{x^j,y^j}_{z^j_1})\rho_{V^j A^j_K B^j_S A^j_S B^j_K}\right].$$

4. The information $z^j_0$, $z^j_1$, $x^j$ and $y^j$, along with the post-measurement state of Alice and Bob form the new adversary state $\rho_{E^j}$. The information $s^j$, $z^j_0$, $z^j_1$, $x^j$ and $y^j$ are stored for parameter checks (according to protocol requirements).

We label each adversary strategy by the choice of quantum channels and measurement,

$$\mathcal{S} = \{\rho_{E_0}, \mathcal{E}^j, \{\mathcal{E}^{x^j}_{A^j}\}_{x^j}, \{\mathcal{E}^{y^j}_{B^j}\}_{y^j}, \{A^{x^j,y^j}_{z^j_0}\}_{x^j,y^j,z^j_0}, \{B^{x^j,y^j}_{z^j_1}\}_{x^j,y^j,z^j_1}\}.$$

It is known that QPV cannot be unconditionally secure [5, 6], so it is necessary to place additional assumptions on the adversary. Here, we assume that the classical and quantum resources of the adversary, i.e. memory, computational power, and entanglement, are upper-bounded by a polynomial in the security parameter $n$. We label this set of strategies $\mathcal{A}_{QPT}$. We also label the measurement device of the verifier as $\mathcal{M}_{\text{prac}} = \{\Pi^\alpha_s\}_{\alpha,s}$. We define the second security condition, soundness, with the associated soundness parameter $\varepsilon_{\text{sou}}$, as the maximum probability that an adversary passes the QPV checks.

**Definition IV.2.** *Protocol 1 is defined to be $\varepsilon_{\text{sou}}$-sound against a QPT adversary if for any QPT adversary strategy, the QPV checks pass with at most probability $\varepsilon_{\text{sou}}$, i.e.*

$$\max_{\mathcal{S}\in\mathcal{A}_{QPT},\mathcal{M}_{\text{prac}}} \Pr[\textit{checks pass}] \leq \varepsilon_{\text{sou}}. \tag{IV.1}$$

### B. Completeness

The completeness is determined from an honest implementation (to be discussed in Sec. VII) of the QPV protocol. Recall from the main text

**Theorem IV.1.** *Consider an honest implementation with single-round prover efficiency $\eta$, QBER $e$, double-click probability $p_{\mathsf{dc}}$, and mismatch probability $p_{\mathsf{mis}}$. Protocol 1 is $\varepsilon_{\text{com}}$-complete with completeness error*

$$\varepsilon_{\text{com}} = \Phi\left(-\sqrt{2NH\left(\eta_{\mathsf{th}} + \frac{1}{N}, \eta\right)}\right) + 1 - \Phi\left(\sqrt{2N\eta_{\mathsf{th}}H\left(e_{\mathsf{th}} - \frac{1}{N\eta_{\mathsf{th}}}, e'\right)}\right) + 1 - \Phi\left(\sqrt{2NH\left(p_{\mathsf{mis,th}} - \frac{1}{N}, p_{\mathsf{mis}}\right)}\right), \tag{IV.2}$$

*where $e' = p_{dc} + (1-p_{dc})e$, $H(x,p) = x\ln(x/p) + (1-x)\ln[(1-x)/(1-p)]$ is the Kullback-Leibler divergence for a binary distribution, $\Phi(x)$ is the cumulative distribution function of a standard normal distribution, $\eta_{\mathsf{th}}$, $e_{\mathsf{th}}$ and $p_{\mathsf{mis,th}}$ are the threshold values for prover efficiency, QBER, and mismatch probability and are selected such that $\eta_{\mathsf{th}} \leq \eta - \frac{1}{N}$, $e_{\mathsf{th}} \geq e' + \frac{1}{N\eta_{\mathsf{th}}}$ and $p_{\mathsf{mis,th}} \geq p_{\mathsf{mis}} + \frac{1}{N}$.*

*Proof.* Since the timing constraints are respected, the checks that can fail are the checks for prover efficiency, QBER, and mismatch.

$$\begin{aligned}\Pr[\textit{checks fail}] &= \Pr[\{N_{\text{det}} < N\eta_{\mathsf{th}}\} \vee \{N_{\text{err}} + N_{\text{dc}} > N_{\text{det}}e_{\mathsf{th}}\} \vee \{N_{\text{mis}} > Np_{\mathsf{mis,th}}\}] \\ &\leq \Pr[N_{\text{det}} < N\eta_{\mathsf{th}}] + \Pr[N_{\text{mis}} > Np_{\mathsf{mis,th}}] + \Pr[N_{\text{det}} \geq N\eta_{\mathsf{th}}, N_{\text{err}} + N_{\text{dc}} > N_{\text{det}}e_{\mathsf{th}}],\end{aligned} \tag{IV.3}$$

where the second step uses the union bound and the property that $A \cup B = A \cup (A^c \cap B)$.

Using the tight concentration bound from Ref. [7], we bound the prover efficiency check failure probability by

$$\Pr[N_{\rm det} < N\eta_{\sf th}] = \Pr[N_{\rm det} \leq \lceil N\eta_{\sf th}\rceil - 1] \leq \Phi\left(-\sqrt{2NH\left(\frac{\lceil N\eta_{\sf th}\rceil}{N}, \eta\right)}\right) \leq \Phi\left(-\sqrt{2NH\left(\eta_{\sf th} + \frac{1}{N}, \eta\right)}\right), \quad \text{(IV.4)}$$

where we use the property that $\frac{\partial}{\partial x}\Phi(-\sqrt{2NH(x,p)}) \geq 0$ for $x < p$.

We apply the same concentration bound to the mismatch probability,

$$\begin{aligned}\Pr[N_{\rm mis} > Np_{\sf mis,th}] = 1 - \Pr[N_{\rm mis} \leq \lfloor Np_{\sf mis}\rfloor] &\leq 1 - \Phi\left(\sqrt{2NH\left(\frac{\lfloor Np_{\sf mis,th}\rfloor}{N}, p_{\sf mis}\right)}\right) \\ &\leq 1 - \Phi\left(\sqrt{2NH\left(p_{\sf mis,th} - \frac{1}{N}, p_{\sf mis}\right)}\right),\end{aligned} \quad \text{(IV.5)}$$

where we use the property that $\frac{\partial}{\partial x}\Phi(\sqrt{2NH(x,p)}) \geq 0$ if $x \geq p$.

Using the same concentration bound, we similarly bound the error check failure probability,

$$\begin{aligned}&\Pr[N_{\rm det} \geq N\eta_{\sf th}, N_{\rm err} + N_{\rm dc} > N_{\rm det}e_{\sf th}] \\ =& \sum_{N_{\rm det} \geq N\eta_{\sf th}} \Pr[N_{\rm det}]\Pr[N_{\rm err} + N_{\rm dc} > N_{\rm det}e_{\sf th}|N_{\rm det}] \\ =& \sum_{N_{\rm det} \geq N\eta_{\sf th}} \Pr[N_{\rm det}](1 - \Pr[N_{\rm err} + N_{\rm dc} \leq \lfloor N_{\rm det}e_{\sf th}\rfloor|N_{\rm det}]) \\ \leq& \sum_{N_{\rm det} \geq N\eta_{\sf th}} \Pr[N_{\rm det}]\left[1 - \Phi\left(\sqrt{2N_{\rm det}H\left(\frac{\lfloor N_{\rm det}e_{\sf th}\rfloor}{N_{\rm det}}, e'\right)}\right)\right] \\ \leq& \sum_{N_{\rm det} \geq N\eta_{\sf th}} \Pr[N_{\rm det}]\left[1 - \Phi\left(\sqrt{2N_{\rm det}H\left(e_{\sf th} - \frac{1}{N_{\rm det}}, e'\right)}\right)\right] \\ \leq& 1 - \Phi\left(\sqrt{2N\eta_{th}H\left(e_{\sf th} - \frac{1}{N\eta_{\sf th}}, e'\right)}\right),\end{aligned} \quad \text{(IV.6)}$$

where the fourth line applies the tight concentration bound [7]. The fifth line uses the property that $\frac{\partial}{\partial x}\Phi(\sqrt{2NH(x,e')}) \geq 0$ if $x \geq e'$. In the sixth line, we use the property that $\frac{\partial}{\partial N}\Phi(\sqrt{2NH(x-1/N,e')}) \geq 0$ when $x - \frac{1}{N} \geq e'$ and upper bound each term in the summation by the worst case term at $N_{\rm det} = N\eta_{\sf th}$. □

### C. Soundness

For simplicity, we first prove the soundness of the protocol against a more restricted QPT adversary, which is only allowed to perform independent and identical (i.i.d.) attacks across all rounds. We label the set of strategies available for the adversary as $\mathcal{A}_{QPT,iid}$. We will further extend the proof to establish soundness against general QPT adversaries in Sec. V.

As discussed in the Methods section, the proof can be split into five steps, which are presented in the respective subsections below. The first step of universal squashing reduces the measurement device from the practical model to a lossless qubit measurement device, and restricts the adversary's strategy to one where quantum system $V$ is prepared with no photons or with a single photon (single qubit). We label this set of adversary strategies $\mathcal{A}_q$, and the set of measurements $\mathcal{M}_{\rm q} = \{|\bot\rangle\langle\bot|, |\psi_0^\alpha\rangle\langle\psi_0^\alpha|, |\psi_1^\alpha\rangle\langle\psi_1^\alpha|\}_\alpha$, where $|\bot\rangle\langle\bot|$ detects the zero-photon event and $|\psi_0^\alpha\rangle$ and $|\psi_1^\alpha\rangle$ are the qubit basis states for measurement basis $\alpha$. The second step reduces the set of QPT strategies to a set of restricted strategies $\mathcal{A}_{res}$ under the QROM. $\mathcal{A}_{res}$ describes a set of strategies where the adversary prepares a shared state, $\rho_{VAB}$ and measures subsystems $A$ and $B$ independently to generate responses $z_0$ and $z_1$. The third step further restricts $\mathcal{A}_{res,q}$ to strategies where the adversary always prepares a single qubit in quantum system $V$, and we label this set of adversary strategies $\mathcal{A}_{res,nl}$. The fourth step relates single-round quantities of the restricted adversary, showing that at high prover efficiency and low mismatch, the QBER must be high. This relation is then used in the fifth step to prove the soundness against adversaries performing restricted strategies.

Recall from the main text that

**Theorem IV.2.** *Protocol 1 is $\varepsilon_{\rm sou}$-sound against i.i.d. query-bounded adversaries in the quantum random oracle model, i.e.*

$$\max_{\mathcal{S}\in\mathcal{A}_{QPT,iid},\mathcal{M}_{\rm prac}} \Pr[checks\,pass] \leq \varepsilon_{\rm sou}, \tag{IV.7}$$

*where*

$$\begin{aligned}\varepsilon_{\rm sou} = \varepsilon_{ROM} + \max\Bigg\{ & 1 - \Phi\left(\sqrt{2NH\left(\eta_{\sf th} - \frac{1}{N}, \eta_{\sf b}\right)}\right), \Phi\left(-\sqrt{2NH\left(p_{\sf mis,th} + \frac{1}{N}, p_{\sf mis,b}\right)}\right), \\ & \Phi\left(-\sqrt{2NH\left(e_{\sf th} + \frac{1}{N}, e_{SDP}\right)}\right)\Bigg\}\end{aligned} \tag{IV.8}$$

*for any $\eta_{\sf b} \leq \eta_{\sf th} - \frac{1}{N}$ and $p_{\sf mis,b} \geq p_{\sf mis,th} + \frac{1}{N}$, and where $e_{SDP}(\eta_{\sf b}, p_{\sf mis,b})$ is the solution to the SDP in Eqn. (IV.19), with the added condition that $e_{SDP} \geq e_{\sf th} + \frac{1}{N\eta_{\sf th}}$.*

*Proof.* It is sufficient to prove that

$$\max_{\mathcal{S}\in\mathcal{A}_{res,iid,nl},\mathcal{M}_{\rm q}} \Pr[checks\,pass] \leq \varepsilon - \varepsilon_{ROM} \tag{IV.9}$$

since Thms. IV.3 to IV.5 guarantee the $\varepsilon_{\rm sou}$-soundness.

$$\begin{aligned}&\max_{\mathcal{S}\in\mathcal{A}_{res,iid,nl},\mathcal{M}_{\rm q}} \Pr[checks\,pass] \leq \varepsilon_{\rm sou} - \varepsilon_{ROM} \xrightarrow{Thm.\ IV.5} \max_{\mathcal{S}\in\mathcal{A}_{res,iid,q},\mathcal{M}_{\rm q}} \Pr[checks\,pass] \leq \varepsilon_{\rm sou} - \varepsilon_{ROM} \\ &\xrightarrow{Thm.\ IV.4} \max_{\mathcal{S}\in\mathcal{A}_{QPT,iid,q},\mathcal{M}_{\rm q}} \Pr[checks\,pass] \leq \varepsilon_{\rm sou} \xrightarrow{Thm.\ IV.3} \max_{\mathcal{S}\in\mathcal{A}_{QPT,iid},\mathcal{M}_{\rm prac}} \Pr[checks\,pass] \leq \varepsilon_{\rm sou}.\end{aligned} \tag{IV.10}$$

Applying Thm. IV.7 completes the proof. □

We note that the added condition $e_{SDP} \geq e_{\sf th} + \frac{1}{N\eta_{\sf th}}$ can be checked when performing numerical analysis, and it is not violated in the regime of study.

### 1. *Universal Squashing*

The first reduction concerns the imperfection of the measurement device, which can cause a discrepancy between the ideal qubit measurement described by the protocol and the actual measurement device in the presence of multi-photon events. Multi-photon events are not unique to LT-QPV, and standard prepare-and-measure QPV protocols face similar issues where an imperfect source at the verifier may generate multiple photons. To account for this imperfection, we consider a realistic model for the verifier's imperfect detector shown in Fig. 3a. Ref. [8] demonstrated that using a universal squashing model, one can relate the practical measurement device to an ideal one with a squashing map shown in Fig. 3b. Specifically, the single click events in both models match and double-click events in the practical measurement device can be assigned to the 0/1 outcomes of the ideal measurement device to obtain a bound on the ideal measurement device statistics. For all other events, such as prover mismatch or no detection by the prover, the verifier's 0/1 outcome does not impact the assignment.

We summarize this reduction as Thm. IV.3:

**Theorem IV.3.** *Suppose protocol 1 is $\varepsilon$-sound with an ideal measurement device. Then, protocol 1 is $\varepsilon$-sound with a practical measurement device, i.e. for any set of adversary strategies $\mathcal{A}$,*

$$\max_{\mathcal{S}\in\mathcal{A}_q,\mathcal{M}_{\rm q}} \Pr[checks\,pass] \leq \varepsilon \implies \max_{\mathcal{S}\in\mathcal{A},\mathcal{M}_{\rm prac}} \Pr[checks\,pass] \leq \varepsilon, \tag{IV.11}$$

*where $\mathcal{A}_q$ is the corresponding set of adversary strategies with channel $\mathcal{E}^{x_i}_{A_i}$ restricted to generate zero or single-photon states in quantum system $V_i$.*

*Proof.* Since we assume that the loss on the verifier-side is basis-independent and balanced, we can model the measurement device loss as a beam-splitter with transmission $\eta_V$ placed before a lossless measurement device. We focus on the lossless measurement device model shown in Fig. 3 for the remainder of this proof.

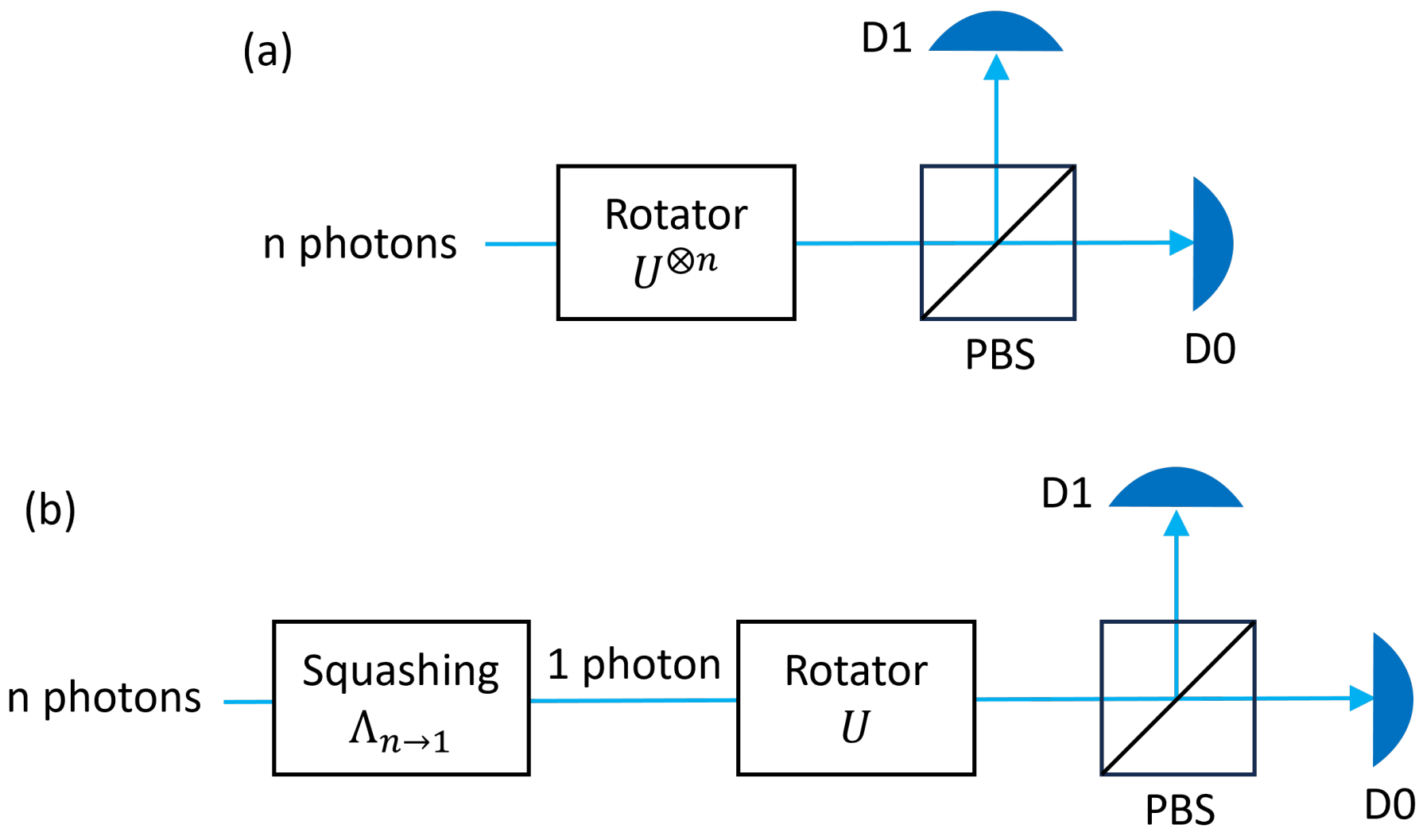


FIG. 3. Comparison between the verifier's (a) practical measurement device and (b) ideal single qubit measurement device with squashing map. In the practical measurement device, the rotation $U$ to the correct measurement basis is applied to all incoming photons. In the ideal measurement device, it accepts only a single (or zero) photon and behaves like an ideal single qubit measurement. This single photon input is enforced by a squashing map $\Lambda_{n\to 1}$, which randomly selects one of the $n$ photons to transmit. $D0$ and $D1$: threshold detectors, PBS: Polarizing beam splitter.

Suppose there exists an adversary strategy $\mathcal{S}$ for the practical measurement device scenario that passes the following checks with probability $\varepsilon'$: (1) $N_{\text{det}} \geq N\eta_{\text{th}}$, (2) $N_{\text{dc}} + N_{\text{err}} \leq N_{\text{det}} e_{\text{th}}$ and (3) $z_{0,i} = z_{1,i}, \forall i \in [N]$. Let us consider when the same strategy $\mathcal{S}$ is utilized with the ideal measurement device and squashing map, and let $N$, $N'_{\text{det}}$, $N'_{\text{dc}}$ and $N'_{\text{err}}$ be the statistics. By universal squashing [8], the outcomes in the practical measurement device can be matched to outcome assignment in the ideal measurement device, with

1. If neither D0 nor D1 clicks (no photons), then neither would click in the ideal measurement device.
2. If D0 clicks and D1 does not, then D0 clicks and D1 does not in the ideal measurement device.
3. If D1 clicks and D0 does not, then D1 clicks and D0 does not in the ideal measurement device.
4. If both D0 and D1 click (double-click event), then either D0 or D1 clicks in the ideal measurement device.

Since detection events $N$ and $N_{\text{det}}$ count events where a click is detected at the measurement device (with differing prover response), the number of such events remains the same in the ideal measurement device, $N' = N$ and $N'_{\text{det}} = N_{\text{det}}$. Since the ideal measurement device only accepts a single-photon after the squashing map, there are no double-click events, $N'_{\text{dc}} = 0$. We note that there can be double-click events when no detection occurs at the prover, $z_{0,i} = \perp$, but in both models, these events are considered to be part of prover loss.

Error events are events where the measurement device outcome does not match the prover's response, $s_i \neq z_{0,i}$ and $z_{0,i}, s_i \neq \perp$. For the idealized model, this includes both the error from single click events, and some contribution from double-click events. Consequently, we bound the number of error events $N_{\text{err}} \leq N'_{\text{err}} \leq N_{\text{err}} + N_{\text{dc}}$ since double-clicks can either be assigned to errors in the worst case or correct responses in the optimistic case. This implies that when $N_{\text{err}} + N_{\text{dc}} \leq N_{\text{det}} e_{\text{th}}$, then $N'_{\text{err}} + N'_{\text{dc}} \leq N_{\text{det}} e_{\text{th}}$. Combining the earlier results, if the strategy $\mathcal{S}$ passes the checks with the practical measurement device with probability $\varepsilon'$, then it will also pass the checks with the ideal measurement device and squashing map with probability at least $\varepsilon'$.

Without loss of generality, we consider the squashing map $\Lambda_{n\to 1}$ and the detector loss ($\eta_V$ beam-splitter) to be implemented by the adversary, i.e. when the adversary generates the state $\rho_{V_i A_K^j A_S^j B_K^j B_S^j}$, it applies the beam-splitter $U_{\eta_V}$ and $\Lambda_{n\to 1}$ to $V_i$ before sending the state to $V_0$. The verifier's measurement is the ideal qubit measurement (and zero-photon measurement) $\mathcal{M}_{\text{q}} = \{|\perp\rangle\langle\perp|, |\psi_0^\alpha\rangle\langle\psi_0^\alpha|, |\psi_1^\alpha\rangle\langle\psi_1^\alpha|\}_\alpha$, where $\alpha = f(x, y)$ is the basis choice. We define $\mathcal{A}_q$ as the set of adversary strategies where the channel $\mathcal{E}_{A_i}^{x_i}$ is restricted to generated only single or zero-photon state in quantum system $V_i$. Since the state that is sent to the ideal qubit measurement can only be a single or zero-photon state (imposed by $\Lambda_{n\to 1}$), any strategy $\mathcal{S} \in \mathcal{A}$ which implements the beam-splitter and squashing map can be mapped to an equivalent strategy in $\mathcal{A}_q$. The earlier discussion implies that for any strategy $\mathcal{S} \in \mathcal{A}$ with verifier measurement

$\mathcal{M}_{\rm prac}$ that can pass the checks with probability $\varepsilon'$, there exists a strategy $\mathcal{S}' \in \mathcal{A}_q$ with verifier measurement $\mathcal{M}_{\rm q}$ can pass the checks with probability at least $\varepsilon'$. If $\varepsilon' > \varepsilon$, it contradicts the presumption that protocol 1 is $\varepsilon$-sound with an ideal measurement device, and thus, there cannot be such a strategy $\mathcal{S} \in \mathcal{A}$ that passes the checks with probability $\varepsilon' > \varepsilon$ with a practical measurement device. As such, protocol 1 is $\varepsilon$-sound with a practical measurement device. □

### 2. *Random Oracle Model*

The QROM, along with a suitable choice of hash function $f$, reduces the set of QPT strategies $\mathcal{A}_{QPT}$ to a corresponding set of restricted strategies $\mathcal{A}_{res}$, where channels $\mathcal{E}^{x^j}_{A^j}$ and $\mathcal{E}^{y^j}_{B^j}$ are independent of $x^j$ and $y^j$, and the POVMs $A^{x^j,y^j}_{z_0^j} = A^{\alpha^j}_{z_0^j}$ and $B^{\alpha^j}_{z_1^j}$ are dependent only on the basis choice $\alpha^j$. Since the restricted strategy is independent of $x^j$ and $y^j$ except through $\alpha^j$, we combine Alice's and Bob's operations with the preparation of $\rho_{A^jB^j}$, and summarize the adversary strategy as

1. The adversary generates a shared state $\rho_{V^jA^jB^j} = \widetilde{\mathcal{E}}^j(\rho_{E^{j-1}})$.
2. When Alice receives $y^j$ and Bob receives $x^j$, Alice performs a measurement on her subsystem $A^j$ with POVM $A^{\alpha^j}_{z_{0,i}}$ while Bob performs a measurement on his subsystem $B^j$ with POVM $B^{\alpha^j}_{z_1^j}$. Alice and Bob send $z_0^j$ and $z_1^j$ to verifiers $V_0$ and $V_1$ respectively.

Formally, we show that

**Theorem IV.4.** *In the quantum random oracle model, if Protocol 1 is $\varepsilon$-sound against adversaries performing a restricted strategy, then Protocol 1 is $(\varepsilon + \varepsilon_{ROM})$-sound against QPT adversaries, where $\varepsilon_{ROM}$ is negligible in the security parameter $n$, i.e. for any set of QPT adversary strategies $\mathcal{A}_{QPT}$,*

$$\max_{\mathcal{S}\in\mathcal{A}_{res},\mathcal{M}} \Pr[checks\,pass] \le \varepsilon \implies \max_{\mathcal{S}\in\mathcal{A}_{QPT},\mathcal{M}} \Pr[checks\,pass] \le \varepsilon + \varepsilon_{ROM}, \tag{IV.12}$$

*where $\mathcal{A}_{res}$ is the corresponding set of restricted strategies.*

*Proof.* We will prove this via a sequence of hybrids. We set $\varepsilon_{ROM} = 2qN_{\max}2^{-n} + N_{\max}^2/2^{2n+1}$, where $q$ is an upper bound on the total number of oracle queries that can be performed by an adversary.

- **Hybrid 0:** This is the original QPV soundness experiment, where the adversaries are allowed to perform any QPT strategy in $\mathcal{A}_{QPT}$.
- **Hybrid 1:** This is the same as **Hybrid 0**, except the challenges $(x^j, y^j)$ are sampled such that $(x^j, y^j) \neq (x^k, y^k)$ for $j \neq k$.
- **Hybrid 2.0:** In this hybrid, Alice and Bob are given additional information. In every round $j$, Alice learns $\alpha^j$ at time $t_P^j + d_{V_0P}/c$ and Bob learns $\alpha^j$ at time $t_P^j + d_{V_1P}/c$.
- **Hybrids 2.1 to 2.$N_{\max}$:** For $j \in [N_{\max}]$, Hybrid $j$ is the same as **Hybrid** $j-1$, except the random oracle $f$ is reprogrammed on input $(x^j, y^j)$ to satisfy $f(x^j, y^j) = u^j$, where $u^j \in [m]$ is a uniformly random number independent from the rest of the protocol.
- **Hybrid 3:** In this hybrid, Alice and Bob only receive $\alpha^j$ at times $(t_P^j + d_{V_0P}/c, t_P^j + d_{V_1P}/c)$ as before, but they do not receive values $(x^j, y^j)$, and there is no random oracle.

Let $p_i$ be the maximum probability of success that an adversary can achieve in **Hybrid** $i$, i.e. the maximum probability of passing the checks in QPV. We will reach a contradiction through the following sequence of claims:

**Claim IV.1.** $p_1 \ge p_0 - N_{\max}^2/2^{2n+1}$.

*Proof.* Conditioned on no collision between pairs $(x^j, y^j)$, the two hybrids are identical. Therefore, for a fixed adversary, the difference between the success probabilities $|p_1 - p_0|$ is upper-bounded by the probability that there exists $j \neq k$ with $(x^j, y^j) = (x^k, y^k)$, which is at most $N_{\max}^2/2^{2n+1}$ by the birthday bound. Note that here we used the total number of rounds $N_{\max}$ as opposed to $N$ because the adversary can affect which rounds are detected by $V_0$ based on the value $x^j$. □

**Claim IV.2.** $p_{2.0} \ge p_1$.

*Proof.* Since the adversary receives more information, the optimal success probability cannot decrease. □

**Claim IV.3.** *$p_{2.j} \geq p_{2.(j-1)} - 2q2^{-n}$ for $j \in [N_{\max}]$.*

*Proof.* Let $j \in [N_{\max}]$ and $(\mathcal{A}, \mathcal{B})$ be an adversary that succeeds in **Hybrid 2.**$(j-1)$ with probability $p_{2.(j-1)}$, with random oracle $f^{j-1}$. Let $f^j$ be the random oracle in **Hybrid 2.**$j$, so that $f^j(x^j, y^j) = u^j$ and $f^j(x', y') = f^{j-1}(x', y')$ for $(x', y') \neq (x^j, y^j)$. We will describe a reduction $(\mathcal{A}', \mathcal{B}')$ for **Hybrid 2.**$j$ as follows:

- At times $t < t_P^j + d_{V_0 P}/c$, $\mathcal{A}'$ runs $\mathcal{A}$ using $f^j$ as the oracle. Similarly, at times $t < t_P^j + d_{V_1 P}/c$, $\mathcal{B}'$ runs $\mathcal{B}$ using $f^j$ as the oracle.
- At times $t \geq t_P^j + d_{V_0 P}/c$, $\mathcal{A}'$ can simulate $f^{j-1}$ using $(x^j, y^j, f^j)$ by reprogramming $f^j$ to output $\alpha^j$ on input $(x^j, y^j)$, so $\mathcal{A}'$ can run $\mathcal{A}$ with the correct oracle $f^{j-1}$. Similarly, $\mathcal{B}'$ can run $\mathcal{B}$ with the correct oracle $f^{j-1}$ at times $t \geq t_P^j + d_{V_1 P}/c$.
- $\mathcal{A}'$ (resp. $\mathcal{B}'$) outputs what $\mathcal{A}$ (resp. $\mathcal{B}$) outputs, responding to the same verifier.

Note that the view of $(\mathcal{A}, \mathcal{B})$ as simulated by $(\mathcal{A}', \mathcal{B}')$ differs from **Hybrid 2.**$(j-1)$ at times $(t < t_P^j + d_{V_0 P}/c, t < t_P^j + d_{V_1 P}/c)$, where the first time is for $\mathcal{A}$ and the second time is for $\mathcal{B}$. Additionally, the oracles only differ on input $(x^j, y^j)$. Recall that the adversary makes at most $q$ total queries to the oracle. Therefore, by Lemma II.1, the total query weight of $(\mathcal{A}, \mathcal{B})$ on input $(x^j, y^j)$ at times $(t < t_P^j + d_{V_0 P}/c, t < t_P^j + d_{V_1 P}/c)$ is at least $4(p_{2.(j-1)} - p_{2.j})^2/q$. On the other hand, since at these times $y^j$ is information-theoretically hidden from $\mathcal{A}$ and $x^j$ is information-theoretically hidden from $\mathcal{B}$, it implies that the query weight is at most $q2^{-2n}$. Thus, we have $\left|p_{2.(j-1)} - p_{2.j}\right| \leq 2q2^{-n}$ as desired. □

**Claim IV.4.** *$p_3 \geq p_{2.N_{\max}}$.*

*Proof.* Let $(\mathcal{A}, \mathcal{B})$ be an adversary for **Hybrid 2.**$N_{\max}$ that succeeds with probability $p_{2.N_{\max}}$. We give a reduction $(\mathcal{A}', \mathcal{B}')$ that succeeds in **Hybrid 3** with the same probability:

- At time $t = -\infty$ (i.e. before the protocol begins), $\mathcal{A}'$ and $\mathcal{B}'$ sample a $2q$-wise independent hash function $f'$. In addition, they sample distinct pairs $(x^j, y^j) \in \{0,1\}^n \times \{0,1\}^n$ for $j \in [N_{\max}]$. Then, they reprogram the function $f'$ on inputs $(x^j, y^j)$ so that $f'(x^j, y^j) = \alpha^j$.
- $\mathcal{A}'$ simulates $\mathcal{A}$ and $\mathcal{B}'$ simulates $\mathcal{B}$, using the reprogrammed $f'$ to answer oracle queries and faithfully forwarding any messages output by $(\mathcal{A}, \mathcal{B})$ to the corresponding verifier.

By Lemma II.2, the view of $(\mathcal{A}, \mathcal{B})$ is identical to **Hybrid 2.**$N_{\max}$, hence the checks will pass with probability $p_{2.N_{\max}}$ as desired. □

Note that **Hybrid 0** captures unrestricted strategies and **Hybrid 3** captures the restricted strategies available to the adversary. Therefore, there exists a restricted strategy that succeeds with probability $p_3$, i.e. $\varepsilon \geq p_3$. Combining with Claims IV.1 to IV.4, we have

$$p_0 \leq p_3 + 2qN_{\max}2^{-n} + N_{\max}^2/2^{2n+1} \leq \varepsilon + \varepsilon_{ROM},$$

which suffices for the proof. □

**Remark IV.1** (Deciding the security parameter)**.** *The security parameter (i.e. the hash input length) $n$ should be set according to the computational power of the adversary, measured by the total number of hash function evaluations ($q$) that the adversary can reasonably perform. It should satisfy $2^n \gg q^2$. If the hash function is picked during the protocol, then the number of queries is estimated based on the protocol duration. If an unkeyed hash function is used, then $n$ should be chosen to be larger to protect against pre-processing attacks.*

### 3. Channel Loss Independence

Consider the restricted adversary model with the restriction that $V^j$ is zero or single-photon, $\mathcal{A}_{res,q}$. When generating $\rho_{V^jA^jB^j}$, the adversary has to make a decision about whether a round is detected (by the choice of preparation of $V^j$) independent of basis $\alpha^j$. Intuitively, non-detected rounds provide no advantage to the adversary. As such, it is sufficient to consider an adversary that does not introduce any loss and for a protocol in which exactly $N$ rounds are performed. Let us define the set of adversary strategies with no verifier-side loss as $\mathcal{A}_{res,nl}$, i.e. with the additional imposition that the channel $\widetilde{\mathcal{E}}^j$ generates a qubit $V^j$ (single photon).

**Theorem IV.5.** *Suppose that Protocol 1 is $\varepsilon$-sound against a restricted adversary that prepares qubit $V^j$ with ideal verifier measurement devices. Then, Protocol 1 is $\varepsilon$-sound against a restricted adversary that prepares a single-photon or vacuum state in $V^j$ with ideal verifier measurement devices, i.e.*

$$\max_{\mathcal{S}\in\mathcal{A}_{res,nl},\mathcal{M}_{\mathrm{q}}} \Pr[\mathit{checks\,pass}] \leq \varepsilon \implies \max_{\mathcal{S}\in\mathcal{A}_{res,q},\mathcal{M}_{\mathrm{q}}} \Pr[\mathit{checks\,pass}] \leq \varepsilon \tag{IV.13}$$

*Proof.* Consider a restricted adversary with strategy $\mathcal{S} \in \mathcal{A}_{res,q}$ that can pass protocol 1 with probability $\varepsilon$. This adversary prepares the quantum state $\rho_{V^jA^jB^j} = \widetilde{\mathcal{E}}^j(\rho_{E^{j-1}})$, where the adversary is free to decide whether to prepare $V^j$ with a qubit (recall the adversary holds the squashing map) or with no qubit (i.e. no detection). The protocol round then proceeds with the remaining steps, including the measurements of both the adversary and the verifier, which we express as a quantum channel $\mathcal{C}^j$, performed by both the adversary and verifiers. We note that steps in $\mathcal{C}^j$ performed by the verifier, such as the choice of messages $x^j$ and $y^j$, are independent across rounds and random, thereby allowing for the adversary to simulate these steps internally. Note that the adversary's actions, such as the choice of measurement, can differ between rounds.

In the protocol, the statistics are only generated from the first $N$ rounds where a detection is registered, $s_i \neq \perp$. We summarize the process of the entire protocol as:

1. The adversary prepares $\rho_{E_0}$.
2. Initialize $i = 1$ and $j = 1$. Repeat until $i = N$:
   (a) The adversary generates $\rho_{V^jA^jB^j} = \widetilde{\mathcal{E}}^j(\rho_{E_{j-1}})$.
   (b) The verifiers and adversary proceed with the protocol by applying map $\mathcal{C}^j$ on the prepared state, $\rho_{S^jZ_A^jZ_B^jA^{j\prime}B^{j\prime}} = \mathcal{C}^j(\rho_{V^jA^jB^j})$.
   (c) The verifiers record $S^j$, $Z_A^j$ and $Z_B^j$ as the outcome of round $i$ if $S^j \neq\perp$.
   (d) Increment $j = j + 1$. If $S^j \neq\perp$, increment $i$ by 1.

Let us consider a modified adversary strategy with a different process:

1. The adversary prepares $\rho_{E_0}$.
2. Initialize $i = 1$ and $j = 1$. Repeat for $N$ rounds:
   (a) The adversary generates $\rho_{V^jA^jB^j} = \widetilde{\mathcal{E}}_i(\rho_{E^{j-1}})$.
   (b) The adversary performs a measurement to determine if a qubit is present in quantum system $V^j$ using $\{|\perp\rangle\langle\perp|, |\top\rangle\langle\top|\}$, where $|\top\rangle\langle\top| = |\psi_0^\alpha\rangle\langle\psi_0^\alpha| + |\psi_1^\alpha\rangle\langle\psi_1^\alpha|$ is basis-independent by assumption. Note that the above measurement does not impact the measurement statistics since the measurement operators commute with the verifier's measurement $\{|\perp\rangle\langle\perp|, |\psi_0^\alpha\rangle\langle\psi_0^\alpha|, |\psi_1^\alpha\rangle\langle\psi_1^\alpha|\}$.
   (c) If a photon is not detected, the adversary simulates the channel $\mathcal{C}^j$ internally to generate a post-selected state $\rho_{E^j}$. Increment $j$ by 1 and restart from step (a).
   (d) If a photon is present, the verifier and adversary proceed with the protocol, i.e. apply map $\mathcal{C}^j$ on the prepared state to generate statistics $S^j$, $Z_A^j$ and $Z_B^j$, and post-selected state $\rho_{E^j}$.
   (e) The verifier record $S^j$, $Z_A^j$ and $Z_B^j$ as the outcome of round $i$. Increment $i$ by 1.

By construction, the modified strategy will always give a verifier measurement outcome $S_i = S^j \neq\perp$, and the verifier observes exactly $N$ rounds. The two processes mainly differ in where the decision to add a round into the statistics occurs. In the former, this is performed by the verifier if $S^j \neq\perp$, while in the latter, the adversary chooses to proceed to

the next index $j$ if it determines that $V^j$ contains no qubit. Since the processes match, the statistics generated by both processes remain the same, and thus both processes yield an equal probability of passing the checks. Consequently, for any restricted adversary with a strategy that can pass protocol 1 for $N$ detected rounds with probability greater than $\varepsilon$, we can construct a restricted adversary with a strategy for protocol 1 for $N$ rounds with no verifier-side loss that can pass with probability greater than $\varepsilon$. The theorem follows by contradiction. If there exists a strategy $\mathcal{S} \in \mathcal{A}_{res,q}$ that passes protocol 1 with probability exceeding $\varepsilon$, the constructed strategy in $\mathcal{A}_{res,nl}$ would also pass protocol 1 with probability exceeding $\varepsilon$, contradicting the assumption. □

### *4. Single-Round Analysis*

Now that we restrict our attention to adversarial strategies in $\mathcal{A}_{res,nl}$, we will switch to indexing rounds by $i \in [N]$ and displaying the round number $i$ as a subscript (see caption of Table I).

Consider the set of restricted adversary strategies $\mathcal{A}_{res,iid,nl}$ with qubit preparation for $V_i$ and i.i.d. attacks. Since the attacks are independent and identical, we can define expectation values for prover efficiency $\eta$, probability of mismatched responses $(z_{0,i} \neq z_{1,i})$ $p_{\mathsf{mis}}$, and QBER $e$, which takes identical values across rounds. We characterize the relation between these three parameters for a single round before proving security.

**Theorem IV.6.** *Consider a single round of Protocol 1 with an adversary strategy $\mathcal{S} \in \mathcal{A}_{res,iid,nl}$ and ideal verifier measurement $\mathcal{M}_{\mathsf{q}}$. Let $\eta$, $p_{\mathsf{mis}}$ and $e$ be the prover efficiency, mismatch probability, and QBER of the round respectively. If $\eta \geq \eta_{\mathsf{b}}$ and $p_{\mathsf{mis}} \leq p_{\mathsf{mis,b}}$, then the effective QBER $e \geq e_{SDP}$, where $e_{SDP}(\eta_{\mathsf{b}}, p_{\mathsf{mis,b}})$ is the solution to the SDP in Eqn.* (IV.19).

*Proof.* For a single round of the restricted adversary strategy, the statistics are given by

$$\Pr[s, z_0, z_1] = \frac{1}{m} \sum_{\alpha} \mathrm{Tr}\big[(\Pi_s^\alpha \otimes A_{z_0}^\alpha \otimes B_{z_1}^\alpha)\rho_{VAB}\big], \tag{IV.14}$$

where $V$ is a qubit system, and we drop the index $i$ for brevity. Without loss of generality, by state purification and Stinespring's dilation theorem, we can consider a pure strategy with a pure state $|\psi\rangle_{VAB}$ where the purification system can be kept with $A$ to keep $dim(V) = 2$, and projective measurements $\{\tilde{A}_{z_0}^\alpha\}$ and $\{\tilde{B}_{z_1}^\alpha\}$ corresponding to the POVMs. The statistics remain unchanged,

$$\Pr[s, z_0, z_1] = \frac{1}{m} \sum_{\alpha} \langle\psi|\Pi_s^\alpha \otimes \tilde{A}_{z_0}^\alpha \otimes \tilde{B}_{z_1}^\alpha|\psi\rangle\,. \tag{IV.15}$$

The pure state can be expressed as

$$|\psi\rangle_{VAB} = |0\rangle_V |\phi_0\rangle_{AB} + |1\rangle_V |\phi_1\rangle_{AB}\,, \tag{IV.16}$$

where $|\phi_i\rangle_{AB}$ can be sub-normalized. When ideal qubit measurements are performed by the verifier, the state held by Alice and Bob collapses to one of the sub-normalized states depending on measurement outcomes $s = 0, 1$,

$$\begin{aligned} |\phi_0^\alpha\rangle_{AB} &= \cos\left(\frac{\theta_\alpha}{2}\right)|\phi_0\rangle_{AB} + \sin\left(\frac{\theta_\alpha}{2}\right)e^{-i\phi_\alpha}|\phi_1\rangle_{AB} \\ |\phi_1^\alpha\rangle_{AB} &= \sin\left(\frac{\theta_\alpha}{2}\right)|\phi_0\rangle_{AB} - \cos\left(\frac{\theta_\alpha}{2}\right)e^{-i\phi_\alpha}|\phi_1\rangle_{AB}\,. \end{aligned} \tag{IV.17}$$

The minimum QBER attainable by any adversary $\mathcal{S} \in \mathcal{A}_{res.iid.nl}$ is expressed as an optimization problem

$$\begin{aligned} \min_{|\phi_0\rangle,|\phi_1\rangle,\{\tilde{A}_{z_0}^\alpha\}_{z_0},\{\tilde{B}_{z_1}^\alpha\}_{z_1}} \quad & \frac{1}{m\eta}\sum_{\alpha=1}^{m} \langle\phi_0^\alpha|\tilde{A}_1^\alpha \otimes \mathbb{I}_B|\phi_0^\alpha\rangle + \langle\phi_1^\alpha|\tilde{A}_0^\alpha \otimes \mathbb{I}_B|\phi_1^\alpha\rangle \\ \textit{subj. to} \quad & \frac{1}{m}\sum_{\alpha=1}^{m}\sum_{j=0}^{1} \langle\phi_j^\alpha|\tilde{A}_\perp^\alpha \otimes \mathbb{I}_B|\phi_j^\alpha\rangle = 1-\eta \\ & \frac{1}{m}\sum_{\alpha=1}^{m}\sum_{j=0}^{1}\sum_{z_0\neq z_1} \langle\phi_j^\alpha|\tilde{A}_{z_0}^\alpha \otimes \tilde{B}_{z_1}^\alpha|\phi_j^\alpha\rangle \leq p_{\mathsf{mis,b}} \\ & \langle\phi_0^\alpha|\phi_0^\alpha\rangle + \langle\phi_1^\alpha|\phi_1^\alpha\rangle = 1,\ \forall\alpha \in \{1,\cdots,m\} \\ & \eta \geq \eta_{\mathsf{b}} \end{aligned}\,. \tag{IV.18}$$

For any strategy with prover efficiency $\eta > \eta_{\mathsf{b}}$, we can construct an equivalent strategy at $\eta_{\mathsf{b}} = \eta$ with the same QBER by randomly setting responses $z_0$ and $z_1$ to $\perp$. It is therefore sufficient to optimize at $\eta = \eta_b$. The optimization problem can be relaxed into an SDP following Ref. [9],

$$\begin{aligned}
\min_{G\geq 0} \quad & \frac{1}{m\eta_b}\sum_{\alpha=1}^{m} \langle\phi_0^\alpha|\tilde{A}_1^\alpha|\phi_0^\alpha\rangle + \langle\phi_1^\alpha|\tilde{A}_0^\alpha|\phi_1^\alpha\rangle \\
\textit{subj. to} \quad & \frac{1}{m}\sum_{\alpha=1}^{m}\sum_{j=0}^{1} \langle\phi_j^\alpha|\tilde{A}_\perp^\alpha|\phi_j^\alpha\rangle = 1-\eta_b \\
& \frac{1}{m}\sum_{\alpha=1}^{m}\sum_{j=0}^{1}\sum_{z_0\neq z_1} \langle\phi_j^\alpha|\tilde{A}_{z_0}^\alpha\tilde{B}_{z_1}^\alpha|\phi_j^\alpha\rangle \leq p_{mis,b} \\
& \langle\phi_0|\phi_0\rangle + \langle\phi_1|\phi_1\rangle = 1 \\
& [\tilde{A}_{z_0}^\alpha, \tilde{B}_{z_1}^{\alpha'}] = 0,\ \forall\alpha,\alpha',z_0,z_1,
\end{aligned} \tag{IV.19}$$

where $G$ is a Gram matrix with elements $G_{(j,O),(j',O')} = \langle\phi_j|OO'|\phi_j'\rangle$, and $O$ refers to operators formed from combinations of $\mathbb{I}$, $\{\tilde{A}_{z_0}^\alpha\}$, and $\{\tilde{B}_{z_1}^\alpha\}$. □

### 5. *Finite-Size Analysis*

The relation between the single-round expectation values constrains the adversary strategy. Since these expectation values are the same across rounds, we relate them to observed values using concentration bounds. This in turn allows for a comparison between the observed and the threshold values to compute the soundness parameter.

**Theorem IV.7.** *Protocol 1 with an ideal verifier measurement device is $\varepsilon$-sound against any restricted adversary that prepares $V_i$ as a qubit state, i.e.*

$$\max_{\mathcal{S}\in\mathcal{A}_{res,iid,nl},\mathcal{M}_{\mathsf{q}}} \Pr[\textit{checks pass}] \leq \varepsilon, \tag{IV.20}$$

*where*

$$\varepsilon = \max\left\{1-\Phi\left(\sqrt{2NH\left(\eta_{\mathsf{th}}-\frac{1}{N},\eta_{\mathsf{b}}\right)}\right), \Phi\left(-\sqrt{2NH\left(p_{\mathsf{mis,th}}+\frac{1}{N},p_{\mathsf{mis,b}}\right)}\right), \Phi\left(-\sqrt{2NH\left(e_{\mathsf{th}}+\frac{1}{N},e_{SDP}\right)}\right)\right\}, \tag{IV.21}$$

*for any $\eta_{\mathsf{b}} \leq \eta_{\mathsf{th}} - \frac{1}{N}$ and $p_{\mathsf{mis,b}} \geq p_{\mathsf{mis,th}} + \frac{1}{N}$, and where $e_{SDP}(\eta_{\mathsf{b}}, p_{\mathsf{mis,b}})$ is the solution to the SDP in Eqn. (IV.19), with the added condition that $e_{SDP} \geq e_{\mathsf{th}} + \frac{1}{N\eta_{\mathsf{th}}}$.*

*Proof.* To prove security, we show that the adversary cannot simultaneously pass all checks with high probability,

$$\Pr[\textit{checks pass}] \leq \Pr[N_{\text{det}} \geq N\eta_{\mathsf{th}}, N_{\text{err}} \leq Ne_{\mathsf{th}}, N_{\text{mis}} \leq Np_{\mathsf{mis,th}}]. \tag{IV.22}$$

We split the set of adversary strategies $\mathcal{A}_{res,iid,nl}$ into three subsets: (1) $\mathcal{A}_{res,iid,nl}^{<\eta_{\mathsf{b}}}$ where the single-round expected prover efficiency is $\eta < \eta_{\mathsf{b}} < \eta_{\mathsf{th}}$, (2) $\mathcal{A}_{res,iid,nl}^{>p_{\mathsf{mis,b}}}$ where the single-round mismatch probability is high, $p_{\mathsf{mis}} > p_{\mathsf{mis,b}} > 0$, and (3) $\mathcal{A}_{res,iid,nl}^{\geq\eta_{\mathsf{b}},\leq p_{\mathsf{mis,b}}}$ where the single-round expected prover efficiency is high and the mismatch probability is low, $\eta \geq \eta_{\mathsf{b}}$ and $p_{\mathsf{mis}} \leq p_{\mathsf{mis,b}}$.

For the first subset, $\mathcal{A}_{res,iid,nl}^{<\eta_{\mathsf{b}}}$, the prover efficiency check fails with high probability,

$$\begin{aligned}
\Pr[\textit{checks pass}] \leq& \Pr[N_{\text{det}} \geq N\eta_{\mathsf{th}}] \\
=& 1-\Pr[N_{det} \leq \lceil N\eta_{th}\rceil - 1] \\
\leq& 1-\Phi\left(\sqrt{2NH\left(\frac{\lceil N\eta_{\mathsf{th}}\rceil-1}{N},\eta_{\mathsf{b}}\right)}\right) \\
\leq& 1-\Phi\left(\sqrt{2NH\left(\eta_{\mathsf{th}}-\frac{1}{N},\eta_{\mathsf{b}}\right)}\right),
\end{aligned} \tag{IV.23}$$

where the third line applies the concentration bound from Ref. [7], and the final line uses the property that $\frac{\partial}{\partial x}\Phi\left(\sqrt{2NH(x,p)}\right) \geq 0$ when $x \geq p$.

For the second subset, $\mathcal{A}_{res,iid,nl}^{>p_{\mathsf{mis,b}}}$, the mismatch check fails with high probability,

$$\begin{aligned}
\Pr[checks\,pass] \leq& \Pr[N_{\mathrm{mis}} \leq N p_{\mathsf{mis,th}}] \\
=& \Pr[N_{\mathrm{mis}} \leq \lfloor N p_{\mathsf{mis,th}} \rfloor] \\
\leq& \Phi\left(-\sqrt{2NH\left(\frac{\lfloor N p_{\mathsf{mis,th}}\rfloor + 1}{N}, p_{\mathsf{mis,b}}\right)}\right) \\
\leq& \Phi\left(-\sqrt{2NH\left(p_{\mathsf{mis,th}} + \frac{1}{N}, p_{\mathsf{mis,b}}\right)}\right),
\end{aligned} \tag{IV.24}$$

where the third line applies the concentration bound from Ref. [7], and the final line uses the property that $\frac{\partial}{\partial x}\Phi\left(-\sqrt{2NH(x,p)}\right) \geq 0$ when $x \leq p$.

For the final subset, $\mathcal{A}_{res,iid,nl}^{\geq\eta_{\mathsf{b}},\leq p_{\mathsf{mis,b}}}$, given that $\eta \geq \eta_{\mathsf{b}}$ and $p_{\mathsf{mis}} \leq p_{\mathsf{mis,b}}$, Thm. IV.6 implies that the QBER satisfies $e \geq e_{SDP}(\eta_{\mathsf{b}}, p_{\mathsf{mis,b}})$. Therefore, the QBER check passes with low probability,

$$\begin{aligned}
\Pr[checks\,pass] =& \Pr[N_{\mathrm{err}} \leq N_{\mathrm{det}} e_{\mathsf{th}}, N_{\mathrm{det}} \geq N\eta_{\mathsf{th}}] \\
=& \sum_{N_{\mathrm{det}} \geq N\eta_{\mathsf{th}}} \Pr[N_{\mathrm{det}}] \times \Pr[N_{\mathrm{err}} \leq N_{\mathrm{det}} e_{\mathsf{th}}] \\
\leq& \sum_{N_{\mathrm{det}} \geq N\eta_{\mathsf{th}}} \Pr[N_{\mathrm{det}}] \times \Phi\left(-\sqrt{2N_{\mathrm{det}} H\left(e_{\mathsf{th}} + \frac{1}{N_{\mathrm{det}}}, e_{SDP}\right)}\right) \\
\leq& \Phi\left(-\sqrt{2N\eta_{\mathsf{th}} H\left(e_{\mathsf{th}} + \frac{1}{N\eta_{\mathsf{th}}}, e_{SDP}\right)}\right),
\end{aligned} \tag{IV.25}$$

where in the third line, we use the property that $\frac{\partial}{\partial x}\Phi(-\sqrt{2NH(x,p)}) \geq 0$ if $x \leq p$. In the fourth line, we use the fact that $\frac{\partial}{\partial N}\Phi(-\sqrt{2NH(x+\frac{1}{N},p)}) \leq 0$ if $p \geq x + \frac{1}{N}$. The soundness is thus a maximum over the probabilities of passing the checks across the three subsets of strategies. □

### D. Numerical Simulation

We study the performance of the LT-QPV protocol by simulating the parameter regime that an honest implementation must satisfy to achieve $\varepsilon_{\mathrm{com}}$-completeness and $\varepsilon_{\mathrm{sou}}$-soundness. For later comparison with the experimental implementation, we use the characterization of the verifier's measurement device in the implementation. The implementation uses $m = 16$ bases, and the bases listed in Sec. VII B. We note that these basis states are not uniformly distributed over the Bloch sphere due to experimental imperfections, but they are sufficiently well-distributed to demonstrate security in the parameter regime achieved by the experiment.

Consider an honest implementation with prover efficiency $\eta$, mismatch probability $p_{\mathsf{mis}}$, and double-click probability $p_{\mathsf{dc}}$. Our goal is to compute the maximum QBER $e$ that an honest implementation must not exceed. At each value of $\eta$ and $p_{\mathsf{mis}}$, we compute the thresholds $\eta_{\mathsf{th}}$ and $p_{\mathsf{mis,th}}$ from Thm. IV.1 and determine the bounds on adversary prover efficiency $\eta_{\mathsf{b}}$ and mismatch $p_{\mathsf{mis,b}}$ consistent with the soundness condition in Thm. IV.2. By solving the SDP in Eqn. (IV.19), we compute the minimum QBER achievable by the adversary. This allows us to select an appropriate QBER threshold $e_{\mathsf{th}}$ to match the soundness condition. From $e_{\mathsf{th}}$, we determine the effective QBER $e' = p_{\mathsf{dc}} + (1-p_{\mathsf{dc}})e$ an honest implementation must achieve to satisfy the completeness requirement.

We run the simulation for $\varepsilon_{\mathrm{sou}} = \varepsilon_{\mathrm{com}} = 10^{-2}$ and split the completeness budget as $\varepsilon_{\mathrm{com},\eta} = 0.1\varepsilon_{\mathrm{com}}$, $\varepsilon_{\mathrm{com,mis}} = 0.1\varepsilon_{\mathrm{com}}$, and $\varepsilon_{\mathrm{com,e}} = 0.8\varepsilon_{\mathrm{com}}$,

which we find to perform well in the regime of interest. For simplicity, we assume $\varepsilon_{ROM} = 0$, but the value is generally small for a suitably chosen hash function $f$ and large enough $n$. We solve the SDP using tools from ncpol2sdpa [10], the PICOS interface [11], and the QICS solver [12] at NPA level $1 + AB$, where the $AB$ operators are restricted only 15 basis combinations: (2,6), (10,14), (1,5), (2,14), (9,13), (1,13), (6,9), (10,11), (11,15), (4,5), (0,4), (2,3), (3,7), (8,9), (14,15). These basis combinations corresponds to verifier basis pairs with the highest overlap

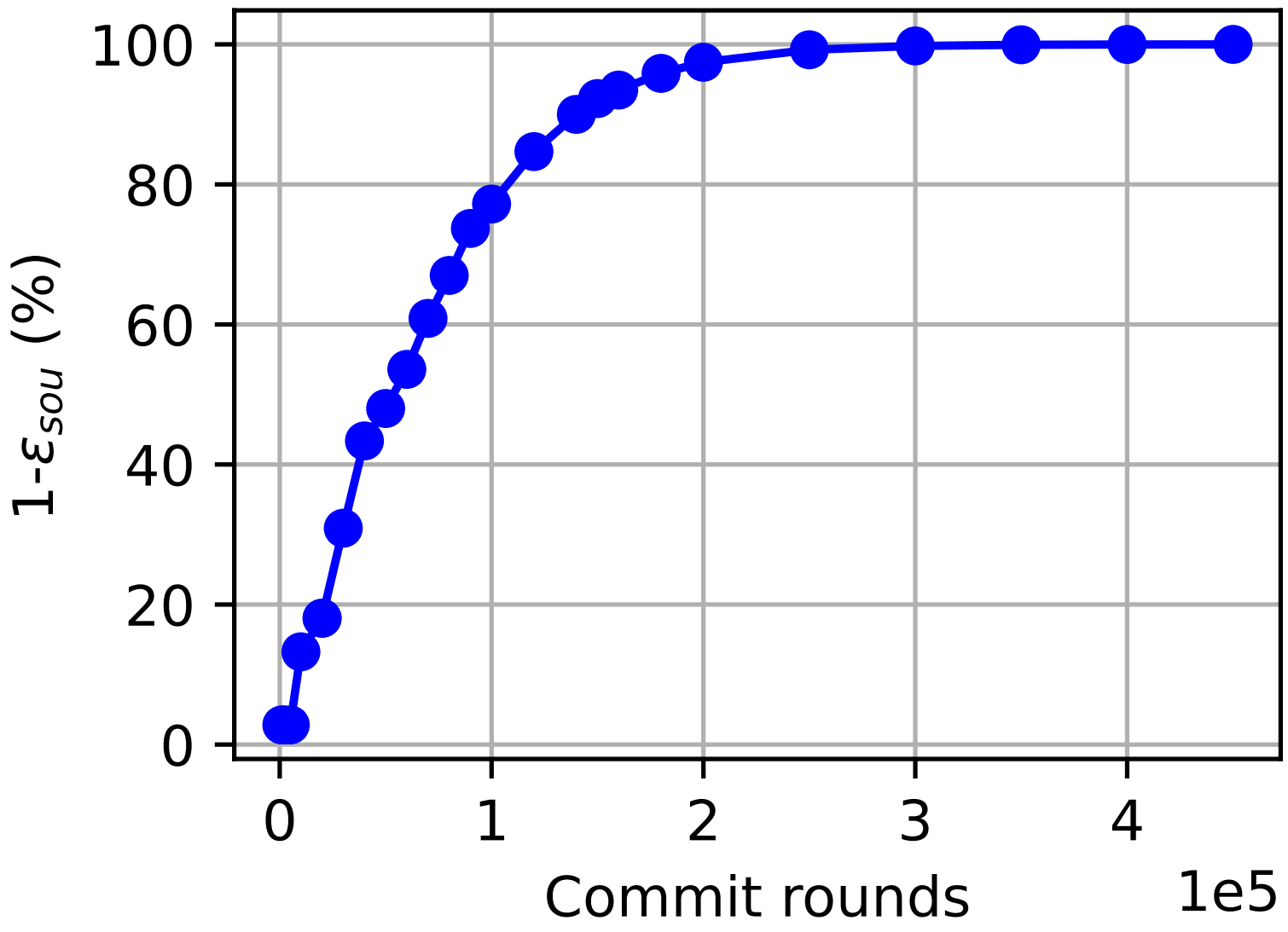


FIG. 4. Confidence of Prover's location (as 1-$\epsilon_{sou}$) vs rounds.

$\max\{\left|\left\langle\psi_0^{\alpha}\middle|\psi_1^{\alpha'}\right\rangle\right|^2, \left|\left\langle\psi_0^{\alpha}\middle|\psi_0^{\alpha'}\right\rangle\right|^2\}$, but the choice to retain 15 combinations is heuristic since having more basis combinations would lead to longer SDP solve times while providing tighter QBER bounds. We use the experimentally measured values of $p_{\mathsf{dc}} = 1.45 \times 10^{-4}$ and $p_{\mathsf{mis}} = 5.49 \times 10^{-6}$. We plot maximum tolerable QBER as a function of prover efficiency, which is shown in Main Text Fig. 2. Based on the experimental performance of $\eta = 0.201$, $e = 0.0117$, $p_{\mathsf{mis}} = 5.49 \times 10^{-6}$ and $p_{\mathsf{dc}} = 1.45 \times 10^{-4}$, a run of $N = 2.55 \times 10^5$ is required, with corresponding threshold $\eta_{\mathsf{th}} = 19.86\%$, $e_{\mathsf{th}} = 1.31\%$ and $p_{\mathsf{mis,th}} = 2.9 \times 10^{-5}$ (equivalent to a check of $N_{\mathrm{mis}} \leq 7$). For the same experimental performance, we also plot the confidence $(1 - \varepsilon_{\mathrm{sou}})$ as the number of rounds chosen is increased in Fig. 4.

### E. Calibration Rounds

In LT-QPV, there are multiple variables that are required to be set before the protocol is performed. For instance, the number of basis $m$ is chosen based on a rough estimate of the performance of our implementation. Crucially, the number of valid rounds $N$, and the threshold parameters $\eta_{\mathsf{th}}$, $e_{\mathsf{th}}$, and $p_{\mathsf{mis,th}}$ have to be carefully chosen to ensure that it remains secure against adversaries (low soundness) while allowing an honest prover's location to be certified (low completeness error). This requires a model characterizing the honest performance of the implementation, which can be obtained through calibration rounds. While calibration rounds are not strictly required to guarantees soundness, having a soundness guarantee alone may not be meaningful. For instance, while setting $\eta_{\mathsf{th}} = 1$, $e_{\mathsf{th}} = 0$, $p_{\mathsf{mis,th}} = 0$ and $N = 10^{10}$ can ensure strong soundness guarantees, the stringent requirements makes it almost impossible for any honest prover to be certified. Therefore, the calibration rounds are necessary to ensure good choice of threshold parameters and have completeness guarantees.

In our protocol, we do not explicitly include the calibration rounds and they are performed prior to the actual QPV run. We note that it is in general possible to have the calibration rounds right before the protocol run [13], but the parameters of the protocol (derived from the calibration) must be selected before the protocol begins. In practice, an adversary can compromise the calibration rounds, but that does not compromise the soundness of the protocol. A fixed choice of parameters $N$, $\eta_{\mathsf{th}}$, $e_{\mathsf{th}}$ and $p_{\mathsf{mis,th}}$ will guarantee the soundness that can be computed from Theorem IV.2. For a fixed soundness parameter $\varepsilon_{\mathrm{sou}}$, an adversary that has influence over the calibration round can influence the choice of the parameters, but any choice of the parameters will either: (1) guarantee a soundness of $\varepsilon_{\mathrm{sou}}$ or (2) trigger the conclusion that the performance of the implemented system cannot guarantee the required soundness, i.e. the actual QPV run will not be performed since soundness cannot be guaranteed.

We stress that in the protocol and security analysis, the parameters are assumed to be chosen prior to the protocol run. As such, one cannot use the calibration data (or part of it) to demonstrate a secure QPV run. For example, taking the data of the first half of a calibration round as the QPV run would imply that during and after the QPV protocol, the parameters $N$, $\eta_{\mathsf{th}}$, $e_{\mathsf{th}}$, and $p_{\mathsf{mis,th}}$ are variables that can change both during and after the protocol run.

This opens up an attack that after the first half of the calibration round (data used for QPV run), the adversary can still alter the thresholds in the second half of the calibration, providing the adversary with an advantage that is not considered in the security analysis. Moreover, using the same data for both calibration of the thresholds and protocol run would imply that the thresholds and protocol run data are correlated random variables. For instance, $N_{\rm det}$, $N$ and $\eta_{\sf th}$ are random variables that are strongly correlated through a complex relation that depends on how $\eta_{\sf th}$ is computed from the data. As such, statistical bounds used in the security analysis where fixed thresholds are assumed, such as the concentration bound from Ref. [7] for i.i.d. attacks and from Ref. [14] for general attacks, would no longer apply.

## V. LT-QPV SECURE AGAINST GENERAL ATTACKS

We extend the security of LT-QPV and demonstrate that LT-QPV remains sound against general QPT adversaries. More precisely, we prove the security for an LT-QPV protocol where the error check (in parameter estimation) is $N_{\rm dc} + N_{\rm err} \leq N\tilde{e}_{\sf th}$. We term this LT-QPV protocol with unconditioned error check. We term $\tilde{e}$ the unconditioned error, which is not conditioned on detection at $V_0$, to distinguish it from the QBER $e$.

### A. Completeness

The completeness analysis is slightly modified from that of Protocol 1 due to the change in error check.

**Theorem V.1.** *Consider an honest prover at $X_P$, with single-round prover efficiency $\eta$, QBER $e$, double-click probability $p_{\sf dc}$, and mismatch probability $p_{\sf mis}$. LT-QPV protocol with unconditioned error check is $\varepsilon_{\rm com}$-complete with completeness error*

$$\varepsilon_{\rm com} = \Phi\left(-\sqrt{2NH\left(\eta_{\sf th} + \frac{1}{N}, \eta\right)}\right) + 1 - \Phi\left(\sqrt{2NH\left(\tilde{e}_{\sf th} - \frac{1}{N}, \tilde{e}'\right)}\right) + 1 - \Phi\left(\sqrt{2NH\left(p_{\sf mis,th} - \frac{1}{N}, p_{\sf mis}\right)}\right), \quad \text{(V.1)}$$

*where $e' = p_{dc} + (1-p_{dc})e$, $\tilde{e}' = e'\eta$, $H(x,p) = x\ln(x/p) + (1-x)\ln[(1-x)/(1-p)]$, $\Phi(x)$ is the cumulative distribution function of a standard normal distribution, $\eta_{\sf th}$, $\tilde{e}_{\sf th}$, and $p_{\sf mis,th}$ are the threshold values for prover efficiency, unconditioned error and mismatch probability and are selected such that $\eta_{\sf th} \leq \eta - \frac{1}{N}$, $\tilde{e}_{\sf th} \geq \tilde{e}' + \frac{1}{N\eta_{\sf th}}$ and $p_{\sf mis,th} \geq p_{\sf mis} + \frac{1}{N}$.*

*Proof.* We follow the proof of Thm. IV.1, except for the error check. We bound the error check failure probability,

$$\begin{aligned}
\Pr[N_{\rm err} + N_{\rm dc} > N\tilde{e}_{\sf th}] \leq& \Pr[N_{\rm err} + N_{\rm dc} > \lfloor N\tilde{e}_{\sf th}\rfloor] \\
=& 1 - \Pr[N_{\rm err} + N_{\rm dc} \leq \lfloor N\tilde{e}_{\sf th}\rfloor] \\
\leq& 1 - \Phi\left(\sqrt{2NH\left(\frac{\lfloor N\tilde{e}_{\sf th}\rfloor}{N}, \tilde{e}'\right)}\right) \\
\leq& 1 - \Phi\left(\sqrt{2NH\left(\tilde{e}_{\sf th} - \frac{1}{N}, \tilde{e}'\right)}\right),
\end{aligned} \quad \text{(V.2)}$$

where the third line applies the tight concentration bound [7] and the fourth line uses the property that $\frac{\partial}{\partial x}\Phi(\sqrt{2NH(x,\tilde{e}')}) \geq 0$ if $x \geq \tilde{e}'$. □

### B. Soundness

The proof of soundness is similar to the i.i.d. analysis, except for the single-round analysis and finite-size analysis step. The first three steps of universal squashing, random oracle model, and channel loss independence are agnostic to whether the attack is general or i.i.d., and can therefore apply directly to the soundness proof against general QPT adversaries. The single-round analysis follows similarly, where we seek the minimum unconditioned error $\tilde{e}$ instead of the minimum QBER $e$ conditioned on detection, and use the dual variables of the SDP to demonstrate a linear relation between $\tilde{e}$, $p_{\sf mis}$, and $\eta$. The finite-size analysis then applies a bound from Ref. [14], noting that the linear relation derived from single-round analysis holds even when conditioned on past rounds.

**Theorem V.2.** *LT-QPV protocol with unconditioned error check is $\varepsilon_{\text{sou}}$-sound against general QPT adversaries in the quantum random oracle model, i.e.*

$$\max_{\mathcal{S}\in\mathcal{A}_{QPT},\mathcal{M}_{\text{prac}}} \Pr[\textit{checks pass}] \leq \varepsilon_{\text{sou}}, \tag{V.3}$$

*where* $\varepsilon_{\text{sou}} = \varepsilon'_{\text{sou}} + \varepsilon_{ROM}$*,*

$$\nu + \beta(1-\eta_{\mathsf{th}}) - \tilde{e}_{\mathsf{th}} + \gamma p_{\mathsf{mis,th}} = \sqrt{2V}\sqrt{\frac{\log_2(1/\varepsilon'_{\text{sou}})}{N}} + \frac{\nu}{3}\frac{\log_2(1/\varepsilon'_{\text{sou}})}{N}, \tag{V.4}$$

*for any choice of dual variables* $\beta, \gamma \leq 0$ *and* $\{\varphi_i\}_i$*, and* $\nu(\beta,\gamma,\{\varphi_i\}_i)$ *is a function of dual variables of the SDP in Eqn.* (V.9) *and* $V$ *is chosen such that* $V = \max\{\nu^2, (\nu + \min\{\beta,-1\}+\gamma)^2\}$*.*

*Proof.* It is sufficient to prove that

$$\max_{\mathcal{S}\in\mathcal{A}_{res,nl},\mathcal{M}_{\text{q}}} \Pr[\textit{checks pass}] \leq \varepsilon - \varepsilon_{ROM} \tag{V.5}$$

since Thms. IV.3 to IV.5 guarantee the $\varepsilon_{\text{sou}}$-soundness.

$$\begin{aligned}&\max_{\mathcal{S}\in\mathcal{A}_{res,nl},\mathcal{M}_{\text{q}}} \Pr[\textit{checks pass}] \leq \varepsilon_{\text{sou}} - \varepsilon_{ROM} \xRightarrow{\textit{Thm. IV.5}} \max_{\mathcal{S}\in\mathcal{A}_{res,q},\mathcal{M}_{\text{q}}} \Pr[\textit{checks pass}] \leq \varepsilon_{\text{sou}} - \varepsilon_{ROM} \\ &\xRightarrow{\textit{Thm. IV.4}} \max_{\mathcal{S}\in\mathcal{A}_{QPT,q},\mathcal{M}_{\text{q}}} \Pr[\textit{checks pass}] \leq \varepsilon_{\text{sou}} \xRightarrow{\textit{Thm. IV.3}} \max_{\mathcal{S}\in\mathcal{A}_{QPT},\mathcal{M}_{\text{prac}}} \Pr[\textit{checks pass}] \leq \varepsilon_{\text{sou}}.\end{aligned} \tag{V.6}$$

Applying Thm. V.4 completes the proof. □

### 1. *Single-Round Analysis*

We present the single-round analysis to link $\tilde{e}$, $p_{\mathsf{mis}}$ and $\eta$, noting that any double-click events have been assigned to error outcomes in the ideal verifier measurement $\mathcal{M}_{\text{q}}$ via universal squashing (see Theorem IV.3).

**Theorem V.3.** *Consider a single round of LT-QPV protocol with unconditioned error check with restricted adversary strategy* $\mathcal{A}_{res,nl}$ *and ideal verifier measurement* $\mathcal{M}_{\text{q}}$*. Let* $\eta_i$*,* $p_{\mathsf{mis},i}$ *and* $\tilde{e}_i$ *be the prover efficiency, mismatch probability and unconditioned error of the i-th round respectively. For any choice of dual variables* $\beta,\gamma \leq 0$ *and* $\{\varphi_i\}_i$*, the single-round quantities satisfy the tradeoff*

$$\tilde{e}_i \geq \nu + \beta(1-\eta_i) + \gamma p_{\mathsf{mis},i}, \tag{V.7}$$

*where* $\nu(\beta,\gamma,\{\varphi_i\}_i)$ *is a function of dual variables of the SDP in Eqn.* (V.9)*.*

*Proof.* We follow the argument of Thm. IV.6, in which we minimize the unconditioned error $\tilde{e}_i$, which can be expressed as

$$\begin{aligned}\min_{|\phi_0\rangle,|\phi_1\rangle,\{\tilde{A}^\alpha_{z_0}\}_{z_0},\{\tilde{B}^\alpha_{z_1}\}_{z_1}} \quad & \frac{1}{m}\sum_{\alpha=1}^{m} \langle\phi_0^\alpha|\tilde{A}_1^\alpha\otimes\mathbb{I}_B|\phi_0^\alpha\rangle + \langle\phi_1^\alpha|\tilde{A}_0^\alpha\otimes\mathbb{I}_B|\phi_1^\alpha\rangle \\ \textit{subj. to} \quad & \frac{1}{m}\sum_{\alpha=1}^{m}\sum_{j=0}^{1} \langle\phi_j^\alpha|\tilde{A}_\perp^\alpha\otimes\mathbb{I}_B|\phi_j^\alpha\rangle \leq 1-\eta_i \\ & \frac{1}{m}\sum_{\alpha=1}^{m}\sum_{j=0}^{1}\sum_{z_0\neq z_1} \langle\phi_j^\alpha|\tilde{A}_{z_0}^\alpha\otimes\tilde{B}_{z_1}^\alpha|\phi_j^\alpha\rangle \leq p_{\mathsf{mis},i} \\ & \langle\phi_0^\alpha|\phi_0^\alpha\rangle + \langle\phi_1^\alpha|\phi_1^\alpha\rangle = 1, \,\forall\alpha\in\{1,\cdots,m\}\end{aligned}. \tag{V.8}$$

Since $\tilde{e}_i$ is not conditioned on detection, the objective function differs by a factor of $\eta$ from that in Thm. IV.6. We

use the same method [9] to relax the optimization problem into an SDP,

$$\begin{aligned} \min_{G\geq 0} \quad & \frac{1}{m}\sum_{\alpha=1}^{m} \langle \phi_0^\alpha|\tilde{A}_1^\alpha|\phi_0^\alpha\rangle + \langle \phi_1^\alpha|\tilde{A}_0^\alpha|\phi_1^\alpha\rangle \\ subj.\,to \quad & \frac{1}{m}\sum_{\alpha=1}^{m}\sum_{j=0}^{1} \langle \phi_j^\alpha|\tilde{A}_\perp^\alpha|\phi_j^\alpha\rangle \leq 1-\eta_i \\ & \frac{1}{m}\sum_{\alpha=1}^{m}\sum_{j=0}^{1}\sum_{z_0\neq z_1} \langle \phi_j^\alpha|\tilde{A}_{z_0}^\alpha\tilde{B}_{z_1}^\alpha|\phi_j^\alpha\rangle \leq p_{\mathsf{mis},i} \\ & \langle\phi_0|\phi_0\rangle + \langle\phi_1|\phi_1\rangle = 1 \\ & [\tilde{A}_{z_0}^\alpha, \tilde{B}_{z_1}^{\alpha'}] = 0,\ \forall \alpha,\alpha',z_0,z_1 \end{aligned} . \tag{V.9}$$

We now use the SDP to construct a relation between the single-round quantities. Let $F_\eta$, $F_{\mathrm{mis}}$, and $F_{\mathrm{obj}}$ be matrices representing the prover efficiency constraint, mismatch constraint, and objective function respectively, such that $\mathrm{Tr}[F_\eta G] \leq 1-\eta_i$, $\mathrm{Tr}[F_{\mathrm{mis}} G] \leq p_{\mathsf{mis},i}$, and $\mathrm{Tr}[F_{\mathrm{obj}} G]$. Let $F_j$ and $c_j$ be the matrix and constant of the remaining constraints that are expressed as $\mathrm{Tr}[F_j G] = c_j$. We define a Lagrangian with dual variables $\beta'$, $\gamma'$, and $\{\varphi_j\}_j$,

$$L(\beta',\gamma',\{\varphi_j\}_j,G) = \mathrm{Tr}[F_{\mathrm{obj}}G] + \beta'(\mathrm{Tr}[F_\eta G] - 1 + \eta_i) + \gamma'(\mathrm{Tr}[F_{\mathrm{mis}}G] - p_{\mathsf{mis},i}) + \sum_j \varphi_j(\mathrm{Tr}[F_jG]-c_j) \tag{V.10}$$

and construct the Lagrangian dual function,

$$\begin{aligned} g(\beta',\gamma',\{\varphi_j\}_j) =& \inf_{G\geq 0} \mathrm{Tr}\left[\left(F_{\mathrm{obj}} + \beta' F_\eta + \gamma' F_{\mathrm{mis}} + \sum_j \varphi_j F_j\right) G\right] - \beta'(1-\eta_i) - \gamma' p_{\mathsf{mis},i} - \sum_j \varphi_j c_j \\ =& \nu(\{\varphi_j\}_j,\beta',\gamma') - \beta'(1-\eta_i) - \gamma' p_{\mathsf{mis},i}, \end{aligned} \tag{V.11}$$

where $\nu = \inf_{G\geq 0} \mathrm{Tr}\Big[(F_{obj} + \beta' F_\eta + \gamma' F_{\mathrm{mis}} + \sum_j \varphi_j F_j)G\Big] - \sum_j \varphi_j c_j$. Since the SDP is a convex optimization problem, the dual function lower bounds the solution to the SDP for any choice of $\varphi$ and $\beta',\gamma' \geq 0$ [15], i.e.

$$\tilde{e}_i \geq \tilde{e}_{SDP} \geq \nu - \beta'(1-\eta_i) - \gamma' p_{\mathsf{mis},i}. \tag{V.12}$$

A simple redefinition of $\beta = -\beta'$ and $\gamma = -\gamma'$ completes the proof. □

### 2. *Finite-Size Analysis*

**Theorem V.4.** *LT-QPV protocol with unconditioned error check with an ideal verifier measurement device is $\varepsilon$-sound against any restricted adversary that prepares a qubit state in $V_i$, i.e.*

$$\max_{\mathcal{S}\in\mathcal{A}_{res,nl},\mathcal{M}_{\mathrm{q}}} \Pr[checks\, pass] \leq \varepsilon, \tag{V.13}$$

*where*

$$\nu + \beta(1-\eta_{\mathsf{th}}) - \tilde{e}_{\mathsf{th}} + \gamma p_{\mathsf{mis,th}} = \sqrt{2V}\sqrt{\frac{\log_2(1/\varepsilon)}{N}} + \frac{\nu}{3}\frac{\log_2(1/\varepsilon)}{N}, \tag{V.14}$$

*for any choice of dual variables $\beta,\gamma \leq 0$ and $\{\varphi_i\}_i$, and $\nu(\beta,\gamma,\{\varphi_i\}_i)$ is a function of dual variables in Eqn. (V.9) and $V$ is chosen such that $V = \max\{\nu^2, (\nu + \min\{\beta,-1\}+\gamma)^2\}$.*

*Proof.* Before we begin the analysis, it is convenient to define the random variables

$$C_i = \begin{cases} \beta, & z_{0,i} = \perp \\ 0, & z_{0,i} \neq \perp, z_{0,i} = s_i \\ -1, & z_{0,i} \neq \perp, z_{0,i} \neq s_i \end{cases} \tag{V.15}$$

and

$$T_i = \begin{cases} 0, & z_{0,i} = z_{1,i} \\ \gamma, & z_{0,i} \neq z_{1,i} \end{cases} \tag{V.16}$$

and

$$Y_i = \nu + C_i + T_i. \tag{V.17}$$

By the single-round tradeoff in Thm. V.3, the expectation value $\mathbb{E}[Y_i]$ is given by

$$\mathbb{E}[Y_i] = \nu + \beta(1 - \eta_i) - \tilde{e}_i + \gamma p_{\mathsf{mis},i} \leq 0. \tag{V.18}$$

Since the tradeoff has to be obeyed by any adversary strategy in set $\mathcal{A}_{res,nl}$, this bound holds even when conditioned on past histories. This implies that $\mathbb{E}[Y_i|Y^{i-1}] \leq 0$. Proposition 11 of Ref. [14] states that when $\mathbb{E}[Y_i|Y^{i-1}] \leq 0$ for all $i \in \{1, \cdots, N\}$,

$$\Pr\left[\sum_i Y_i \geq N\delta\right] \leq \varepsilon \tag{V.19}$$

where

$$\delta = \sqrt{2V}\sqrt{\frac{\log_2(1/\varepsilon)}{N}} + \frac{Y_{max}\log_2(1/\varepsilon)}{3N} \tag{V.20}$$

and $V$ is chosen such that $\mathbb{E}[(Y_i)^2|Y^{i-1}] \leq V$ and $Y_{max} = \max_i Y_i = \nu$, since $\beta, \gamma \leq 0$.

We expand the summation as

$$\sum_i Y_i = N\nu + \beta(N - N_{\text{det}}) + \gamma N_{\text{mis}} - N_{\text{err}}. \tag{V.21}$$

Since the checks pass when $N_{\text{mis}} \leq Np_{\mathsf{mis,th}}$, $N_{\text{err}} \leq Ne_{\mathsf{th}}$ and $N_{\text{det}} \geq N\eta_{\mathsf{th}}$, we set

$$\delta = \nu + \beta(1 - \eta_{\mathsf{th}}) - \tilde{e}_{\mathsf{th}} + \gamma p_{\mathsf{mis,th}} \tag{V.22}$$

to give

$$\Pr[-\beta N_{\text{det}} + \gamma N_{\text{mis}} - N_{\text{err}} \geq -\beta N\eta_{\mathsf{th}} - Ne_{\mathsf{th}} + \gamma Np_{\mathsf{mis,th}}] \leq \varepsilon. \tag{V.23}$$

Since passing all checks implies $\sum_{i=1}^{N} Y_i \geq N\delta$, we conclude that

$$\Pr[\mathit{checks\,pass}] \leq \Pr\left[\sum_{i=1}^{N} Y_i \geq N\delta\right] \leq \varepsilon. \tag{V.24}$$

To choose $V$, we notice that $Y_i \in [\nu + \min\{\beta, -1\} + \gamma, \nu]$, since $\beta, \gamma \leq 0$. Therefore, $Y_i^2 \leq \max\{\nu^2, (\nu + \min\{\beta, -1\} + \gamma)^2\}$. As such, $\mathbb{E}[Y_i^2|Y^{i-1}] \leq \max\{\nu^2, (\nu + \min\{\beta, -1\} + \gamma)^2\}$ and it suffices to choose the RHS term as $V$. □

### C. Numerical Simulation

We simulate the performance of the LT-QPV protocol against a general QPT adversary using the same set of parameters as the i.i.d. case in Sec. IV D. We run the simulation for $\varepsilon_{\text{sou}} = \varepsilon_{\text{com}} = 10^{-2}$ and split the completeness budget as $\varepsilon_{\text{com},\eta} = 0.1\varepsilon_{\text{com}}$, $\varepsilon_{\text{com,mis}} = 0.1\varepsilon_{\text{com}}$ and $\varepsilon_{\text{com,e}} = 0.8\varepsilon_{\text{com}}$.

The numerical simulation is performed in a similar manner to that described in Sec. IV D, where we seek to determine the QBER $e$ that an honest implementation requires. We also use the same numerical tools as in Sec. IV D. At each $\eta$ and $p_{\mathsf{mis}}$, we compute $\eta_{\mathsf{th}}$ and $p_{\mathsf{mis,th}}$ using the completeness bound in Thm. V.1. These can then be used to compute $\tilde{e}_{\mathsf{th}}$, though the method of computing $\tilde{e}_{\mathsf{th}}$ from $\eta_{\mathsf{th}}$ and $p_{\mathsf{mis,th}}$ differs between the two soundness proofs. From $\tilde{e}_{\mathsf{th}}$, we determine the effective error $e' = p_{\mathsf{dc}} + (1 - p_{\mathsf{dc}})e$ an honest implementation must achieve based on the completeness.

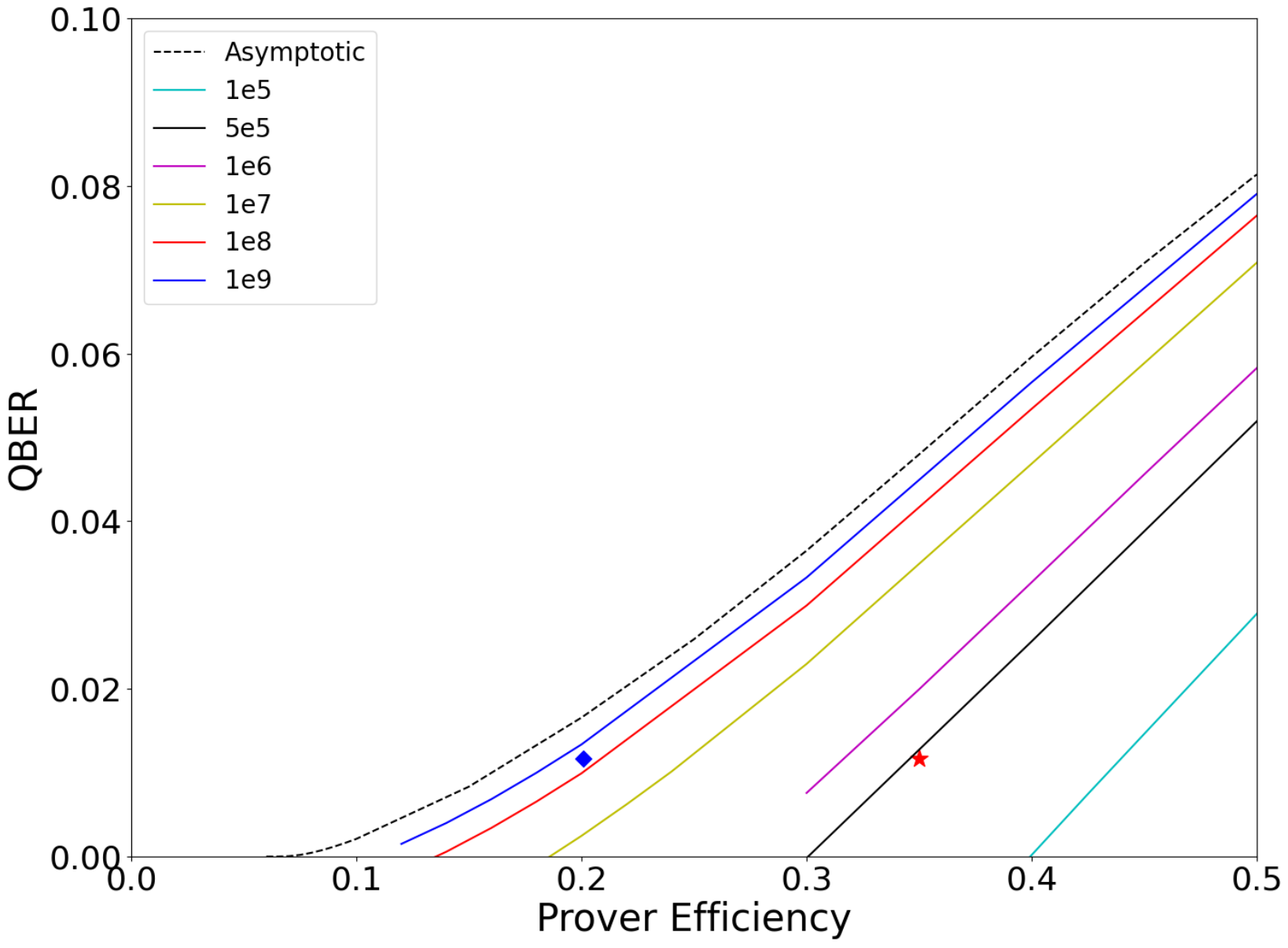


FIG. 5. Plot of QBER and prover efficiency regime secure against general QPV adversaries for $N = 10^5$ to $N = 10^9$ (represented by different colored lines) and in the asymptotic regime (dotted black line). The blue diamond represents the experimental performance and the red star represents the expected performance with improved end-to-end efficiency.

We compute $e_{\text{th}}$ from Eqn. (V.4). We maximize $e_{\text{th}}$ by optimizing the choice of dual variables, where different dual variable combinations are obtained by solving the SDP in Eqn. (V.9) at different $p_{\text{mis},i}$ and $\eta_i$. A larger $e_{\text{th}}$ corresponds to a less stringent requirement on the honest prover.

The simulation results are shown in Fig. 5, with the SDP solved at NPA level $1 + AB$, where the $AB$ operators are restricted to the 15 basis combinations corresponding to the closest bases. The number of rounds necessary to guarantee the same level of security is increased due to the generalization of the adversary. For the same experimental implementation, a run of the order of $N = 5 \times 10^8$ is required, compared to $N = 2.55 \times 10^5$ against an iid adversary. Improving the end-to-end efficiency to obtain an overall efficiency of 35% can significantly reduce the number of rounds by almost 3 orders of magnitude to $N = 5 \times 10^5$.

### D. Reuse of Verifier Basis

In the experimental implementation, the verifier basis is reused across multiple rounds due to the slow rate of LCR. In general, this would be insecure against general adversaries which may capitalize on this basis correlation to gain an advantage. Consider the extreme example that only 16 blocks of bases (LCR only switches 16 times) are utilized. An adversary strategy where two colluding parties $E_1$ and $E_2$ located at $V_0$ and $V_1$ respectively first guess the basis of each block, and has $E_1$ prepare quantum states $V$ in that guessed basis and with bit value drawn from a shared random string with $E_2$, can spoof the prover's location if all basis guesses are correct. This implies that soundness against a general adversary is at least $(1/16)^{16}$, and cannot be reduced further for instance to be below $2^{-128}$. There may be more sophisticated strategies that may impact the security analysis.

The key part of the analysis where basis reuse fails is the finite size analysis, where we argue that $\mathbb{E}[Y_i] \leq 0$ holds when conditioned on any past histories. This assumes that the basis is freshly chosen in the $i$-th round of the protocol and not correlated to the previous rounds. If the basis is reused between the $(i-1)$-th and $i$-th rounds, consider the described strategy where an adversary makes a guess of the basis in the $(i-1)$-th round, and use the same basis guess in the $i$-th round. The specific history where the $(i-1)$-th round basis guess is correct would imply that the $i$-th round basis guess is correct as well, and $\mathbb{E}[Y_i] > 0$ since the responses with be correct and without mismatch. If the verifier instead switches to independent per-round switching by using Pockels cell, the basis for every round is now independently chosen, and this recovers the finite-size analysis argument.

## VI. ADVERSARY WITH IMPERFECT DEVICES

In our security analysis, the adversary is assumed to possess perfect quantum devices. However, in practice, they remain subjected to limitations on device performance. This introduces additional challenges for an adversary, potentially restricting the set of available strategies.

In LT-QPV, a key challenge for the adversary is to coordinate responses to both verifiers while maintaining a mismatch $p_{\mathsf{mis}}$ below the protocol threshold $p_{\mathsf{mis,th}}$. An honest prover can utilize classical communication which has minimal errors to satisfy the mismatch condition, while an adversary has to coordinate the responses from Alice and Bob. Consider the set of adversary strategies $\mathcal{A}_{res}$, in which Alice and Bob independently measure their respective quantum systems to generate responses $z_0$ and $z_1$. If the measurement devices have detection efficiency $\eta_{det} < 1$, each device produces a no-detection outcome $z_0 =\perp$ and $z_1 =\perp$ with probability $1 - \eta_{det}$. The independent losses on both sides directly contribute to mismatch. Even at a detector loss of 1%, the induced mismatch of approximately 2% significantly exceeds the typical $p_{\mathsf{mis,th}}$. Comparable effects may arise from other imperfections, such as imperfect state generation, and the extent of such deviations is strongly dependent on the adversary's implementation platform.

Consider an adversary that has access to detectors with $\eta_{det} < 1$ with no mismatch tolerance, $p_{\mathsf{mis,th}} = 0$. The resulting mismatch forces the adversary to forgo quantum measurements entirely and rely solely on classical communication. This reduces the set of adversary strategies to $\mathcal{A}_{cl}$, the set of classical strategies. Here, an adversary prepares a classical-quantum state $\rho_{VQ} = \sum_q p_q \left|q\right\rangle\!\left\langle q\right|_Q \otimes \rho_V^q$, where $q$ is shared between Alice and Bob and $V$ is a single qubit. The responses $z_0$ and $z_1$ are determined by Alice and Bob classically $z_0 = z_1 = g(q, \alpha)$ using the shared information $q$ and the basis $\alpha$.

This reduction raises a natural question: What parameter regimes, if any, does a gap exist between the set of restricted strategies $\mathcal{A}_{res}$ and the set of classical strategies $\mathcal{A}_{cl}$? Where such a gap exists, improved QPV performance can be established under a reasonable assumption that the adversary's devices are imperfect. In the parameter regimes where none exists, the computation of error tolerance can be significantly simplified. Since the classical strategy is fully described by the choice of qubit state $\rho_V$, the optimal error threshold can be computed by optimizing over $\rho_V$. We leave such an exploration for future work.

## VII. EXPERIMENT

### A. Quantum layer details

Before implementing the LT-QPV protocol, we thoroughly characterized the quantum layer of our experiment. Tables II and III show the efficiency of the transmission of the telecom and NIR transmission and detection apparatus. The SNSPDs used which have very low dark counts (< 5000 counts/s) and only dark counts falling within the 1.2ns coincidence/ROI window can affect a protocol round, giving a per-detector dark count probability of approximately $p_d < 6 \times 10^{-6}$. We also report the efficiency mismatch between the horizontal and vertical polarization channels in Table IV. The verifier-side (telecom) detector shows only a minimal difference between the two channels, which is fairly consistent with the detector model's assumption of balanced detection efficiency. We also note that this mismatch and the total transmission is consistent across all 16 basis since transmission through the Liquid Crystal does not depend on the basis choice. We note that in general, any imbalance on the verifier device between the detection efficiencies can be compensated by either artificially introducing loss in the SNSPD arm with higher transmission or randomly performing a 180 degree rotation of the basis to average out the mismatch in detection efficiencies.

TABLE II. Telecom efficiencies

| Location | Step efficiency | Cumulative efficiency |
|---|---|---|
| After deployed fiber | 0.81(4) | 0.81(4) |
| After LCRs | 0.87(5) | 0.70(5) |
| After PBS | 0.97(6) | 0.68(5) |
| After SNSPD fiber | 0.75(7) | 0.51(6) |
| SNSPD quantum efficiency | 0.805(5) | 0.41(5) |

TABLE III. NIR efficiencies

| Location | Step efficiency | Cumulative efficiency |
|---|---|---|
| After delay line | 0.93(5) | 0.93(5) |
| After Pockels cell 1 | 0.97(6) | 0.90(5) |
| After Pockels cell 2 | 0.98(6) | 0.88(5) |
| After PBS | 0.93(6) | 0.82(5) |
| After SNSPD fiber | 0.82(7) | 0.67(6) |
| SNSPD quantum efficiency | 0.765(5) | 0.52(5) |

TABLE IV. Horizontal/Vertical mismatch

| NIR | |
|---|---|
| Horizontal efficiency | 0.61(5) |
| Vertical efficiency | 0.64(5) |
| **Telecom** | |
| Horizontal efficiency | 0.61(5) |
| Vertical efficiency | 0.60(5) |

### B. Quantum Measurement Basis

The LT-QPV protocol requires that the prover and verifier measure half of an entangled photon pair. The prover performs a polarization measurement in the commanded basis, depending on the values $x$ and $y$, out of 16 possible bases. The verifier also performs a polarization measurement on the other half of the entangled pair. To keep the positional uncertainty at a minimum, the prover must be able to change measurement basis rapidly. The speed constraints at the verifier are much less stringent. The optical design for the verifier polarization measurement is presented in Fig. 6, and consists of two back-to-back full-wave liquid crystal retarders (LCR). One is oriented with the slow axis at 0°, and the other at 45°. This configuration gives us the ability to explore the whole Poincaré sphere and thus produce all LT-QPV target states.

The optical design for the prover's measurement station is presented in Fig. 7. We meet the protocol's speed requirements by using fast-switching Pockels cells (RTP 4 mm, Leysop), controlled by a Field Programmable Gate Array (FPGA). The electronic driver (PCD-DPP, BME Bergmann) for each Pockels cell supports 4 input voltages, enabling 16 total basis combinations for our back-to-back Pockels cells. The Pockels cell slow axis is fixed at 45°, and by using a half-wave plate (HWP) and a quarter-wave plate (QWP), we are able to achieve all target states.

We characterize the prover's optical stack using a classical polarization tomography. The theoretical target states are presented in Fig. 8a, and our measured target states for the prover are presented in Fig. 8b. The optical fiber between the prover and the verifier introduces a unitary transformation to the transmitted telecom photons. Thus, we separately calibrate the LCR voltages that match each of the prover bases after the fiber rotation and confirm the shape of the bases constellation via state tomography. The exact basis vectors as observed by the verifier are shown in Fig. 8c, with the detailed list in Table V. The typical RMS Error of the tomography is 0.0055.

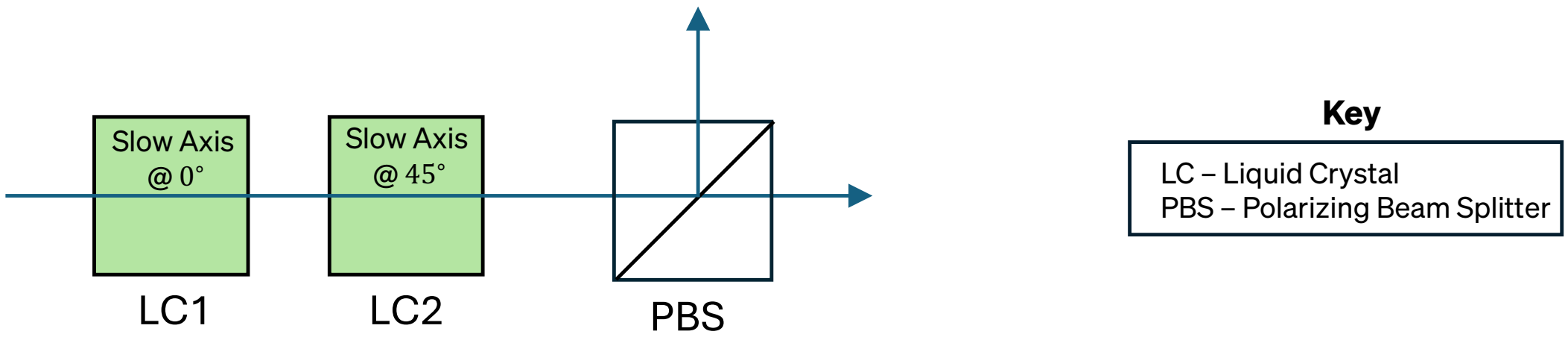


FIG. 6. Verifier Polarization Measurement Station consists of two Liquid Crystal (LC) rotators and a Polarizing Beam Splitter (PBS).

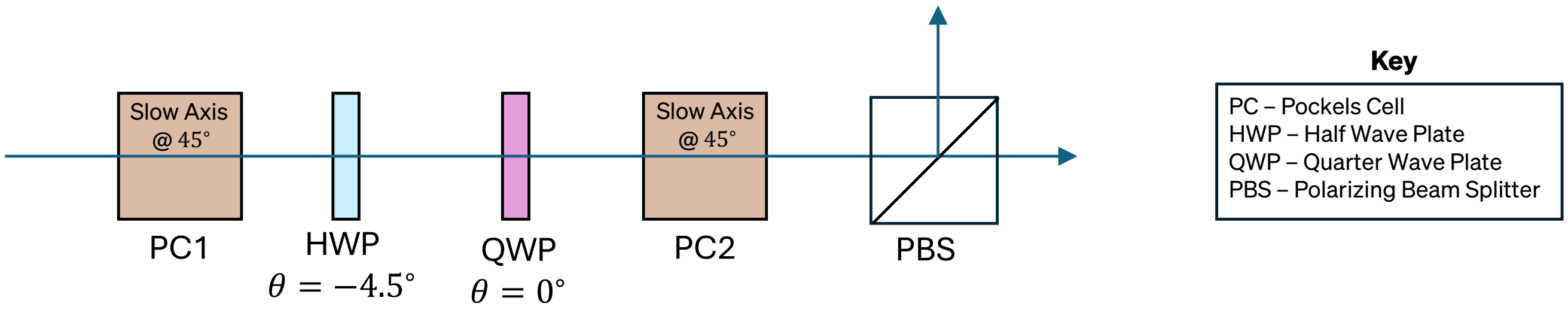


FIG. 7. Prover Polarization Measurement Station consists of two Pockels cells (PC) fast rotators, Half-Wave Plate (HWP), Quarter-Wave Plate (QWP), and Polarizing Beam Splitter (PBS).

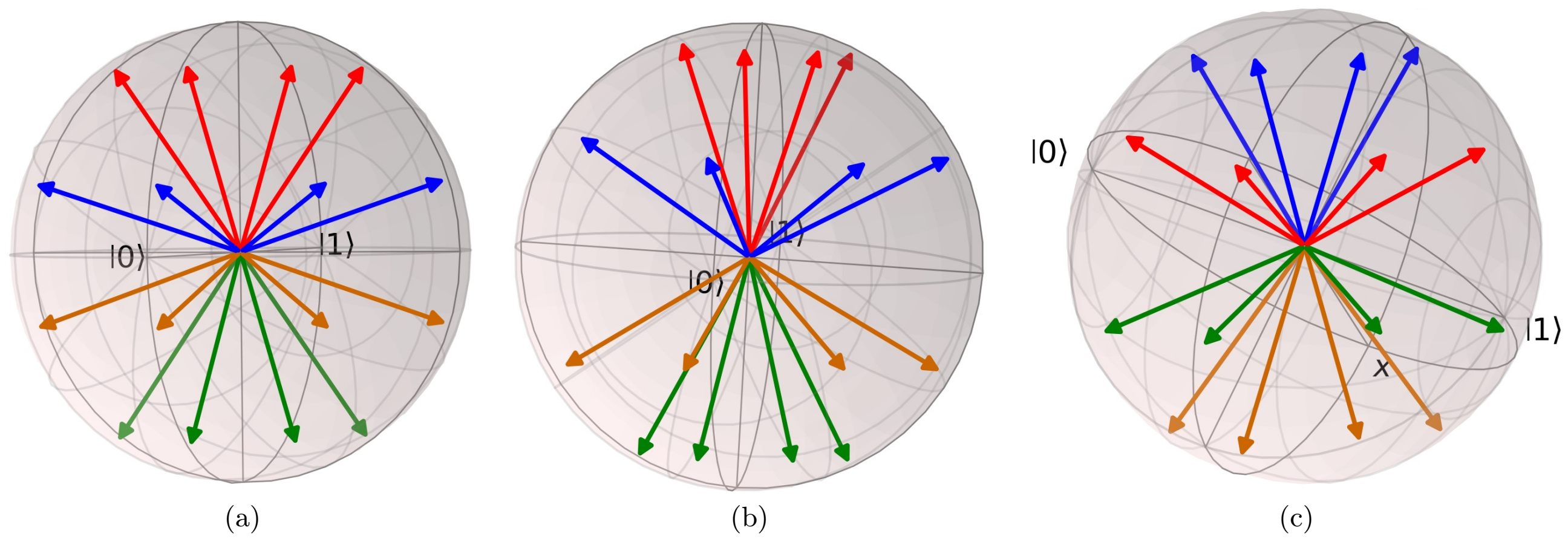


FIG. 8. State Constellation: (a) Ideal States, (b) Pockels Cell States, (c) Verifier States.

### C. Low-Latency Communications System

In order to successfully implement our LT-QPV protocol, we designed a custom low-latency communications link driven by FPGA devices over coaxial RF cables with high-propagation velocity (e.g., approximately 88% speed of light). For the LT-QPV protocol, the prover is connected to each verifier $V_0$ and $V_1$ with a custom low-latency communications link driven by AMD Zynq RFSoC FPGA's separated by approximately 35(2) m and 40(2) m, respectively.

The round-trip-time of a communication system depends on the transmission latency combined with the processing latency. The processing latency of each node is dependent on the signal routing within the FPGA fabric. For our protocol, we use the programmable logic firmware as the central control system for the low-latency link. To reduce processing latency, modulated electrical signals were generated and received using the high-performance general-purpose input/output (HP-GPIO) pins available at each node. Various dedicated hardware peripherals were tested for lowest latency performance with the GPIO 'bit banging' software-driven process offering the highest performance. Minimizing the processing latency is critical for executing the LT-QPV protocol without suffering unwanted delays which result in increased positional uncertainty. This low-latency HP-GPIO is used to transmit and receive the random challenges messages $x, y \in \{0,1\}^4$ using a hardware random number generator in the FPGA's programmable logic (PL) [16].

Existing communication approaches, such as fiber-optic cable or coaxial cable, provide communication speeds slightly slower than the speed of light and cables often do not follow straight lines, both of which can add unacceptable latency. Future improvements to this link will involve utilizing the RFSoC's integrated Radio Frequency (RF) transceivers which can support classical radio signals to travel approximately at the speed of light through air and potentially in a direct line of sight. Current commercially available point-to-point radio links, however, have transmission and reception latencies that are too high to be of practical value for QPV.

| # | Prover basis | Verifier basis | Target basis |
|---|---|---|---|
| 0 | $0.8188 \|H\rangle + (0.3711 + 0.4381i) \|V\rangle$<br>$0.5741 \|H\rangle + (-0.5292 - 0.6247i) \|V\rangle$ | $0.6429 \|H\rangle + (0.7659 + 0.0099i) \|V\rangle$<br>$0.7660 \|H\rangle + (-0.6428 - 0.0083i) \|V\rangle$ | $0.7071 \|H\rangle + (0.5721 + 0.4156i) \|V\rangle$<br>$0.7071 \|H\rangle + (-0.5721 - 0.4156i) \|V\rangle$ |
| 1 | $0.8343 \|H\rangle + (0.0494 + 0.5492i) \|V\rangle$<br>$0.5514 \|H\rangle + (-0.0747 - 0.8309i) \|V\rangle$ | $0.6271 \|H\rangle + (0.6339 - 0.4527i) \|V\rangle$<br>$0.7790 \|H\rangle + (-0.5103 + 0.3644i) \|V\rangle$ | $0.7071 \|H\rangle + (0.2185 + 0.6725i) \|V\rangle$<br>$0.7071 \|H\rangle + (-0.2185 - 0.6725i) \|V\rangle$ |
| 2 | $0.8145 \|H\rangle + (-0.3674 + 0.4490i) \|V\rangle$<br>$0.5802 \|H\rangle + (0.5158 - 0.6303i) \|V\rangle$ | $0.7405 \|H\rangle + (0.1421 + 0.6569i) \|V\rangle$<br>$0.6721 \|H\rangle + (-0.1566 - 0.7237i) \|V\rangle$ | $0.7071 \|H\rangle + (-0.2185 + 0.6725i) \|V\rangle$<br>$0.7071 \|H\rangle + (0.2185 - 0.6725i) \|V\rangle$ |
| 3 | $0.7849 \|H\rangle + (-0.5693 + 0.2446i) \|V\rangle$<br>$0.6196 \|H\rangle + (0.7212 - 0.3098i) \|V\rangle$ | $0.7017 \|H\rangle + (0.5001 + 0.5075i) \|V\rangle$<br>$0.7125 \|H\rangle + (-0.4925 - 0.4998i) \|V\rangle$ | $0.7071 \|H\rangle + (-0.5721 + 0.4156i) \|V\rangle$<br>$0.7071 \|H\rangle + (0.5721 - 0.4156i) \|V\rangle$ |
| 4 | $0.8833 \|H\rangle + (0.3454 + 0.3170i) \|V\rangle$<br>$0.4688 \|H\rangle + (-0.6508 - 0.5972i) \|V\rangle$ | $0.3929 \|H\rangle + (0.9187 - 0.0397i) \|V\rangle$<br>$0.9196 \|H\rangle + (-0.3926 + 0.0169i) \|V\rangle$ | $0.8413 \|H\rangle + (0.4808 + 0.2470i) \|V\rangle$<br>$0.5405 \|H\rangle + (-0.7483 - 0.3845i) \|V\rangle$ |
| 5 | $0.9317 \|H\rangle + (0.0518 + 0.3594i) \|V\rangle$<br>$0.3632 \|H\rangle + (-0.1328 - 0.9222i) \|V\rangle$ | $0.4554 \|H\rangle + (0.7012 - 0.5486i) \|V\rangle$<br>$0.8903 \|H\rangle + (-0.3587 + 0.2806i) \|V\rangle$ | $0.9145 \|H\rangle + (0.1690 + 0.3677i) \|V\rangle$<br>$0.4047 \|H\rangle + (-0.3818 - 0.8309i) \|V\rangle$ |
| 6 | $0.9133 \|H\rangle + (-0.2888 + 0.2872i) \|V\rangle$<br>$0.4073 \|H\rangle + (0.6476 - 0.6440i) \|V\rangle$ | $0.6653 \|H\rangle + (-0.0088 + 0.7466i) \|V\rangle$<br>$0.7466 \|H\rangle + (0.0078 - 0.6652i) \|V\rangle$ | $0.9145 \|H\rangle + (-0.1690 + 0.3677i) \|V\rangle$<br>$0.4047 \|H\rangle + (0.3818 - 0.8309i) \|V\rangle$ |
| 7 | $0.8567 \|H\rangle + (-0.4840 + 0.1786i) \|V\rangle$<br>$0.5159 \|H\rangle + (0.8037 - 0.2965i) \|V\rangle$ | $0.5010 \|H\rangle + (0.4064 + 0.7641i) \|V\rangle$<br>$0.8655 \|H\rangle + (-0.2353 - 0.4423i) \|V\rangle$ | $0.8413 \|H\rangle + (-0.4808 + 0.2470i) \|V\rangle$<br>$0.5405 \|H\rangle + (0.7483 - 0.3845i) \|V\rangle$ |
| 8 | $0.9141 \|H\rangle + (0.3849 + 0.1276i) \|V\rangle$<br>$0.4055 \|H\rangle + (-0.8677 - 0.2876i) \|V\rangle$ | $0.9570 \|H\rangle + (0.1423 - 0.2528i) \|V\rangle$<br>$0.2901 \|H\rangle + (-0.4693 + 0.8340i) \|V\rangle$ | $0.8910 \|H\rangle + (0.4540 - 0.0000i) \|V\rangle$<br>$0.4540 \|H\rangle + (-0.8910 + 0.0000i) \|V\rangle$ |
| 9 | $0.9939 \|H\rangle + (0.0616 + 0.0912i) \|V\rangle$<br>$0.1100 \|H\rangle + (-0.5562 - 0.8237i) \|V\rangle$ | $0.8360 \|H\rangle + (0.2026 - 0.5099i) \|V\rangle$<br>$0.5487 \|H\rangle + (-0.3087 + 0.7770i) \|V\rangle$ | $0.9877 \|H\rangle + (0.1564 - 0.0000i) \|V\rangle$<br>$0.1564 \|H\rangle + (-0.9877 + 0.0000i) \|V\rangle$ |
| 10 | $0.9571 \|H\rangle + (-0.2858 + 0.0479i) \|V\rangle$<br>$0.2898 \|H\rangle + (0.9439 - 0.1581i) \|V\rangle$ | $0.9219 \|H\rangle + (-0.0351 + 0.3858i) \|V\rangle$<br>$0.3874 \|H\rangle + (0.0835 - 0.9181i) \|V\rangle$ | $0.9877 \|H\rangle + (-0.1564 + 0.0000i) \|V\rangle$<br>$0.1564 \|H\rangle + (0.9877 - 0.0000i) \|V\rangle$ |
| 11 | $0.8884 \|H\rangle + (-0.4566 + 0.0474i) \|V\rangle$<br>$0.4590 \|H\rangle + (0.8837 - 0.0917i) \|V\rangle$ | $0.9814 \|H\rangle + (0.0595 + 0.1827i) \|V\rangle$<br>$0.1922 \|H\rangle + (-0.3039 - 0.9331i) \|V\rangle$ | $0.8910 \|H\rangle + (-0.4540 + 0.0000i) \|V\rangle$<br>$0.4540 \|H\rangle + (0.8910 - 0.0000i) \|V\rangle$ |
| 12 | $0.8805 \|H\rangle + (0.4627 - 0.1031i) \|V\rangle$<br>$0.4740 \|H\rangle + (-0.8594 + 0.1916i) \|V\rangle$ | $0.8777 \|H\rangle + (0.4632 - 0.1230i) \|V\rangle$<br>$0.4793 \|H\rangle + (-0.8483 + 0.2252i) \|V\rangle$ | $0.8413 \|H\rangle + (0.4808 - 0.2470i) \|V\rangle$<br>$0.5405 \|H\rangle + (-0.7483 + 0.3845i) \|V\rangle$ |
| 13 | $0.9510 \|H\rangle + (0.0521 - 0.3047i) \|V\rangle$<br>$0.3091 \|H\rangle + (-0.1604 + 0.9374i) \|V\rangle$ | $0.7780 \|H\rangle + (0.4330 - 0.4552i) \|V\rangle$<br>$0.6283 \|H\rangle + (-0.5362 + 0.5637i) \|V\rangle$ | $0.9145 \|H\rangle + (0.1690 - 0.3677i) \|V\rangle$<br>$0.4047 \|H\rangle + (-0.3818 + 0.8309i) \|V\rangle$ |
| 14 | $0.9113 \|H\rangle + (-0.3195 - 0.2597i) \|V\rangle$<br>$0.4117 \|H\rangle + (0.7072 + 0.5747i) \|V\rangle$ | $0.8563 \|H\rangle + (0.1095 + 0.5047i) \|V\rangle$<br>$0.5165 \|H\rangle + (-0.1816 - 0.8368i) \|V\rangle$ | $0.9145 \|H\rangle + (-0.1690 - 0.3677i) \|V\rangle$<br>$0.4047 \|H\rangle + (0.3818 + 0.8309i) \|V\rangle$ |
| 15 | $0.8706 \|H\rangle + (-0.4783 - 0.1154i) \|V\rangle$<br>$0.4920 \|H\rangle + (0.8463 + 0.2042i) \|V\rangle$ | $0.9238 \|H\rangle + (0.2814 + 0.2598i) \|V\rangle$<br>$0.3830 \|H\rangle + (-0.6787 - 0.6267i) \|V\rangle$ | $0.8413 \|H\rangle + (-0.4808 - 0.2470i) \|V\rangle$<br>$0.5405 \|H\rangle + (0.7483 + 0.3845i) \|V\rangle$ |

TABLE V. Basis states used in the experiment by the prover and verifier, and the ideal target basis. Each cell shows the two states in the basis. We note that the prover basis appears rotated relative to the verifier basis to compensate for rotations in the quantum channel.

### D. Time Synchronization

We develop a custom time synchronization solution which combines low-latency communication and control that does not require access to additional dedicated hardware solutions such as White Rabbit [17]. Our approach starts with the Prover transmitting a ready message to each verifier indicating that they are ready to begin the next round. Upon receiving this message, the verifiers wait a custom time delay $\Delta t_x$ and $\Delta t_y$ and then transmit the challenge commands $(x, y)$ to the Prover. Upon receiving the challenge commands, the prover immediately determines the basis and performs the measurement.

### E. Positional Uncertainty

The final positional uncertainty of the prover depends on the round trip time from transmitting basis commands to the prover until each verifier receives a valid measurement response. This round trip time depends on the communication latencies and measurement latencies of the system. Specifically, each verifier can bound the prover's location within a radius from each verifier's location. Such radius is defined as $r_{V_0} = c \cdot (t_{z_0} - t_x)/2$, where $t_x$ is the time Verifier 0 transmits basis commands and $t_{z_0}$ is the time receiving an answer from the prover.

We measure the total round-trip delay time for $V_0$ to be 387.4 ns, and $V_1$ to be 431.3 ns, which corresponds to a positional uncertainty of 58.1 m and 64.6 m, respectively. A latency budget for both $V_0$ and $V_1$ is presented in Tables VI and VII, respectively. We measure the RF cable delays by transmitting a signal down to the prover and measuring the response. The Region of Interest (ROI) is a user setting, which we select as 10 ns. The measured FPGA processing delay includes receiving the 4-bit value $x$ and the 4-bit value $y$, clock signals $X_{CLK}$ and $Y_{CLK}$,

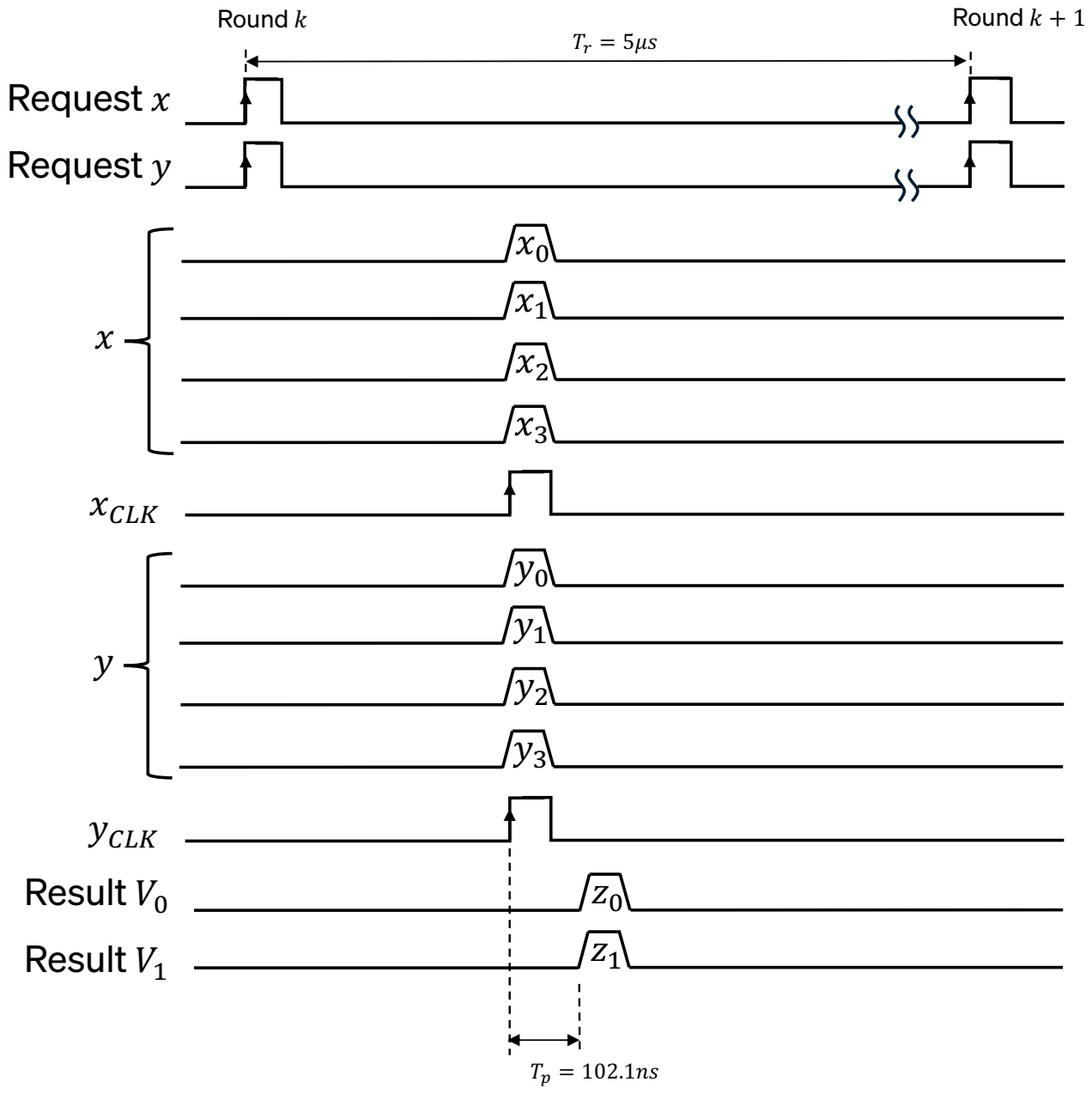


FIG. 9. Prover Timing Diagram. The prover begins each round by transmitting a request for challenge values $x$ and $y$ from Verifier 0, and 1, respectively. The verifiers transmit a 4-bit message for $x$ and a 4-bit message for $y$ along with data clock signals which are timed to arrive simultaneous at the prover's location. Upon, receiving $x$ and $y$ the prover uses a local i.i.d. table to determine the requested basis and perform the measurement. The measurement result is transmitted back to each verifier. The duration of time for the prover to perform the measurement is approximately 102.1 ns.

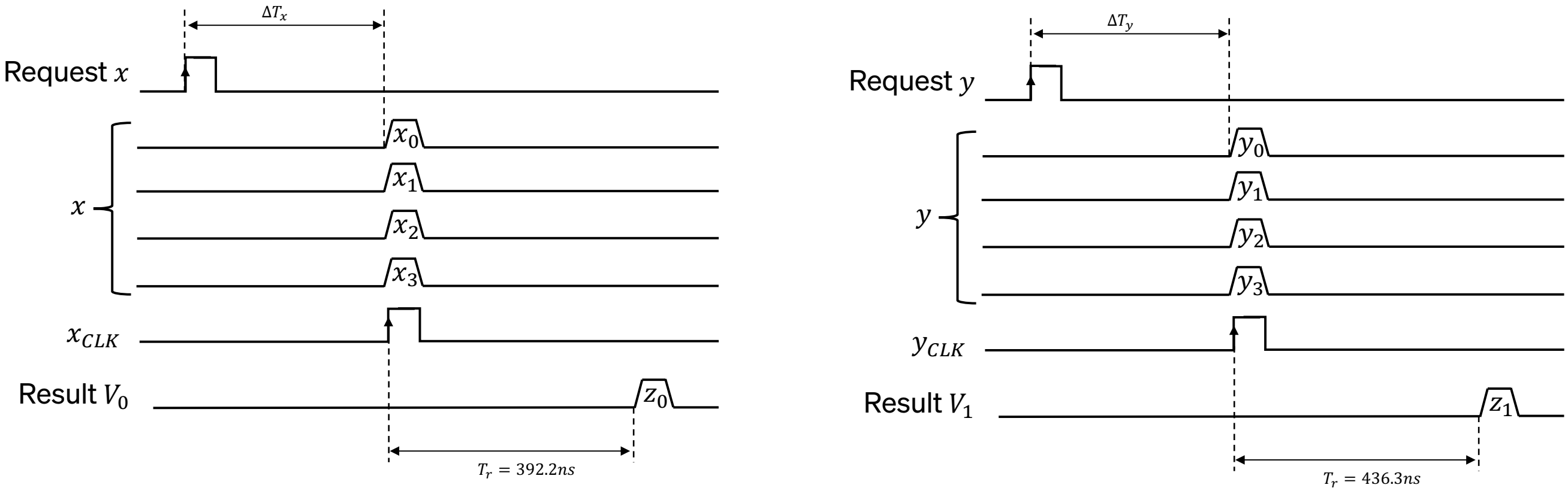


FIG. 10. Verifier Timing Diagram. After receiving the request for challenge messages $x, y$ each verifier waits a predetermined time $\Delta T_x$ and $\Delta T_y$ before transmitting challenge messages $x, y$ to the prover. The round-trip time between transmitting the challenge commands and receiving the measurement result is used to determine the positional uncertainty.

looking up the appropriate basis setting from the lookup table and transmitting the correct Pockels cells control sequence. The FPGA uses 1.8V logic; however, the Pockels cells requires 3.3V logic, which is implemented using external level shifters. The propagation delay time of the level shifters is measured to be 2.9ns. The Pockels cells driver delay and Pockels cells rise time are determined based on datasheet values. The SNSPD delay is measured to be 20 ns. Electrical patch cables interface with the high-propagation speed RF cables and the FPGA as well as the SNSPD to the high-propagation RF cables. Optical patch cables connect the free-space setup to the SNSPD, and the propagation delays represent the free-space photon transmission on the optical table as it passes through the Pockels cells and optical setup. We estimate the total measurement latency for $V_0$ ($V_1$) to be approximately 390 ns (432 ns), which corresponds to an estimated positional uncertainty of 58.4 m (64.8 m). This estimate is very close to our measured round-trip delay for $V_0$ ($V_1$) of 387.4 ns (431.3 ns), which corresponds to a measured positional uncertainty

of 58.1 m (64.6 m).

Furthermore, we have verified the space-like separation between detection event at $V_0$ and the basis revelation. The quantum subsystem $V$ is received by $V_0$ 137 ns after $y$ is sent by $V_1$, but that does not provide sufficient time for $y$ to traverse the 66.5 m straight-line distance between $V_0$ and $V_1$, i.e. $V$ is received before challenge $y$ can arrive at $V_0$.

TABLE VI. LT-QPV Measurement Latency Budget (Verifier 0)

| Node | Component | $V_0$ Latency (ns) | $V_0$ Positional Uncertainty (m) |
|---|---|---|---|
| Verifier 0 | RF Cable Delays | 277.5 | 41.6 |
| Verifier 0 | Region of Interest | 10 | 1.5 |
| Prover | FPGA Processing Delay | 11.4 | 1.7 |
| Prover | Level Shifter Delay | 2.9 | 0.4 |
| Prover | Pockels Cell Driver Delay$^\dagger$ | 30 | 4.5 |
| Prover | Pockels Cell Rise Time$^\dagger$ | 1.5 | 0.2 |
| Prover | SNSPD Delay | 20 | 3.0 |
| Prover | Misc. Electrical Cable Delays | 22.2 | 3.3 |
| Prover | Misc. Optical Patch Cable Delays | 10.1 | 1.5 |
| Prover | Misc. Optical Propagation Delays | 4 | 0.6 |
| Estimated Total Delay ($V_0$) | | 389.6 | 58.4 |
| Measured Total Delay ($V_0$) | | 387.4 | 58.1 |

$^\dagger$Determined by datasheet.

TABLE VII. LT-QPV Measurement Latency Budget (Verifier 1)

| Node | Component | $V_1$ Latency (ns) | $V_1$ Positional Uncertainty (m) |
|---|---|---|---|
| Verifier 1 | RF Cable Delays | 320.1 | 48.0 |
| Verifier 1 | Region of Interest | 10 | 1.5 |
| Prover | FPGA Processing Delay | 11.4 | 1.7 |
| Prover | Level Shifter Delay | 2.9 | 0.4 |
| Prover | Pockels Cell Driver Delay$^\dagger$ | 30 | 4.5 |
| Prover | Pockels Cell Rise Time$^\dagger$ | 1.5 | 0.2 |
| Prover | SNSPD Delay | 20 | 3.0 |
| Prover | Misc. Electrical Cable Delays | 22.2 | 3.3 |
| Prover | Misc. Optical Patch Cable Delays | 10.1 | 1.5 |
| Prover | Misc. Optical Propagation Delays | 4 | 0.6 |
| Estimated Total Delay ($V_1$) | | 432.2 | 64.8 |
| Measured Total Delay ($V_1$) | | 431.3 | 64.6 |

$^\dagger$Determined by datasheet.

### F. Look-up Table

We have provided the full iid look-up table in Table VIII. We note that while the look-up table may show non-uniform basis selection, it is meant as an instantiation of a random oracle, whose output is indistinguishable from a random function (i.e. uniform basis choice) in the QROM model. In practice, any deviation from a random oracle (e.g. non-uniform output) would have been accounted for by the QROM analysis.

## VIII. PROTOCOL COMPARISON

### A. Comparison with $QPV_m^\eta$ protocol with commitment at metropolitan distances

Of the previous implementations and loss tolerance proposals, the implementation by Fan Yuan et. al. [20], of the $QPV_m^\eta$ protocol in Ref. [18] (a loss-tolerant version of the $f$-measure protocol [21] with slow quantum information), and whose channel loss tolerance can be provided through a commitment scheme [22] is a promising candidate for metropolitan area deployment. However, there remains two challenge in the near-term to such an implementation using only commercially available hardware. The first is that the current performance of the protocol scales poorly

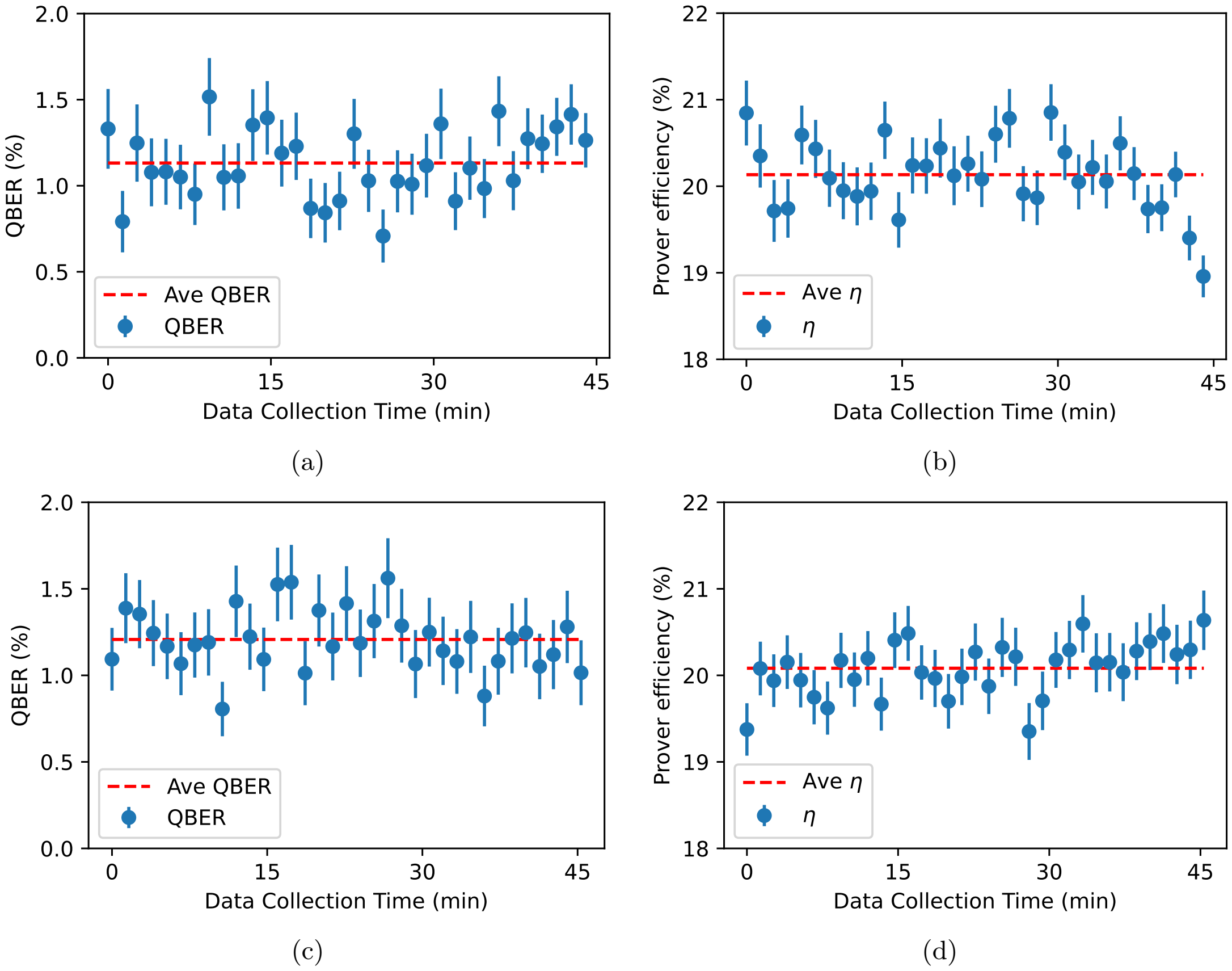


FIG. 11. System Performance: (a) QBER (Calibration), (b) Prover's local Efficiency (Calibration), (c) QBER (QPV Data Run), (d) Prover's local Efficiency (QPV Data Run).

TABLE VIII. Prover Lookup i.i.d. Table

| | Y-Value | | | | | | | | | | | | | | | |
|---|---|---|---|---|---|---|---|---|---|---|---|---|---|---|---|---|
| X-Value | 0 | 1 | 2 | 3 | 4 | 5 | 6 | 7 | 8 | 9 | 10 | 11 | 12 | 13 | 14 | 15 |
| 0 | 6 | 3 | 12 | 14 | 10 | 7 | 12 | 4 | 6 | 9 | 2 | 6 | 10 | 10 | 7 | 4 |
| 1 | 3 | 7 | 7 | 2 | 5 | 4 | 1 | 7 | 11 | 13 | 5 | 1 | 15 | 11 | 4 | 0 |
| 2 | 11 | 9 | 5 | 12 | 11 | 8 | 0 | 10 | 10 | 14 | 9 | 15 | 11 | 11 | 15 | 14 |
| 3 | 13 | 13 | 14 | 13 | 2 | 11 | 6 | 3 | 15 | 8 | 2 | 4 | 2 | 6 | 4 | 8 |
| 4 | 6 | 1 | 3 | 8 | 11 | 13 | 1 | 9 | 8 | 9 | 4 | 1 | 3 | 11 | 14 | 11 |
| 5 | 6 | 11 | 12 | 7 | 14 | 2 | 13 | 0 | 3 | 1 | 7 | 3 | 1 | 13 | 5 | 5 |
| 6 | 9 | 3 | 5 | 12 | 14 | 1 | 9 | 11 | 1 | 9 | 13 | 3 | 13 | 14 | 15 | 14 |
| 7 | 7 | 13 | 6 | 11 | 8 | 13 | 7 | 4 | 15 | 12 | 1 | 14 | 4 | 7 | 9 | 8 |
| 8 | 11 | 11 | 11 | 12 | 8 | 12 | 14 | 12 | 0 | 8 | 6 | 8 | 7 | 0 | 11 | 7 |
| 9 | 7 | 14 | 10 | 2 | 0 | 7 | 2 | 2 | 0 | 10 | 4 | 9 | 6 | 9 | 8 | 11 |
| 10 | 6 | 8 | 7 | 11 | 1 | 0 | 15 | 6 | 6 | 13 | 7 | 4 | 2 | 11 | 7 | 5 |
| 11 | 10 | 2 | 0 | 2 | 4 | 14 | 13 | 2 | 0 | 4 | 9 | 6 | 14 | 13 | 6 | 10 |
| 12 | 8 | 14 | 14 | 9 | 9 | 11 | 12 | 2 | 14 | 15 | 6 | 0 | 3 | 12 | 15 | 3 |
| 13 | 13 | 4 | 6 | 6 | 12 | 14 | 10 | 12 | 3 | 12 | 15 | 6 | 10 | 2 | 5 | 11 |
| 14 | 1 | 9 | 12 | 13 | 8 | 4 | 5 | 11 | 11 | 11 | 11 | 3 | 13 | 13 | 10 | 15 |
| 15 | 9 | 6 | 11 | 8 | 6 | 13 | 0 | 0 | 8 | 10 | 13 | 8 | 3 | 12 | 8 | 2 |

with increased number of basis. Fig. 12 shows that while loss tolerance can improve by increasing the number of bases, at realistic QBER of 0.27% [20], the protocol is insecure for transmission below 50%. We stress that this poor scaling is a consequence of the current security analysis and not an inherent property of the protocol itself. Recent security analysis results shows that the error tolerance of $QPV_m^\eta$ can be improved significantly at zero loss [23], but it is not straightforwardly clear if this improvement extends to the lossy scenario. Analyzing the protocol using other methods

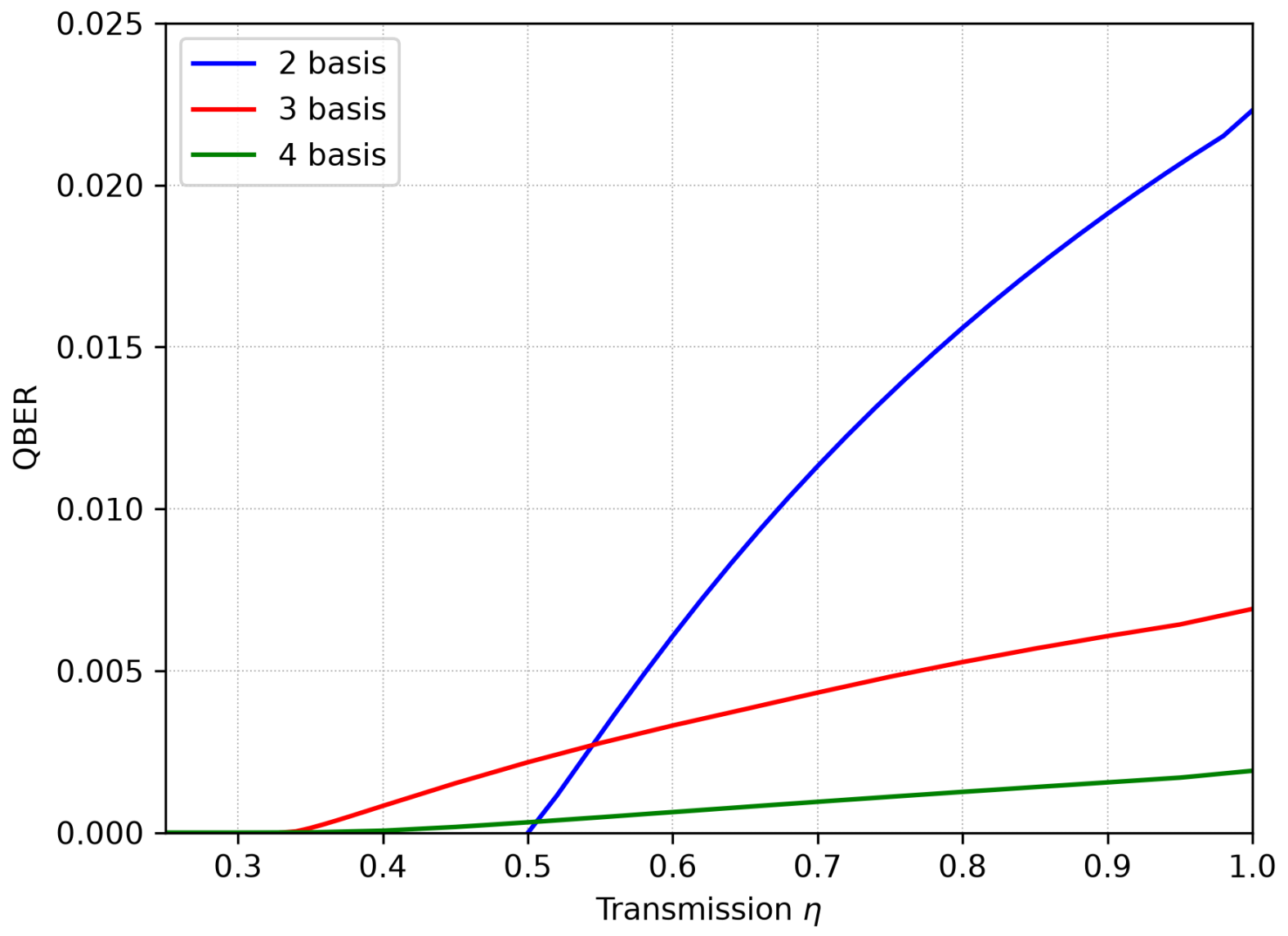


FIG. 12. **Asymptotic Performance of $QPV_m^\eta$ [18].** Performance of $QPV_m^\eta$ protocol [18] in the asymptotic regime, using the tighter bounds presented in Ref. [19] at NPA level 1+AB (2/3 basis) and level 1 (4 basis). While the protocol can achieve loss tolerance of $\eta \sim 1/m$ in principle at zero QBER, at the current QBER in Fan Yuan et. al.'s implementation [20] of 0.27%, which uses coherent states instead of the ideal qubit state, the transmission $\eta$ required still exceeds 50%.

such as under QROM [24] may also significantly increase the performance. The second challenge is the implementation of the commitment scheme proposed in Ref. [22], where a quantum non-demolition measurement is performed by the prover to commit to whether a qubit has been received before the basis is received. In practice, this is implemented via a quantum teleportation or entanglement swapping scheme, where the prover uses an entanglement source and Bell state measurement to teleport the received quantum state, and a Bell state measurement success would indicate the receipt of a qubit. Such an implementation presents four potential issues. Firstly, having a Bell state measurement can introduce a larger QBER that may be beyond the current performance requirements of $QPV_m^\eta$ or require high transmission to remain secure. Secondly, the heralding efficiency of a commercially available entanglement source and detection efficiency at the prover side may result in a low prover-side efficiency (e.g. at 20% if our system is used) that is below the beyond the current performance requirements of $QPV_m^\eta$. Thirdly, the use of coherent states may not be directly compatible with the proposed commitment scheme to achieve full loss-tolerance due to photon-number splitting attack. By post-selecting and committing only on multi-photon events, $QPV_m^\eta$ is no longer secure since the adversaries have access to more than one copy of the same state. Limits on the channel loss or other techniques such as decoy-state [25] may have to be introduced to combine the commitment scheme with a coherent state implementation. Lastly, if classical communication does not occur at the speed of light (e.g. if coaxial cables are used), the scheme loses its fully channel loss-tolerance property since the qubit from the entanglement source at the prover have to be stored for a duration that scales linearly with the prover-verifier separation – a similar challenge that our protocol faces as well. While an implementation of loss-tolerant version of this proposal is challenging with current hardware, future advances in theoretical analysis and better commercial hardware makes this a promising candidate for loss-tolerant QPV.

In comparison, while our proposal may not provide the full loss-tolerance claim without the classical and quantum communication travel at the speed of light, it provides a sufficient level of channel loss tolerance relative to $QPV_m^\eta$ [18] and Fan Yuan et. al.'s demonstration [20], while at the same time being simpler to implement than the commitment scheme. While the use of a QROM model in our security analysis adds computational assumptions to the protocol, it provides better performance requirements, allowing the protocol to be scaled to 16 basis while maintaining reasonable requirements on error tolerance. This eases the implementation with commercial hardware. To provide further loss-tolerance, our scheme requires the prover prepare an entangled state, and the verifier's measurement serves as a commitment indicating that a qubit has been received through the channel, serving a similar role as the commitment scheme in Ref. [22] where the commitment is performed on the prover side. Compared to Ref. [22], our proposal removes the requirement of a Bell-state measurement, which can raise QBER significantly and is more challenging to

TABLE IX. Comparison of our proposal with recent experimental implementations. For Ref. [4, 20], we present the loss tolerance assuming zero QBER / perfect fidelity of state preparation.

| | **Kavuri et al. [4]** | **Fan-Yuan et al. [20]** | **Kanneworff et al. [26]** | **Our Proposal** |
|---|---|---|---|---|
| Adversary Model | Small amount of shared entanglement | Linear amount of shared entanglement | No shared entanglement | Polynomial amount of resource (including entanglement) |
| Trust (Quantum Device) | Not required | State preparation | State preparation | Measurement |
| Loss Tolerance (Asymptotic) | Transmission $> 66\%$ | Transmission $> 50\%$ | Fully loss-tolerant | Prover-side transmission $> 16\%$ |
| Verifier Resource | Pockels Cell + SNSPD | Coherent state source + Modulators | Quantum Dot Source + Modulator | Pockels Cell + SNSPD |
| Challenge Length | 1 bit each (2 basis) | 20 bits each (2 basis) | NA (quantum challenge) | 4 bits each (16 basis) |
| Detection Efficiency | $> 81\%$ | $\sim 70\%$ | 7.2–12.4% | $\sim 20\%$ (prover-side) |
| Classical Channel | Coaxial Cable ($\sim 0.85$–$0.87\,c$) | Hollow-core fiber (near $c$) | Not implemented | Coaxial Cable ($\sim 0.88\,c$) |
| Prover–verifier separation | ~100 m | 0.98 km | 200 m | 35–40 m |
| Localization Metric | Area uncertainty $\sim 1/4\times$ of classical, Position uncertainty $\sim$ 116–157 m | Position uncertainty $\sim$ 75 m | NA | Position uncertainty $\sim$ 60 m |
| Soundness Parameter $\varepsilon_{\mathrm{sou}}$ | $2^{-64}$ | $5 \times 10^{-10}$ | NA (insecure) | $10^{-2}$ (iid adversary) |
| Acquisition Time | 2 min | 5 s | NA (insecure) | 22 min (iid adversary) |

implement and achieve security especially when prover-side efficiency is not high. As such, our scheme serves as a good in-between, trading off the level of loss-tolerance guarantee of Ref. [22] for the ease of implementation.

### B. Comparison with other implementations

Table IX compares the experimental implementation and theoretical guarantees of our protocol and Ref. [4, 20, 26].

---



\end{document}